\documentclass[11pt]{article}

\usepackage{isabelle,isabellesym}
\usepackage{amsmath}
\usepackage{amssymb}
\usepackage{amsthm}
\usepackage{booktabs}
\usepackage{authblk}
\usepackage{graphicx}
\usepackage{proof}
\usepackage[top=1in,bottom=1in,left=1.5in,right=1.5in]{geometry}

\usepackage{microtype}

\usepackage[ruled]{algorithm2e} 

\SetAlFnt{\small}
\SetAlCapFnt{\small}
\SetAlCapNameFnt{\small}
\SetAlCapHSkip{0pt}
\IncMargin{-\parindent}

\newcommand*{\prob}[1]{\boldsymbol{\mathcal{P}}_{#1}}
\newcommand*{\W}[0]{\;\boldsymbol{\mathcal{W}}\;}

\theoremstyle{definition}
\newtheorem{theorem}{Theorem}[section]

\newtheorem{lemma}[theorem]{Lemma}

\newlength{\arrow}
\begin{document}

\title{Runtime Fault Detection \\in Programmed Molecular Systems
	\thanks{This research was supported in part by the National Science Foundation grants \#1247051 and \#1545028.}
}

\author[1]{Samuel J. Ellis}
\author[2]{Titus H. Klinge}
\author[1]{James I. Lathrop}
\author[1]{Jack H. Lutz}
\author[1]{Robyn R. Lutz}
\author[1]{Andrew S. Miner}
\author[1]{Hugh D. Potter}
\affil[1]{Iowa State University, Ames, IA 50011}
\affil[2]{Carleton College, Northfield, MN 55057}

\date{}

\maketitle

\begin{abstract}
Watchdog timers are devices that are commonly used to monitor the health of safety-critical hardware and software systems.  Their primary function is to raise an alarm if the monitored systems fail to emit periodic ``heartbeats'' that signal their well-being.  In this paper we design and verify a {\it molecular watchdog timer} for monitoring the health of programmed molecular nanosystems.  This raises new challenges because our molecular watchdog timer and the system that it monitors both operate in the probabilistic environment of chemical kinetics, where many failures are certain to occur and it is especially hard to detect the absence of a signal. 

Our molecular watchdog timer is the result of an incremental design process that uses goal-oriented requirements engineering, simulation, stochastic analysis, and software verification tools.  We demonstrate the molecular watchdog's functionality by having it monitor a molecular oscillator.  Both the molecular watchdog timer and the oscillator are implemented as chemical reaction networks, which are the current programming language of choice for many molecular programming applications.

\end{abstract}

\section{Introduction}

Molecular programming uses the information processing inherent in chemistry, especially the chemistry of DNA and other biomolecules, to engineer useful nanoscale systems. Molecular programming applications including diagnostic biosensors, medical therapeutics,  molecular robots and bio-compatible materials already exist in the laboratory and are poised to have a major impact on society. 

This paper proposes the design of verifiable safety mechanisms for molecular programming applications.  This is an important and timely objective because many of the planned future uses of molecular programmed systems are safety-critical, such as biosensors to detect pollutants in water or drug therapeutics to deliver customized medicine only to the locations in the body where disease has been detected \cite{oWooLin05,  jZhaSee11, jDoBaCh12,  jLiLiYa13}.  During the rapid expansion of the field of molecular programming, basic research has been rightly focused on normal operation (systems doing what they should).  However, as the field progresses to development and deployment, it will also need to focus on avoiding hazards (systems not doing what they shouldn't).  Designing and formally verifying safety mechanisms for programmable control of molecular systems now will help enable this future focus on safety \cite{oElli17}.

In many safety-critical systems, the failure of a monitored system can be dangerous if it goes undetected.  A standard remedy is to introduce a {\it software watchdog timer}, a safety mechanism whose responsibility is to monitor for the occurrence of the failure and to raise an alarm that can trigger recovery action if a failure occurs \cite{oKnig12, oLeve95}.  For example, on the Voyager spacecraft, a heartbeat was sent every two seconds from the attitude control computer to the command computer responsible for monitoring its health. If the attitude control computer's self tests failed, then no heartbeat was generated.  ``After about 10 seconds passed with no heartbeats, the command computer would issue a switch-over command to the backup processor'' \cite{oNasa88}.  Correct behavior of a watchdog timer in the context of the system it monitors is essential. For example, the faulty integration of a watchdog timer such that it did not actually monitor the throttle may have resulted in the failure of a major task in the Toyota unintended acceleration accidents \cite{oKoop14}. 

A watchdog timer receives a periodic heartbeat from the system that it is monitoring.  Receipt of the heartbeat resets the watchdog timer, indicating that the monitored system is still alive.  When the heartbeat signal stops, the absence of the heartbeat causes the watchdog timer to time out and ``bark'.'  Outputting this alarm indicates that a fault has occurred in the monitored system that led to its experiencing a service failure  \cite{jALRL04}.   Watchdog timers often serve both to detect the failures of monitored systems and to trigger their recovery, either through corrective interventions or automated recovery actions. 

Guided by the successes of software watchdog timers, we have chosen a {\it Molecular Watchdog Timer} as our first safety mechanism for molecular programming applications.  This Molecular Watchdog Timer should monitor a molecular system at runtime, detect when the heartbeat signal from the monitored system stops, and alarm to trigger its recovery.

The design of a Molecular Watchdog Timer is an ambitious goal.  Monitoring for the absence of an event is especially difficult in the molecular domain.   Detecting the non-occurrence of an expected event is often not yet possible for components that execute outside the laboratory such as {\it in vivo} applications.   For example, some of the most promising planned molecular systems involve drug delivery to tumors.    For such systems, runtime monitoring and detection of faults must take place in the same molecular environment as the programmed system itself.    
Additional challenges to detecting faults at runtime in molecular programmed systems include the facts that the behavior of molecular systems is probabilistic; the components are nanoscale, so runtime observation is non-trivial; there are very many components so scalability is a problem; and the components are fault-prone so any fault that can occur probably will occur in a significant number of components.   

The main contribution of this paper is the requirements specification, design, and verification of a Molecular Watchdog Timer (MWT) for molecular programming applications.  This contribution includes the following features:
\begin{enumerate}
\item We develop a goal-based model and formally verify the requirements for an MWT using the Isabelle proof assistant \cite{Nipkow2002}.
\item Our MWT is designed as a {\it chemical reaction network}, a well-understood mathematical model \cite{oAndKur11, oAndKur15} that is suitable for automatic compilation into DNA implementations \cite{jSCWB08, cSoSeWi09}.
\item Our MWT detects heartbeat failures at runtime and alarms accordingly.
\item Our MWT's timing functions are carried out by {\it stochastic delay ladders}, introduced here, that are provably reliable, both for detecting failures and for avoiding false alarms.
\item In addition to alarming, our MWT can trigger the recovery of a monitored system.
\item Our MWT is automatically reusable rather than having to be discarded after a single failure.
\item Our MWT is tested with a specific monitored system, a molecular oscillator widely used as a benchmark, modified to produce a heartbeat.  Two very different types of oscillator failure can interrupt this heartbeat.
\item Our MWT is embedded with oscillators at two widely separated ranges of size, with end-to-end behavior verified and validated at these scales using model checking \cite{oBaiKat08, jKwiTha14} and simulation, respectively.
\item These verifications demonstrate that our MWT performs correctly with the oscillators, enabling them to recover correctly from both types of failure.
\end{enumerate}
Earlier versions of our MWT were reported in \cite{cEHKLLL14, oElli14}.  The version reported here has a different architecture and improved functionality. Regarding the above list of features, the earlier versions had goal-based models as in (1), but these were hand-verified.  The earlier versions had features (2) and (3) and rudimentary, inexplicit versions of feature (4).   The present paper's use of the Isabelle proof assistant in feature (1), its systematic treatment of the delay ladders in feature (4), and all aspects of features (5) through (9) are new work.   This paper describes a MWT re-designed to be reusable, embeddable with the system it monitors, and capable of triggering a monitored system's recovery at runtime. 

A broader contribution of the paper is to demonstrate the new use of software engineering techniques and tools to create and verify the requirements and design for a programmable safety mechanism, the Molecular Watchdog Timer,  that will be needed for future molecular systems such as biocompatible drug delivery nanodevices.   We have sought to make the development approach that we use general enough to guide future molecular design work on additional safety mechanisms.  We claim that software engineering helps achieve molecular programmed systems that are safe for use in a dynamic and only partially understood physical environment and show in our results the benefits of a software engineering approach to creating new molecular systems.  

The rest of the paper is organized as follows.  Section~\ref{sec:requirements} describes the goal-oriented requirements analysis and machine-checked refinement proofs for the Molecular Watchdog Timer.   Section~\ref{designSect} presents the chemical reaction network design models to achieve these capabilities. Section~\ref{sec:verification} shows how simulation, probabilistic model checking and mathematical proofs are used to validate and verify the design. 
Section~\ref{relatedWork} describes related work, and Section~\ref{sec:conclusion} gives concluding remarks.

\section{Requirements}\label{sec:requirements}
In this section we describe a requirements engineering process for  programmed molecular systems and use it to develop the requirements for  a runtime fault detection device called a Molecular Watchdog Timer (MWT). We first describe informally the high-level requirements for the new system to be built.  We then describe the iterative process by which the requirements were formally specified, analyzed, corrected, refined, and verified as our understanding of what was needed improved. Lastly, we discuss how the interplay between requirements and design contributed to the incremental improvement of the requirements. This ongoing process made the requirements significantly more complete, accurate and realistic (i.e., feasible to implement in DNA).

Despite the difficulty of ``getting the requirements right'' for a new molecular device, our use of requirements engineering techniques in molecular programming is novel.  We show how engaging in goal-oriented requirements engineering can benefit the design of programmed molecular systems, especially by finding and solving problems early in development. 

Figure~\ref{contextfigure} shows a context diagram \cite {oLams09} for the MWT.
  The Monitored System sends a regular signal (labeled ``Heartbeat'') to the MWT.   The MWT processes this signal and outputs an Alarm to an external observer or system if the signal has stopped.   
\begin{figure}[h]
    \begin{center}
    \includegraphics[width=3.5in]{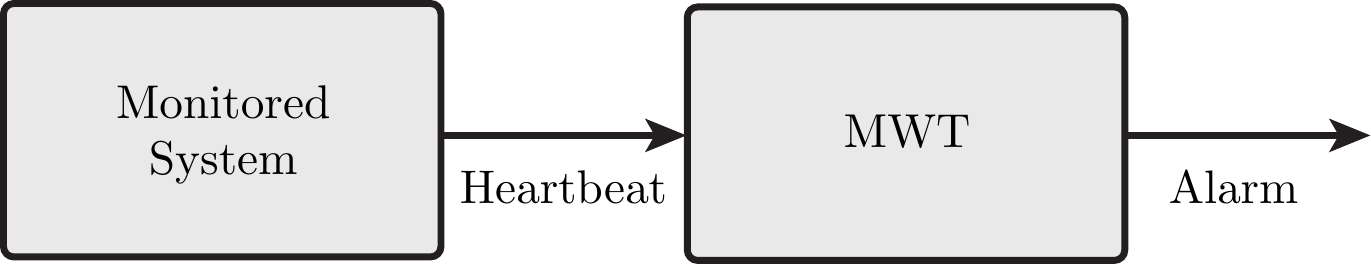}
    \caption{Context Diagram \label{contextfigure}}
    \end{center}
\end{figure}

\subsection{Goal modeling}

We used a goal-oriented requirements engineering approach that we based on van Lamsweerde's KAOS \cite{oLams09} and introduced in \cite{cLLLKHM12,cLLLKMS12} to specify and analyze the MWT system goals.   Goal-oriented techniques support systematic and incremental refinement from informal descriptions to formal specification of properties for simulation and model checking.   

We first performed goal modeling to understand what was needed to achieve a MWT.  A goal model is an AND/OR graph in which the top-level node describes the high-level functional requirement of the system-to-be and the leaf nodes are the subgoals whose collective satisfaction implies satisfaction of the top-level goal.  The goals are specified as goals to ACHIEVE, MAINTAIN, or AVOID various conditions 
in view of the domain properties (such as physical laws governing molecular interactions).   An AND node is satisfied provided that its children nodes are satisfied.  An OR node shows alternative refinements of a node.   Our goal model's nodes are all AND nodes. 

The AND/OR refinement of high-level goals and the graphical presentation of the results is sufficiently intuitive to be useful in our discussions with molecular biologists.    It also allows both top-down and bottom-up development of the hierarchical tree as understanding improves.  The logical underpinnings of KAOS support formal analysis and clean traceability between the textual descriptions of the goals and the formal specification of the goal model. 

Our high-level goal is that at runtime the MWT shall issue an alarm if and only if the monitored system does not provide a heartbeat within a specified time.  Figure~\ref{goalfigure} shows the goal model for the MWT.
\begin{figure}
	\begin{center}
		\includegraphics[width=5.0in]{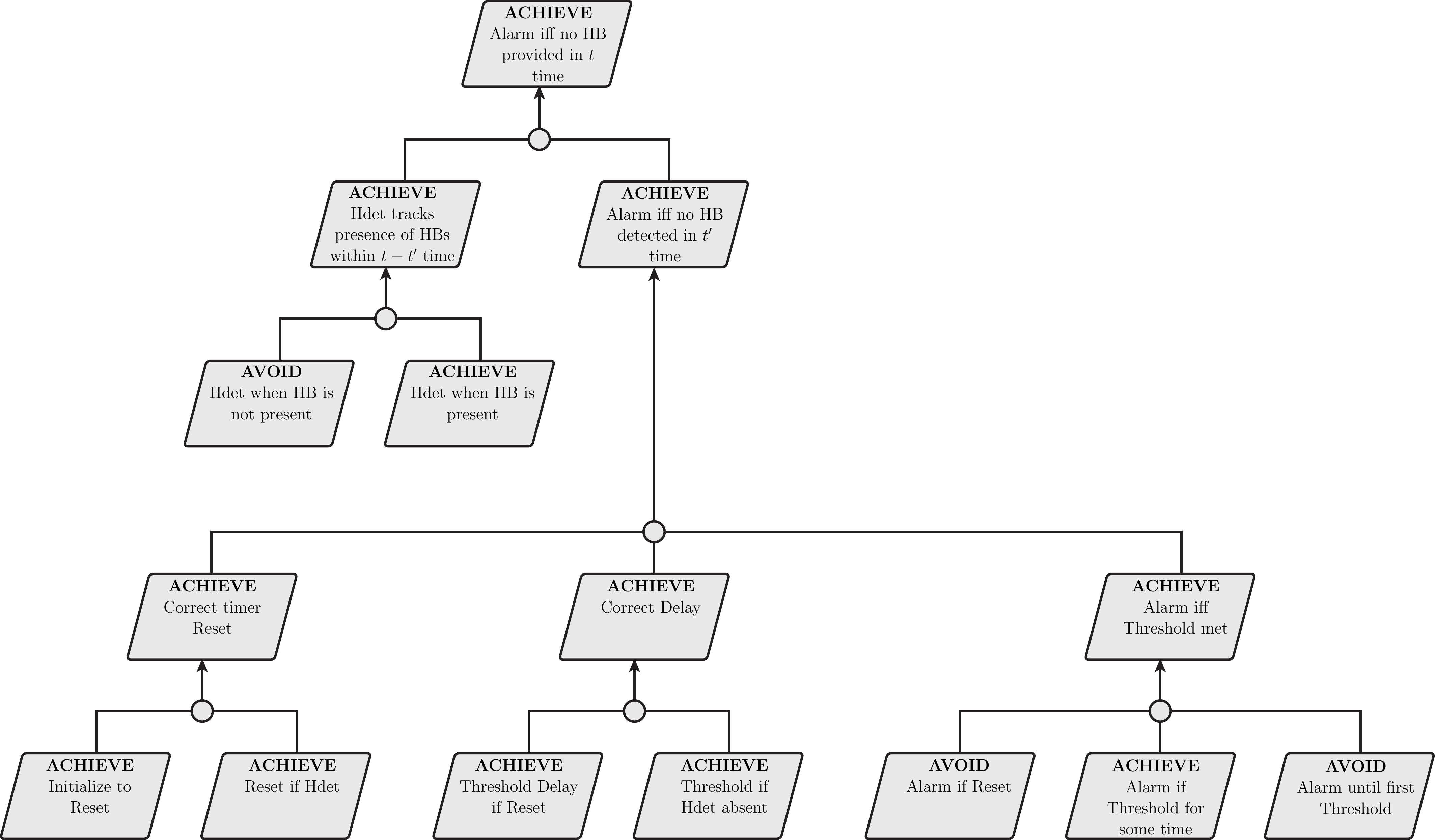}
		\caption{Goal Model \label{goalfigure}}
	\end{center}
\end{figure}
 The top-level goal, {\it Achieve[Alarm iff no Heartbeat provided within t time]} describes the intent of the system-to-be.  That goal is AND-refined into two subgoals, such that it can be satisfied if both of the two second-level goals are met.  The subgoals are further refined until each leaf goal can be assigned to an agent (discussed in Sect. 3). An agent is a system component with responsibility for satisfying the goal(s) assigned to it \cite{oLams09}.   

The MWT's client can use the alarm signal output by the MWT either as an externally observable 
{\it alarm} that notifies the client when the heartbeat stops or as a {\it recovery trigger} prompting the initiation of some recovery action when the heartbeat stops.  The first usage scenario is intended for scientific observations where the alarm signal might, for example, be a fluorescent molecule visible to a human technician. The second scenario is useful when the MWT is monitoring an adjacent molecular system that is capable of autonomous recovery and is a focus of this work. 
\subsection{Environmental assumptions}

Environmental assumptions are statements about the system's operational context that are accepted as true by the developers \cite {cTLNYMN15}.  The proposed MWT will be made of DNA strands and will operate, both literally and figuratively, in a fluid, molecular environment.   
For satisfaction of the subgoals to imply that the parent goal is satisfied, 
we must make certain assumptions about the operational environment.   For example, we assume that the chemical solution in which the MWT operates is well mixed. 
Other environmental assumptions needed to prove the satisfiability of the top-level goal are that the heartbeat species, $H$, is intrinsically ephemeral, meaning that it will not persist and will decay over time; that the number of $H$ molecules in the heartbeat pulse is in a certain range; and that no molecules in the environment other than the heartbeats interact with the MWT.  

If these environmental assumptions are false or cease to be true, the validity of the solution may be at risk.   We will see below how the process of formalizing the leaf goals, proving that the leaf goals implied the top-level goal and verifying the design forced us to revisit and revise several of our original environmental assumptions. 

\subsection{Goal formalization}

For every goal we needed both a natural language description and a formal specification.  We specified the goals in continuous stochastic logic (CSL) \cite{jASSB00} because of its availability in the model checking tool that we use and because it handles the continuous time Markov chains (CTMCs) \cite{oAndKur11, oAndKur15} that form the semantics of our chemical reaction networks (described in Section~\ref{designSect}).   Formally specifying the goals and assumptions enabled us to prove that the satisfaction of the lower-level goals, together with the stated assumptions, implies the satisfaction of the higher-level goals.   

The complete CSL specifications for the goal model appear in the table in the Appendix.  For each node in the goal model, the table there shows (1) the name of the node, (2) the formal specification of the node in CSL, and (3) the agent assigned responsibility for each leaf goal.  The goals are ordered breadth first following Figure~\ref{goalfigure}. As an example, the leaf goal \textit{Achieve[Alarm if Threshold for some time] } has the CSL specification $\prob{\ge 1}\Box\left[Th_H\implies \prob{\ge 1-\eta_4}\Diamond_{\le w_{th}}\left(Alarm\;\lor\;\lnot Th_H\right)\right]$ and is assigned to the Threshold Filter as described in Section 3. 

We proved that the goals, together with the specified environmental assumptions and facts and parameter values in their specified ranges,  satisfy the top-level goal. This proof was carried out manually and iteratively during the development of the goal model.  

\subsection{Obstacle analysis} 

The MWT must work safely and with high reliability to detect and recover from faults in the molecular applications for which it is intended, such as bio-sensing and medical therapeutics.  A major challenge in developing the MWT was determining accurately the requirements for a {\it feasible} system that could operate in the molecular environment.  To do this we needed to analyze whether the stated goals could be realistically satisfied in this stochastic and dynamic setting. 

To investigate feasibility, we used KAOS's obstacle analysis extensively \cite {jLamLet00, oLams09, cCaiLam14}. 
An obstacle to a goal is a precondition for that goal's non-satisfaction \cite{oLams09}.  Given a goal model in which A and B are the AND subgoals of parent goal C, an obstacle for C is  a state of affairs in which A and B are true and C is false.    

Obstacle analysis {\it identifies} ways in which the goals might fail to be satisfied (i.e., obstacles to satisfaction), {\it assesses} the likelihoods and impacts of the obstacles, and investigates how to {\it resolve} them.  We found in previous work on DNA nanopliers that the early analysis of obstacles to satisfying the goals worked well in helping find and remedy missing and unrealistic requirements in the programmed molecular system \cite{cLLLKMS12}. 

Most of the obstacles for the MWT were subtle and found manually while proving that subgoals satisfied their parent goal.   The process of proving the satisfiability of the formal CSL goals repeatedly revealed both additional cases and uncertainties that could be introduced by the stochastic behavior.  Gaps and errors in the goal model often were due to the non-deterministic, asynchronous nature of reactions, to the very large number of molecular agents operating in parallel, and to the physical environment in which the MWT operates.  Model-based mathematical analysis, simulation with the MATLAB extension, SimBiology, and probabilistic model-checking with PRISM \cite{cKwNoPa11}, 
as described below, 
helped us assess the likelihood and impact of candidate obstacles.

Representative examples of obstacles we found in the early version of the MWT were previously reported in \cite{cEHKLLL14}. These include:

{\it Incorrect agent}. We initially assigned a binary counting device introduced in \cite{cJiRiPa11}  to be the clock agent responsible for detecting the absence of a heartbeat.  However, simulation in SimBiology revealed that this device did not satisfy our specification. It was designed to work in a setting in which it is assumed that all reactions are ``fast'' or ``slow'', and that all fast reactions occur before all slow reactions. The stochastic (and more realistic) model on which we work violates this assumption, allowing slow reactions to interfere with the clock's function. Over time the accumulation of such violations leads to failure of the clock.  To resolve this, we assigned responsibility for that goal to a new agent in which the delay is instead achieved by a programmed cascade of interactions. 

{\it Missing property}. In refining a goal into two subgoals, we had to introduce the domain property that it takes a positive amount of time to detect a heartbeat. This is because the detection occurs via the chemical interaction of the heartbeat molecules with the molecular component that detects missing heartbeats. The subgoals did not satisfy the goal without this domain property. This obstacle was resolved by introducing a ``grace period'' before heartbeat detection is required. When we initially failed to propagate the addition of the grace period back up to the parent goal, model checking with PRISM detected the omission, and it was corrected. 

{\it Incorrect initialization}. Model checking revealed a failure mode that can occur just after the MWT begins execution. The initial intent was that operation of the MWT begin at time zero, i.e., when the MWT was ``poured into the test tube''. However, as specified, the MWT could violate this intent by alarming {\it before} the monitored system had a chance to send a heartbeat. To resolve this, we added a new CSL property to the high level goal specifying that the alarm must remain off for a period of time after initialization. This new goal propagated through the goal diagram to create a new leaf-goal, {\it Achieve[Initialize to Reset]}, specifying that the timer must be considered to be in a reset stage at initialization. Together with the leaf goal {\it Achieve[Threshold Delay if Reset]}, this implies that the alarm will not be active upon initialization. We proved manually that the implication holds after the change and confirmed using PRISM that a model with an alarm in the initial state satisfies the goals. 

Obstacle analysis helped identify missing requirements, explore design alternatives and find idealized environmental assumptions that had to be weakened to conform to physical realities. The obstacle analysis process was on-going and a major contributor to de-idealizing the goals for the MWT into requirements that were feasible for implementation in a programmed molecular system operating in a chemical ``soup''.   During the obstacle analysis we experienced extensive back-and-forth interweaving between the requirements and the design.  This iterative, incremental nature of the modeling and analysis effort is typical of complex systems \cite{jWGCMHR13} and often described in terms of ``twin peaks'' \cite{jNuse01}. 

Related work on specifying goals in uncertain environments formalizes the required degree of goal satisfaction, as in \cite{cLetLam04}, or the required probability of goal satisfaction, as in \cite{jCaiLam13}.  However, molecular systems such as the MWT may have more than $10^{10}$ individual components in solution, so failures with any significant probability probably will occur in many individual components.  The system must nevertheless be robust enough to 
operate correctly (to alarm or not to alarm) with probability approaching 1, even in the presence of some component failures. The watchdog timer design must be one in which we have confidence that, if the system it is monitoring fails, the MWT will detect and notify us, and that if the MWT notifies us that the monitored system has failed, then we can trust its accuracy.   

\subsection{Verifying the goal model} 
After the goal model was stabilized, we re-proved its internal correctness (the fact that the leaf goals imply the top-level goal) with the aid of the Isabelle proof assistant \cite{Nipkow2002, Nipkow2014}. In order to maximize the use of Isabelle's automated features and the readability of the result, we used a hybrid of human mathematics and automated theorem proving.  Specifically, we formulated a short list of CSL lemmas that are relatively simple (especially, not burdened with specific parameters of our design) and capture the higher-order logical aspects of our goals that are not readily amenable to automation.   After proving these lemmas succinctly and transparently, we used Isabelle to prove that, given these lemmas, the goal model is internally correct.  

The verification of the goal model is reported completely in the Appendix. 

More specifically, every CSL operator defines a state formula, so we represent CSL operators in Isabelle as paramatrized functions from CTMC states to the set $\{\mathit{true,} \;  \mathit{false}\}$.  For example, we can supply the $\boldsymbol{\mathcal{P}}\Diamond$ operator with probability and time parameters $\alpha$ and $t$ and a component CSL formula $\phi$; the result is the CSL formula $\prob{\ge\alpha}\Diamond_{\le t} \phi$, which can take on the values $\mathit{true}$ or $\mathit{false}$ at different CTMC states.

We define rules for logical connectives in CSL using lambda expressions.  For example, given two CSL formulae $\phi$ and $\psi$, we define conjunction as
\begin{equation*}
\phi \land \psi = \lambda x\ \phi(x) \land \psi(x).
\end{equation*}
That is, $\phi \land \psi$ is a single function that returns true if both $\phi$ and $\psi$ are true for the given argument.  In combination with the operator definitions discussed above, logical connectives like this are suffient to encode arbitrary CSL statements and implications in Isabelle.

To construct a verified proof using this formalism in Isabelle, we provide a sequence of intermediate goals that link our assumptions and other known statements with our final goal.  For each step in this proof, we provide Isabelle with assumptions and lemmas to reference and with proof methods to apply.  Isabelle's powerful Sledgehammer tool \cite{Meng2006, Paulson2007} automated significant parts of this process, making it easier to supply the correct facts and methods and take larger steps.

\section{Design}\label{designSect}   

The goal-oriented requirements refinement and analysis described above assigns responsibility for achieving the MWT's (Molecular Watchdog Timer's) leaf goals to two system agents:  the Absence Detector component and the Threshold Filter component, which outputs an Alarm signal.   Figure~\ref{archfigure} shows the high-level  design with these two components.  The connectors among the components reflect the intended flow of information from detection of individual heartbeats or their absences, to determining whether the incidence of faults exceeds a programmed threshold and issuing an alarm signal when that threshold is exceeded.
\begin{figure}
        \begin{center}
                \includegraphics[width=4.0in]{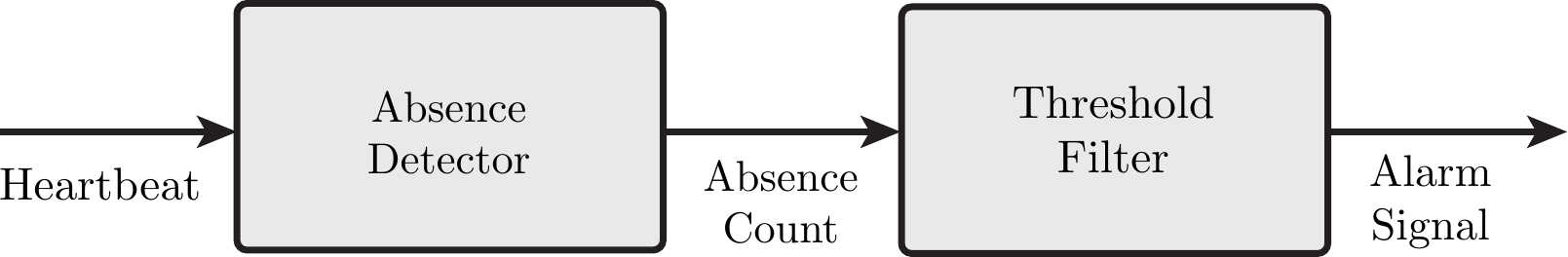}
                \caption{High-level Design \label{archfigure}}
        \end{center}
\end{figure}
  We describe the mapping of leaf goals in the goal model to the components responsible for them and the detailed modeling of the components below.  

\subsection{Chemical Reaction Networks as a programming language}

Our MWT design uses the language of chemical reaction networks (CRNs), which are abstract models of molecular processes in well-mixed solutions \footnote{This expository subsection is adapted from \cite{cEHKLLL14}.}.  All CRNs in this paper are {\it stochastic CRNs}, which model processes in which the presence or absence of very small numbers of certain types of molecules (e.g., a single copy of a viral genome in a living cell) may be significant.  We henceforth omit ``stochastic'' from the terminology.

The CRN model, which goes back at least to 1940 \cite{jDelb40}, has three desirable features.  First, it is mathematically simple.  A CRN is a finite collection of {\it reactions}, each of which has a simple form such as $A+C \overset{r}{\longrightarrow} 2B+C$, where the {\it species} $A$, $B$, and $C$, are abstract types of molecules and the {\it rate constant} $r$ is a positive real number representing the ``propensity'' of an $A$ and a $C$ that encounter one another to react, thereby being replaced by two $B$s and a $C$ in the solution.  A species that, like the species $C$ here, appears on both sides of a reaction is called a {\it catalyst} of the reaction.  Catalysts are extremely important in biochemical processes, and they are extremely useful in our MWT construction.  A {\it state} of a CRN is a vector specifying the number of each species present, and the dynamics of the CRN proceed as a continuous time Markov process with rates derived from the rate constants \cite{oAndKur11,oAndKur15, oAthLah06}. 

A second feature of CRNs is that they are very general.  Every algorithm can, in at least one sense, be efficiently simulated by a CRN \cite{jSCWB08}.

A third desirable feature of CRNs, discovered recently, is that they can be implemented in a uniform way using DNA strand displacement reactions \cite{jSoSeWi10}.  This is fortuitous, because dynamic systems in DNA nanotechnology, including DNA walkers and logic circuits, are typically implemented using DNA strand displacement reactions \cite{oSolo08, jZhaSee11}.  The details of strand displacement reactions are not needed for this paper, but it is relevant to note that they are relatively easy to implement in the laboratory, and that they are easy to specify.  There is a programming language, DSD, in which a large, expressive class of such reactions can be specified and compiled into abstract DNA sequences \cite{jPhiCar09, jLYPEP11}.  CRNs have recently been used as a higher level programming language that can be compiled into DSD \cite{jCDSPCS13, cBSJDTW17}.  

\subsection{Stochastic delay ladders}

The timing functions in our Molecular Watchdog Timer are carried out by {\it stochastic delay ladders} (or simply {\it ladders}), which are CRNs introduced here having very predictable behaviors.  The simplest type of ladder is a $k+1$-rung {\it unary ladder}, which consists of species $X_0, \ldots,X_k$ and the $2k$ reactions 
\begin{align}\label{eq:unary}
\begin{split}
X_i &\overset{u}{\longrightarrow}  X_{i+1}  \quad   (0 \leq i < k),\\
X_i &\overset{r}{\longrightarrow}  X_0      \quad     (0 < i \leq k).
\end{split}
\end{align}
We call $X_0, \ldots,X_k$ the {\it rungs}, $u$ the {\it upward rate constant}, $r$ the {\it reset rate constant}, and  $k$ the {\it height} of this ladder.  The ladder is ``unary'' because each of its reactions has only one molecule on its left-hand side.  The ladder is initialized with all its population on the bottom rung, i.e., a  positive integer $p$ instances of the species $X_0$ and no instances of $X_i$ for $0 < i \leq k$.  Over time, members of this population ``try to climb'' the ladder, sometimes going up from one rung to the next and sometimes falling all the way back to the bottom rung.  (This is a CRN implementation of the ``frog in the well" Markov process \cite{oAthLah06}.)  The total population $p$ of the ladder remains fixed throughout these climbing attempts.

Because the ladder (\ref{eq:unary}) is unary, its kinetic behavior is linear.  This implies that a ladder with population $p$ behaves exactly like an aggregate of $p$ statistically independent ladders with population 1.  (Most CRNs have nonlinear behavior and thus cannot be decomposed into independent, single-molecule nanodevices in this manner.)  Since $p$ is usually large (and often {\it very} large), this statistical independence enables us to predict with high confidence how long it will take (as a function of $p$, $u$, $r$, and $k$) for a given fraction of the population to simultaneously occupy the top rung of the ladder.

In practice, the rate at which reactions occur is governed more by catalysts than by rate constants.  Thus, instead of the unary ladder (\ref{eq:unary}) we introduce {\it catalyzed ladders} of the form \begin{align}\label{eq:catalyzed}
\begin{split}
X_i  + U &\overset{1}{\longrightarrow}  X_{i+1} + U  \quad   (0 \leq i < k),\\
X_i  + R &\overset{1}{\longrightarrow}  X_0 + R    \quad     (0 < i \leq k).
\end{split}
\end{align}
The kinetics of stochastic CRNs make (\ref{eq:catalyzed}) very similar to (\ref{eq:unary}).  For example, if $\#X_i(t)$ is the number of $X_i$ molecules at time $t$ and $u(t)$ is the {\it concentration} of $U$ molecules at time $t$ (i.e., $u(t) = \#U(t)/V$, where V is the volume of the solution), then at time $t$ the first reaction in (\ref{eq:unary}) takes place at rate $u\#X_i(t)$, while the first reaction in (\ref{eq:catalyzed}) takes place at rate $u(t)\#X_i(t)$.  In particular, if $u(t)=u$ and $r(t)=r$ are constant, then the unary ladder (\ref{eq:unary}) and the catalyzed ladder (\ref{eq:catalyzed}) have identical statistical behaviors.  Moreover, even if $u(t)$ and $r(t)$ fluctuate, but do so independently of the ladder's rung populations, the catalyzed ladder (\ref{eq:catalyzed}) enjoys the same decomposability into independent, population-1 ladders as the unary ladder (\ref{eq:unary}).

When controlling reaction rates by catalysts in the above manner, one often takes the rate constants to be 1, as we have done in (\ref{eq:catalyzed}).  In this case, the rate constants are omitted from the notation, and assumed to be 1, as we do below. 

It is now straightforward to specify the two main components of our Molecular Watchdog Timer.
\subsection{Absence Detector}

This component detects when a heartbeat signal has not been present for a specified period of time.   The heartbeat is a ``pulse" of a specific molecular species $H$ that is expected to be periodically output by the molecular application that is being monitored by the MWT.  If the heartbeat is not detected by the MWT for an extended period of time, we can conclude that the molecular application being monitored has failed.  The Absence Detector is assigned responsibility for achieving the leaf goals {\it Avoid [Hdet when Heartbeat is not present]}, {\it Achieve [Hdet when Heartbeat is present]}, {\it Achieve [Initialize to Reset]}, {\it Achieve [Reset if Hdet]}, {\it Achieve [Threshold Delay if Reset]}, and {\it Achieve [Threshold if Hdet is absent]} as shown in Figure~\ref{goalfigure}.  

The Absence Detector component consists of the catalyzed ladder
\begin{align*}
     L_i + U &\rightarrow L_{i+1}+U   \quad(0\leq i<k ),\\
     L_i + H &\rightarrow L_0 + H \quad(0<i\leq k). 
\end{align*}
We also write Y for the top rung $L_k$ of this ladder to emphasize its special role as the upward catalyst for the Threshold Filter below. 

\subsection{Threshold Filter}\label{threshFilt}

This component detects when a target number of individual instances of the Absence Detector have reached the $L_k$ state. The Threshold Filter trips an alarm if and only if enough Absence Detectors are in the $L_k$ state to overcome the {\it constant} number of instances of the reset catalyst R of the Threshold Filter.  The Threshold Filter is assigned responsibility for the leaf goals {\it Avoid [Alarm if Reset]}, {\it Achieve [Alarm if Threshold for some time]}, and {\it Avoid [Alarm until first Threshold]}.

The Threshold Filter consists of the catalyzed ladder
\begin{align*}
        T_i + Y &\rightarrow T_{i+1}+Y  \quad(0\leq i<k ),\\
        T_i + R &\rightarrow T_0 + R \quad(0<i\leq k). 
\end{align*}

We also write D for the top rung $T_k$ of the Threshold Filter.  This species D is the alarm species of the MWT.  It 
is used by external modules to trigger a response. In previous work \cite{cEHKLLL14} we described a simple, one-time, Alarm response produced by amplifying the alarm molecular species.   In this paper we will present a more powerful response, namely a 
Recovery component, after introducing a monitored system in Section~\ref{mwtRecovery}.

\section{Design Verification}\label{sec:verification}

Formal design analysis provides some assurance that the behavior specified in the design matches the system's intended behavior.   We must ensure that with very high probability when there is a heartbeat, the MWT does not alarm, and that with very high probability when there is no heartbeat, the MWT quickly alarms.
  
In this section we first describe the analysis and verification techniques used to verify the design of the MWT and report the verification results.   We check that the MWT works for a very long time in normal conditions (i.e., when the heartbeat from the monitored system is present) and show via quantitative simulation that the MWT alarms as quickly as the client needs (indicated via the initial parameter values) when the heartbeat from the external system disappears. 

We then demonstrate the functionality of the MWT by introducing a specific example of a system that needs to be monitored at runtime and verifying that its heartbeat behavior is correct.  We connect the CRN model of the monitored system to that of the MWT and verify the correct functioning of the composed system.  Finally, we describe how we extended the existing system to enable the MWT to not only detect the monitored system's failure at runtime, i.e., the absence of the expected heartbeat pulse, but also to trigger the monitored system's runtime recovery, and how we verified this additional capability. 

\subsection{Verifying the MWT design}

To check the correctness and robustness of the MWT design, we followed an incremental process of simulation of the CRN model for sanity checks and selection of likely parameter value ranges, followed by model checking of the CSL leaf goals on the CRN model across those parameter values. We also injected faults that had been previously discovered by analytical reasoning to confirm that the model checkers found them.   
To recap, the agents to which the leaf subgoals in the goal model are assigned--the Absence Detector and the Threshold Filter--are described in the CRN high-level programming language. The properties to be checked against the CRN model are the CSL formal specification of the  leaf goals.    

{\it CRN input to verification tools.} A strength of CRN as a language for molecular programming is that a CRN model can be readily imported into MATLAB's SimBiology package, allowing simulations to be run on it.   We used SimBiology to understand the behavior and performance of our models and to debug them.   

A CRN model can also be used as input into the probabilistic model checker  
PRISM \cite{cKwNoPa11}, used previously to analyze molecular systems, e.g., in \cite{cKwia14, jKwiTha14}. We used PRISM to verify that the MWT design satisfied the CSL properties derived from the goal model.  
The MWT model takes a number of parameters ranging from rate constants (of DNA reactions) to the length of the absence detector ladders (which are DNA strands). The parametrized design allowed us to automate testing across ranges of parameter values for optimization and verification.  

{\it Parameters.} The values of the client parameters are specified by the client and depend on the system being monitored. For example, how quickly a heartbeat failure must be detected will vary among different applications.  Additionally, there are modeling parameters that were needed to verify the realizability of the specified goals.  An example is the number of each species of molecule represented in a particular model.    The values for these parameters vary as the design space is explored.   Their correct representations emerge from the formal analysis, mathematical proofs that satisfaction of the subgoals given the environmental assumptions satisfy the root goal, simulation of the composed system (monitored component and MWT), and the model checking results. Correct parametrization of the models was time-consuming and incremental.   We used custom MATLAB scripts to generate models of arbitrary parameters that could integrate with the model checkers. These scripts along with some features provided by PRISM automated the exploration of the parameter space to discover models that provably satisfied our requirements.

{\it Reusable MWT.}  Making the MWT reusable enables the composed system (monitored system and MWT) to recover from faults and then continue execution autonomously, without outside intervention.   This is important if the MWT is to be used in vivo, for example. If the MWT is not reusable, then even if the monitored system recovers, it is no longer being monitored.  The technical difficulty in moving from a throw-away MWT (one-time use) to a reusable MWT was the creation of a design for which the initial condition (the set of values for the parameters) could be restored autonomously.   At initialization, we thus configured the  MWT to begin with the majority of rung molecules in the lower rungs of the ladder. When a heartbeat is detected, the absence detector enters this reset state again. This resets the MWT for reuse. We validated via simulation in SimBiology that the MWT is reusable rather than needing to be discarded after a single use. After having alarmed, it is reset when the monitored system begins issuing heartbeats again.  

\subsection{Verifying the interaction of the MWT with a monitored system}\label{mwtRecovery}

The question that underlies the verification of the interaction of the MWT with a monitored system is ``Could any such composed molecular system satisfy the assumptions we make on it and be used productively?'' 
Given current model-checking limits, we verify the interaction at two very different molecular population scales, 
both of which are in realistic ranges of molecular counts.   

First we verify using {\it simulation and probabilistic model checking} that the interaction of an abstract monitored system with the MWT is correct with respect to the requirements and show that the assumptions we make on this hypothetical monitored system are realistic.  This first version operates on an abstract signal received from a hypothetical monitored system. We simulate using SimBiology and model check using PRISM that the MWT transforms the absence of a heartbeat signal into an alarm signal correctly. The size for which we can model check it is small, with molecular counts up to 5.  For example, 5 absence detectors produced a CTMC with over 150,000 states, while 10 absence detectors produced a CTMC with over 9 million states. 

Second, we validate by {\it simulation} the correct interaction of an example monitored system with the MWT and show that the assumptions made on this specific monitored system are realistic.  We first introduce an example of a monitored system and verify its correctness, using simulation and model checking.   This entails verifying that the presence or absence of its heartbeat is correlated to the monitored system's health or failure.  The example monitored system cannot be model checked for all possible values of its parameters due to the size of the possible space; however, we simulated it with up to a total species population count of 100,000. 

{\it Introducing a system to be monitored.} To demonstrate the capabilities of the MWT design we used it to monitor the health of a standard molecular system, namely an oscillator.    Chemical oscillators occur widely in nature, so are important, and synthetic molecular oscillators previously have been used as benchmarks in multiple projects, e.g.  \cite{jLYCP12, cDaSePh14,  jDLLMPS15, cHorMur15, jFSCHPS17, cBSJDTW17}.   See Section~\ref{relatedWork} for more information. Another advantage of selecting the oscillator is that we could readily extend it to output a heartbeat (to test the normal case) and cause it to cease output of a heartbeat (to test the failure cases). 

We used the Lotka-Volterra 3-Phase Oscillator \cite{oCard06}, which employs three species A, B, and C. Each of the three species corresponds to a single phase of the oscillator. The oscillator is initialized with the molecular count of one of the three species being high and the molecular counts of the other two species being low. After initialization, the oscillator will cycle between the phases following the order A to B to C and then back to A.
As an example, consider the following case. If A is dominating and B and C have similar molecular counts, then reactions \eqref{lotka1} and \eqref{lotka3} below are equally likely to occur. However, when reaction \eqref{lotka1} or \eqref{lotka3} fires, the rates of all the reactions change, increasing the rate of reaction \eqref{lotka1} and decreasing the rate of reaction \eqref{lotka3}. This continues until B is dominating, completing the transition to phase B. A similar sequence of events occurs for each phase transition.

We extended the CRN model for the stochastic 3-phase oscillator
with a heartbeat interface that produces a heartbeat (H) when the oscillator is healthy.  A heartbeat interface is required in order to use the MWT to monitor the  oscillator. 
The CRN for the oscillator plus its heartbeat interface is: 
\begin{align}
A + B &\xrightarrow{k}2 B + H \label{lotka1} \\
B + C &\xrightarrow{k} 2 C\\
C + A &\xrightarrow{k} 2 A \label{lotka3}\\
H &\xrightarrow{k_2}  \emptyset
\end{align}

In order to be useful, the heartbeat interface must cover all possible failure modes. In a stochastic setting, there are two possible faults in the  oscillator that cause it to fail. First, if any of the three species A, B or C has a molecular count of zero, the oscillations will stop and the oscillator fail. Second, if the oscillator spends a large amount of time with all three species near a state of equilibrium, the oscillations will become negligible and desultory and the oscillator fail.  
While the oscillator is working correctly, a large number of H molecules will be created as reaction \eqref{lotka1} occurs. In order for this interface to be correct, it must produce fewer heartbeats in the case of a fault. For the first fault, if any of the three oscillator species has a count of zero, all oscillator reactions will quickly have a rate of 0, causing a stop in the production of heartbeat molecules. For the second fault, if all three oscillator species are in equilibrium then the H species will fall to a roughly constant amount low enough to cause the Absence Detector to activate.
 
 {\it Checking correlation of heartbeat with oscillator's state.} In ``real-world'' scenarios it is the monitored system's responsibility to provide a correct heartbeat, meaning that the MWT assumes that the heartbeat accurately reflects the normal or failed state of the monitored system.   However, in order to confirm that the MWT operates correctly and robustly, we had to first confirm that the oscillator outputs a correct heartbeat.  We thus had to  check (1) that the oscillator's health correlates with the presence (healthy) or absence (unhealthy) of heartbeat molecules at its interface, and (2) that the MWT's behavior correlates with the presence or absence of heartbeat molecules over a period of time at its interface with the  monitored system.   
 
 A state of the monitored system can be healthy or unhealthy. Informally, in a healthy state, a heartbeat will be sent within a reasonable time or the state will quickly become unhealthy.  In an unhealthy state, no heartbeat will be sent within a reasonable time or the state will quickly become healthy.  We define a state to be {\it healthy} at time t as A, B, C $>$ 0 AND $(A-B)^2 + (B-C)^2 + (C-A)^2 > \tau$, where $\tau$ is defined to ensure that the oscillator is deemed unhealthy if its species counts approach equilibrium.

The three properties to be verified are:

{\it Achieve[Produce heartbeats while healthy]}
\begin{itemize}
\item[] $\prob{\ge1} [ \Box (\text{healthy} \implies \prob{\ge1-\delta_1} [ \Diamond_{\le t_1} ((\text{hbHigh}\;\lor\;\lnot\text{healthy}) ]) ]$
\end{itemize}

{\it Avoid [Produce heartbeats while unhealthy]}
\begin{itemize}
\item[] $\prob{\ge1} [\Box (\lnot\text{healthy} \implies$\\
\hspace*{0.5in}$\prob{\ge 1-\delta_2} [\Diamond_{\le t2} (\prob{\ge 1-\delta_3} [\text{hbLow} \W (\prob{\ge 1-\delta_4} [\Box_{\le t_3} \text{healthy}])])] )]$
\end{itemize}

{\it Heartbeat decays}
\begin{itemize}
	\item[] $\prob{\ge 1} [\Box (\text{hbHigh} \implies \prob{1-\delta_5} [ \Diamond_{\le t4} \lnot \text{hbHigh} ] )]$
\end{itemize}

We simulated in SimBiology the oscillator and heartbeat interface with a range of initial counts of A, B, and C up to 1000 (e.g., 80\% in A, 10\% in B and C)  and an initial count of 0 for H, and checked that these three properties held in the simulations.   The simulations demonstrated that the presence or absence of a heartbeat correlates with the health of the oscillator in both of the two failure modes.  Using a CTMC model of the oscillator  with total population of 200, we then verified in PRISM that the oscillator plus heartbeat interface satisfied the goals.   The model checker verified true for the three CSL properties above.  

\begin{figure}
	\begin{center}
		\includegraphics[width=4.0in]{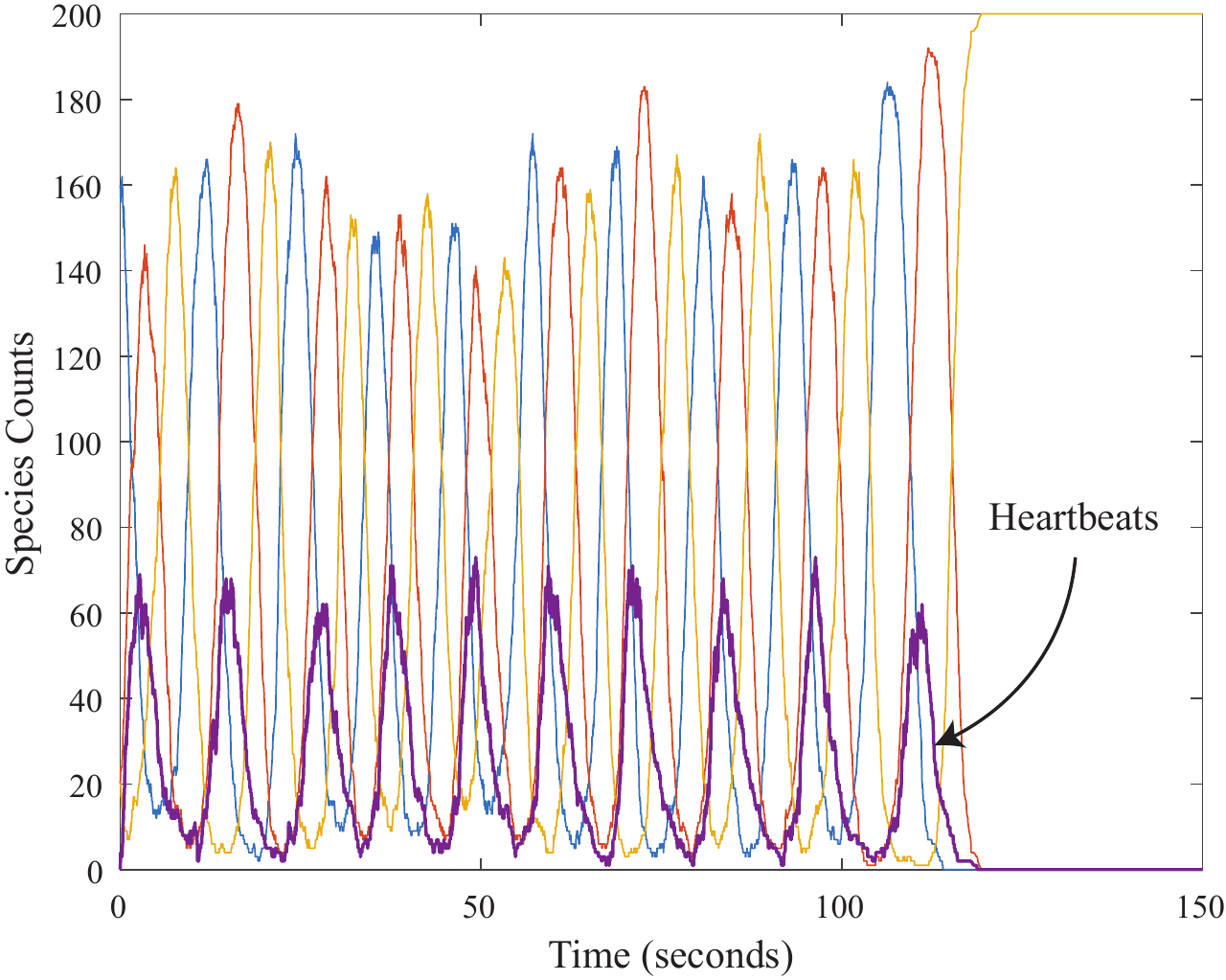}
		\caption{A simulation of the oscillator with the heartbeat interface. \label{oscHBsFigure}}
	\end{center}
\end{figure}
Figure~\ref{oscHBsFigure} shows a simulation of the oscillator for the first failure mode, in which one of the species counts went to 0. Since no oscillations could occur, heartbeats ceased to be produced by the oscillator. 

\subsection{MWT triggers oscillator's runtime recovery}

A system is more robust if it can autonomously recover from a failure.  Beyond detecting and reporting the failure of the monitored system, we also sought to use the MWT's Alarm signal to trigger autonomous recovery in a monitored system after it has failed.  To demonstrate this, we constructed a recovery component that, upon receiving the MWT's Alarm, will recover the oscillator from either of  its failure modes.

The CRN for the oscillator's recovery module is: 
\begin{align}
D + A &\xrightarrow{k} D + B \\
D + B &\xrightarrow{k} D + C \\
D + C &\xrightarrow{k_2} D + A,
\end{align}
where the third reaction has a different rate from the other two. 

These reactions are triggered by the presence of the MWT Alarm component's output signal (the D's) produced when the MWT detects that the oscillator has failed. In both the case where the oscillator fails due to running out of the species A, B, or C, and the case where it fails because it reached an equilibrium state (an equal number of A's, B's and C's) such that heartbeats stops, these reactions will recover the oscillator .  The recovery ``jump-starts" the oscillator and, when the heartbeat starts up again, the Alarm signal (the D's) fade away.  To check the correct behavior of the recovery capability, we ran simulations with populations of 100 to 1000 molecules for the oscillator and varying percentages of D molecules.   The MWT's Alarm signal correctly triggered the oscillator recovery.    

\begin{figure}
	\begin{center}
		\includegraphics[width=4.0in]{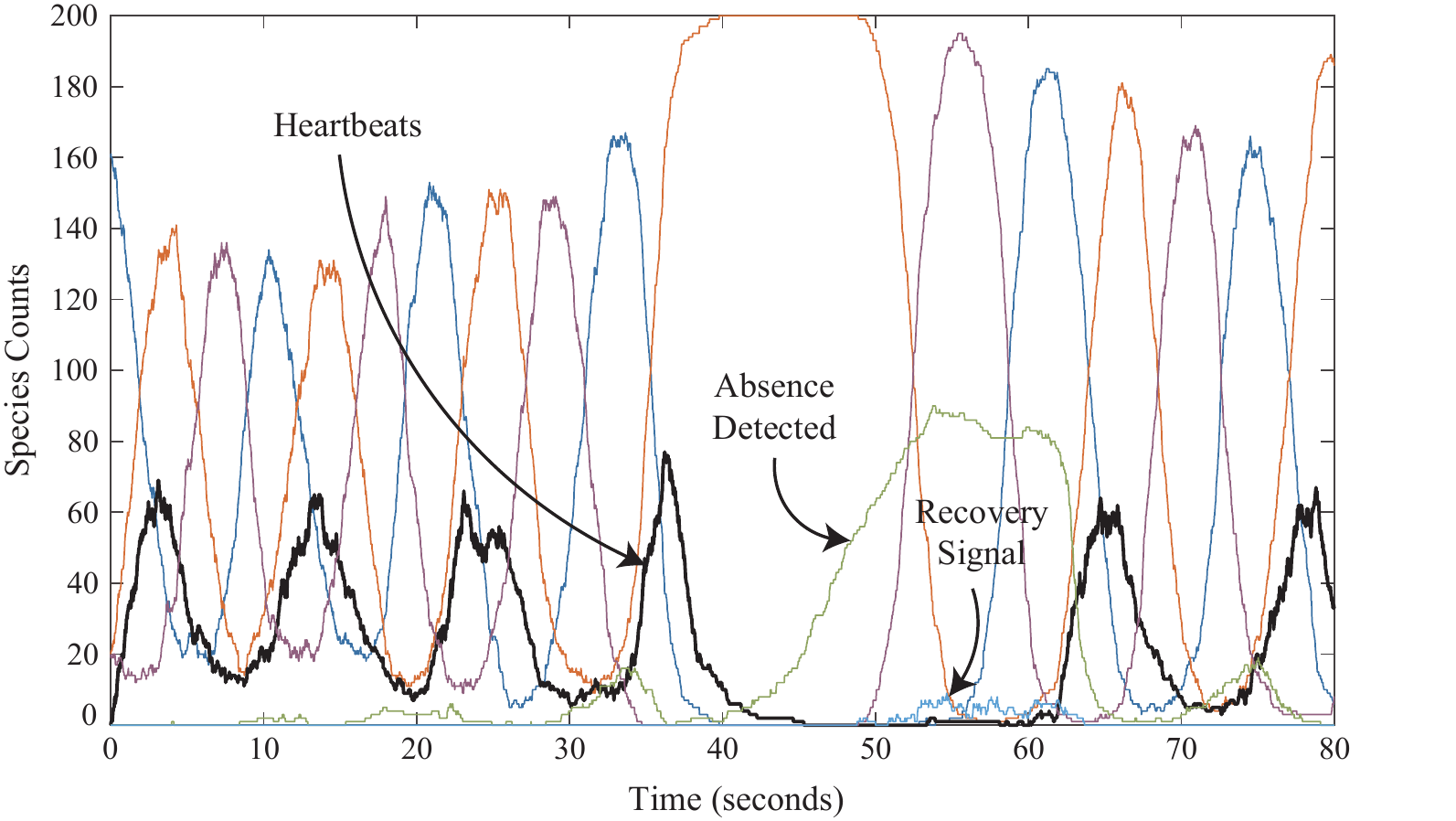}
		\caption{A simulation of the MWT detecting the oscillator's failure and triggering its recovery. \label{oscMWTFigure}}
	\end{center}
\end{figure}
Figure~\ref{oscMWTFigure} shows a single stochastic simulation of the composed system.  Initially, the oscillator works and produces heartbeats. During this time, the absence detected signal remains low. However, once the oscillator fails and the heartbeats stop, the absence is quickly detected and the recovery signal is released to trigger the oscillator to recover. Shortly after the recovery signal is turned on, the oscillator begins normal operation again.

{\it Generalizing to other applications.}  Usage of the MWT is intended to be broad.  The MWT is designed to work with any molecular system that needs to be monitored for the absence of a heartbeat, meaning that it is independent of how the heartbeats are input to it.  The oscillator provides an example of how, given a heartbeat interface from a client monitored system specified as a rate at which heartbeats are produced and in what quantity, we can construct a MWT to monitor it. 
To use our MWT to monitor a system, the client monitored system (here, the oscillator) describes what it needs the MWT to do by specifying a {\it polytope}, that is a multi-dimensional space, defined by four parameters. The four client-provided parameters are: 
$u$---the minimum time between a heartbeat and an alarm,
$v$---the maximum time between a heartbeat and an alarm, 
$\varepsilon$---the probability of error allowed by the $u$ delay,  and
$\delta$---the probability of error in achieving the $v$ delay. 
For example, a client might specify that the MWT should allow a minimum of 10 seconds and a maximum of 20 seconds after a heartbeat before an alarm, and it must achieve these time bounds with probabilities of 95\% for both.  Thus, $u$ is 10, $v$ is 20, and $\varepsilon$ and $\delta$ are 5\%. 
The client also specifies the minimum and maximum size of the heartbeat pulse.  If the client gives us parameter values from within this space, we {\bf will provide} a design model for a MWT that satisfies the goal diagram.  

There are other, internally generated (rather than client-defined) parameters that formalize the goal diagram's constraints.  For example, $\varepsilon$ is broken into two internally generated parameters to enable the goal proofs.  Thus, in addition to the four client-provided parameters, there are 22 internally generated parameters specifying probability and timing constraints.  These are listed in the Appendix.   

\subsection{Mapping to molecules}

We have designed and verified our MWT at the CRN level of abstraction, but it is useful to estimate the feasibility of actually implementing our design in DNA.  For this purpose we used the compiler reported in a very recent paper by Badelt et al. \cite{cBSJDTW17}.  Using this compiler and its encoding of the translation scheme of Chen et al. \cite{jCDSPCS13}, we compiled the CRN for a small MWT and oscillator into DNA strands.  The MWT has a 4-rung Absence Detector ladder and a 5-rung Threshold Detector ladder and includes the three recovery reactions.  As reported in the Appendix, this MWT alone compiles to 122 distinct DNA strands.  The combination of this MWT with the CRN for the three-phase oscillator (modified to include heartbeat production) compiles to 147 distinct DNA strands.  DNA strand displacement systems of this size are already feasible for laboratory implementation.  For example, Qian and Winfree \cite{jQiaWin11a} reported the successful implementation of a 130-strand DNA displacement device.  Larger MWTs (ones with longer stochastic delay ladders) will be feasible in the near future. 

\subsection{Discussion}
To summarize, the MWT is a programmable, molecular safety mechanism.  In the creation of the design for the MWT, this paper has proposed a new use of software engineering techniques and tools in a development process that can be applied more generally to create other molecular systems.   The goal-oriented requirements engineering process described in Sect. 2 systematically develops the requirements for a molecular programmed system, using continuous stochastic logic (CSL) to specify the requirements (leaf goals) and to verify the goal refinement with machine-checkable proofs.  The assignment of  requirements to the components' designs responsible for achieving them is described in Sect. 3.   The designs are formally specified as stochastic chemical reaction networks (CRNs). CRNs recently have become widely used to specify programmed molecular systems since compilers now exist from CRNs into lower-level (DNA-strand-level) designs and from there into molecules.   Modeling the designs as CRNs supports both simulation and verification, and both MATLAB's SimBiology and the PRISM probabilistic model checker accept CRN input.  This section has described the simulation and verification of the MWT, first as a standalone device; second, when connected to a device, such as the oscillator, that needs to be monitored for a failure event; and third, when connected to a device that needs to be triggered by the MWT to also recover from its failure event at runtime.

\section{Related Work}\label{relatedWork}

In this section we briefly describe additional related work in molecular design software, model checking, and molecular oscillators.

{\it Molecular design software.}  Other design tools for DNA computational devices exist but operate at a very detailed design level.  Two open-source software tools are CaDNAno \cite{jDMTVCS09} and CANDO \cite{jKKDB12}.   They are widely used to design, debug and optimize the stability and physical properties (torque, flexibility, energy wells) of 2-dimensional and 3-dimensional DNA origami structures, and operate at a much lower level than our design considerations here.  

{\it Model checking}. In related work Kwiatkowska and Thachuk described the probabilistic verification of CRNs for biological systems using the probabilistic model checker PRISM \cite{jKwiTha14}. Their work showed the benefits of probabilistic model checking for molecular systems and informed our work for the MWT.  
PRISM interfaces with Visual DSD, a design tool for DNA strand displacement \cite {jLYPEP11}.  

For large systems, including molecular ones, there is a disconnect between the size of the model that can be automatically checked and the system.  One of the problems we face is how to prune the model such that we can do meaningful model checking.   Pavese, Barberman and Uchitel described how to develop partial explorations of a system model automatically \cite{jPaBrUc16}. Their technique has promise for use in molecular programs that we hope to explore. Since many molecular programs deal with extremely large, if not infinite, state spaces, probabilistic model checking on partial system explorations might provide bounds on the reliability of a molecular system that is too large to model check.   

Stochastic models for systems often have distributions that are empirically determined or only partially known.   Moreover, small differences between stochastic models' parameter values and their real-world counterparts can change the results of verification.  Meedeniya et al. have used Monte Carlo simulations to generate a reliability evaluation of a probabilistic model of an antilock brake system with uncertain parameters \cite{jMMAG14}.  Su, Chen, Feng, and Rosenblum recently extended previous work by Su and Rosenblum \cite{oSuRos13} on perturbations in model checking parameters in discrete-time Markov chains to allow model checking on time-bounded CTMCs with imprecise values for transition rates \cite{cSCFR17}. Since molecular systems  have imprecise reaction rates, determining the effects of parameter variance on the models is important, and the applicability of their approaches to programmed molecular systems merits investigation.

More broadly, there has been significant recent progress in modeling biological or chemical systems. Yordanov et al. formalized and encoded DNA computing to allow use of Satisfiability Modulo Theories (SMT) \cite{ jYWHK13}. Fisher, Harel and Henzinger performed computational modeling of biological systems as reactive systems \cite{jFiHaHe11}.  Hetherington et al. and Sumner et al. composed an advanced computational model of a biological system from sub-models describing its different aspects \cite{jHSSLRY12, jSHSLRY12}.   David et al. created translators to convert SimBiology models for biological systems into CTMCs for stochastic model checking or into ODEs for simulation \cite{jDLLMPS15}. 

{\it Molecular oscillators}. Hori and Murray, in a recent paper on synthetic biochemical oscillators, stated that, ``The reliable engineering of oscillators is an important milestone towards robust synthesis of more complex dynamical circuits in synthetic biology'' \cite{cHorMur15}.  Gene regulatory networks, for example, use oscillators, and Fern et al. recently reported the use of timer circuits to precisely coordinate chemical events in vitro \cite{jFSCHPS17}.  3-phase oscillators seem to have been first reported in \cite{oLacSel95} and more recently \cite{oCard06,jCard13, jLYCP12}.    The 2-phase Lotka-Volterra oscillator also has been studied in the context of DNA strand displacement in \cite{jSoSeWi10,  jLYCP12}.  Ballarini, Mardare and Mura, and Ballarini and Guerriero presented analyses of the 3-phase oscillator using PRISM and described both of the failures modes that our MWT design successfully detects \cite{jBaMaMu09, jBalGue10}.  

\section{Conclusion}\label{sec:conclusion}
Monitoring the health of programmed molecular systems at runtime is critically important. Envisioned applications such as biocompatible diagnostic systems and smart drug therapy systems will need such monitoring capabilities to operate safely.  Using goal-oriented requirements engineering, machine-checked proofs, reaction network modeling, stochastic simulation, and probabilistic model checking, we have designed and verified a Molecular Watchdog Timer that can monitor a molecular system at runtime, detect when the heartbeat signal from the monitored system stops, and alarm to trigger its recovery.    The MWT is modular, designed to operate correctly in the probabilistic chemical environment, and robust to failure-prone components.   Using chemical reaction networks as a programming language, we have implemented both the MWT and a monitored system (a molecular oscillator) as chemical reaction networks.   We have demonstrated the MWT's capabilities by showing that the molecular watchdog timer reliably detects failures of the oscillator and triggers its recovery at runtime.   

Many other programmed molecular systems will be needed and developed in the future.  The MWT is an example of a cybermolecular system, a molecular programmed system that senses and controls its environment, including other molecular systems. Cybermolecular systems and bio-compatible computing devices are moving rapidly from the laboratory to widespread usage in daily life.  We hope that our software engineering-inspired approach to designing and verifying the molecular watchdog timer can assist in the future design of predictable and safe molecular systems.

\section*{Acknowledgments}
We thank Samik Basu, Gianfranco Ciardo, Anthony Finkelstein, Carlo Ghezzi, Marta Kwiatkowska, Axel van Lamsweerde, Paul Rothemund, and Erik Winfree for useful discussions. We thank Jeremy Avigad and Johannes H\"olzl for useful suggestions on theorem provers.   We thank the anonymous referees for remarks that improved the paper.   

\bibliographystyle{plain}

\appendix

\begin{center}
	\huge
	Appendix
\end{center}

\vspace*{1em}

\section{Internal Parameters of the Goals}
Here we list the internal parameters generated during the goal refinement.
\begin{itemize}
	\item $\epsilon_1$ and $\epsilon_2$ are refinements from $\epsilon$ that determine the allowed error in avoiding $Alarms$ while heartbeats are present.
	\item $w_a$ is a time bound on turning on the alarm.
	\item $\alpha$ and $\beta$ are allowed error in detecting the presence of heartbeats.
	\item $w_h$ is the maximum time to detect the presence or absence of a heartbeat.
	\item $\epsilon_1'$ and $\epsilon_2'$ are allowed error in avoiding $Alarms$ while a heartbeat is detected.
	\item $g$ is the time allowed between detecting a heartbeat and keeping the $Alarm$ off.
	\item $\delta_1'$ is the allowed error in initiating an $Alarm$ when no heartbeat is detected.
	\item $\lambda_1$ is the allowed error in $Resetting$ when a heartbeat is detected.
	\item $w_{on}$ is the maximum time to $Reset$ when a heartbeat is detected.
	\item $\gamma_1$ is the allowed error in setting the threshold to low when $Reset$ is true.
	\item $\eta_1, \eta_2,$ and $\eta_3$ are allowed errors in setting the threshold to high until a heartbeat is detected from a time when no heartbeat is detected.
	\item $w_{th}$ is a time bound on how long it takes to set the threshold to high.
	\item $\lambda_2$ and $\lambda_3$ are allowed errors in avoiding $Alarms$ while the threshold is low.
	\item $w_{off}$ is the maximum time allowed between a low threshold and keeping the $Alarm$ off.
	\item $\eta_4$ is the allowed error in turning the $Alarm$ on after the threshold is high.
	\item $\gamma_2$ is the allowed error that the $Alarm$ is off at least until the first time that the threshold is not low.

\end{itemize}

\section{Formal Goal Specification and Agent Assignment}
The following table provides a breadth first listing of the goals for the Molecular Watchdog Timer (MWT) in Figure~2. It lists the goal description, the formal CSL specification, and, for leaf nodes, the agent assigned responsibility, where AD and TF represent the Absence Detector and Threshold Filter respectively.
\newpage
\begin{center}
\begin{tabular}[t]{p{3cm}p{7cm}p{1cm}} \toprule
	\textbf{Goal} & \textbf{CSL Specification} & \textbf{Agent}\\ \midrule
	
	\vspace*{-40pt}
	\begin{flushleft}
		\textbf{ACHIEVE:} \\
		Alarm iff no Heartbeat provided within $t$ time
	\end{flushleft}&
	\begin{minipage}{7cm}
		$
		\begin{aligned}
		&\prob{\ge 1-\epsilon}\Box_{\le u}\lnot Alarm\;\land \\
		&\prob{\ge 1}\Box\left[ H_{pres}\implies \right. \\
		& \left.\prob{\ge 1-\epsilon_1}\Diamond_{\le g}\left( \prob{\ge 1-\epsilon_2}\Box_{\le u}\lnot Alarm \right) \right]\;\land\\
		&\prob{\ge 1}\Box\left[ \lnot H_{pres}\implies \right.\\
		& \left.\prob{\ge 1-\delta_1}\Diamond_{\le v-w_a}\left( Alarm\lor H_{pres} \right) \right]
		\end{aligned}
		$
	\end{minipage}
	&  \\ \midrule
	\vspace*{-32pt}
	\begin{flushleft}
		\textbf{ACHIEVE:} \\
		Heartbeat Detected correctly tracks the presence of Heartbeats within $t-t'$ time
	\end{flushleft}&
	\begin{minipage}{7cm}
		$
		\begin{aligned}
		& \prob{\ge 1} \Box\left[H_{pres}\implies\right.\\
		& \left.  \prob{\ge 1-\beta}\Diamond_{\le w_h}H_{det}\right]\;\land\\
		&\prob{\ge 1}\Box\left[\lnot H_{pres}\implies \right. \\
		& \left. \prob{\ge 1-\beta}\Diamond_{w_h}\prob{\ge 1-\alpha}\left(\lnot H_{det}\W H_{pres}\right)\right]
		\end{aligned}
		$
	\end{minipage}
	&  \\ \midrule
	\vspace*{-40pt}
	\begin{flushleft}
		\textbf{ACHIEVE:} \\
		Alarm iff no Heartbeat detected within $t'$ time.
	\end{flushleft}
	&
	\begin{minipage}{7cm}
		$
		\begin{aligned}
		&\prob{\ge 1-\epsilon}\Box_{\le u}\lnot Alarm\;\land \\
		&\prob{\ge 1}\Box\left[ H_{det}\implies \right. \\
		& \left.\prob{\ge 1-\epsilon_1'}\Diamond_{\le g}\left( \prob{\ge 1-\epsilon_2'}\Box_{\le u}\lnot Alarm \right) \right]\;\land\\
		&\prob{\ge 1}\Box\left[ \lnot H_{det}\implies \right. \\
		& \left.\prob{\ge 1-\delta_1'}\Diamond_{\le v-w_a}\left( Alarm\lor H_{pres} \right) \right]
		\end{aligned}
		$
	\end{minipage}
	&  \\ \midrule
	\vspace*{-15pt}
	\begin{flushleft}
		\textbf{AVOID:} \\
		Heartbeat Detected when Heartbeat not present
	\end{flushleft} &
	\begin{minipage}{7cm}
		$
		\begin{aligned}
		&\prob{\ge 1}\Box\left[\lnot H_{pres}\implies \right.\\
		&\left.\prob{\ge 1-\beta}\Diamond_{w_h}\prob{\ge 1-\alpha}\left(\lnot H_{det}\W H_{pres}\right)\right] 
		\end{aligned}
		$
	\end{minipage}
	& AD \\ \midrule
	\vspace*{-6pt}
	\begin{flushleft}
		\textbf{ACHIEVE:} \\
		Heartbeat Detected when Heartbeat present
	\end{flushleft}
	&
	\begin{minipage}{7cm}
		$
		\begin{aligned}
		&\prob{\ge 1}\Box\left[H_{pres}\implies \prob{\ge 1-\beta}\Diamond_{\le w_h}H_{det }\right] 
		\end{aligned}
		$
	\end{minipage}
	& AD \\ \midrule
	\vspace*{-15pt}
	\begin{flushleft}
		\textbf{ACHIEVE:} \\
		Correct Timer Reset
	\end{flushleft}
	&
	\begin{minipage}{7cm}
		$
		\begin{aligned}
		&Reset\;\land\\
		&\prob{\ge 1}\Box\left[H_{det}\implies \prob{\ge 1-\lambda_1}\Diamond_{\le w_{on}} Reset \right]
		\end{aligned}
		$
	\end{minipage}
	&  \\ \midrule

\end{tabular}
\end{center}
\newpage
\begin{center}
\begin{tabular}[t]{p{3cm}p{7cm}p{1cm}} \toprule
	\textbf{Goal} & \textbf{CSL Specification} & \textbf{Agent}\\ \midrule
			\vspace*{-30pt}
	\begin{flushleft}
		\textbf{ACHIEVE:} \\
		Correct Delay
	\end{flushleft}
	&
	\begin{minipage}{7cm}
		$
		\begin{aligned}
		&\prob{\ge 1}\Box\left[Reset\implies \prob{\ge 1-\gamma_1}\Box_{\le u}Th_L \right]\;\land\\
		&Th_L\implies \lnot Th_H\\
		&\prob{\ge 1}\Box[\lnot H_{det}\implies \prob{\ge 1-\eta_1}\Diamond_{v-w_a-2*w_h-w_{th}}\prob{\ge 1-\eta_2}\\
		&\quad\left(Th_H\W \prob{\ge 1-\eta_3}\Diamond_{\le wh}H_{det}\right)]
		\end{aligned}
		$
	\end{minipage}
	&  \\ \midrule
	\vspace*{-30pt}
	\begin{flushleft}
		\textbf{ACHIEVE:} \\
		Alarm iff Threshold met
	\end{flushleft}
	&
	\begin{minipage}{7cm}
		$
		\begin{aligned}
		&\prob{\ge 1}\Box\left[Th_L\implies\right.\\
		&\left.\prob{\ge 1-\lambda_2}\Diamond_{\le w_{off}} \prob{\ge 1-\lambda_3}\Box_{\le u}\lnot Alarm\right]\;\land\\
		&\prob{\ge 1}\Box\left[Th_H\implies\right.\\
		&\left. \prob{\ge 1-\eta_4}\Diamond_{\le w_{th}}\left(Alarm\;\lor\;\lnot Th_H\right)\right]\;\land\\
		&\prob{\ge 1-\gamma_2}\left(\lnot Alarm\W \lnot Th_L\right)
		\end{aligned}
		$
	\end{minipage}
	&  \\ \midrule
	\vspace*{-5pt}
	\begin{flushleft}
		\textbf{ACHIEVE:} \\
		Initialize to Reset
	\end{flushleft}
	&
	\begin{minipage}{7cm}
		$
		\begin{aligned}
		&Reset
		\end{aligned}
		$
	\end{minipage}
	& AD \\ \midrule
	\vspace*{-5pt}
	\begin{flushleft}
		\textbf{ACHIEVE:} \\
		Reset if Hdet
	\end{flushleft}
	&
	\begin{minipage}{7cm}
		$
		\begin{aligned}
		&\prob{\ge 1}\Box\left[H_{det}\implies \prob{\ge 1-\lambda_1}\Diamond_{\le w_{on}} Reset \right]
		\end{aligned}
		$
	\end{minipage}
	& AD \\ \midrule
	\vspace*{-5pt}
	\begin{flushleft}
		\textbf{ACHIEVE:} \\
		Threshold delay if Reset
	\end{flushleft}
	&
	\begin{minipage}{7cm}
		$
		\begin{aligned}
		&\prob{\ge 1}\Box\left[Reset\implies \prob{\ge 1-\gamma_1}\Box_{\le u}Th_L \right]
		\end{aligned}
		$
	\end{minipage}
	& AD \\ \midrule
	\vspace*{-14pt}
	\begin{flushleft}
		\textbf{ACHIEVE:} \\
		Threshold if Hdet is absent
	\end{flushleft}
	&
	\begin{minipage}{7cm}
		$
		\begin{aligned}
		&\prob{\ge 1}\Box[\lnot H_{det}\implies\\
		& \prob{\ge 1-\eta_1}\Diamond_{v-w_a-2w_h-w_{th}}\prob{\ge 1-\eta_2} \\
		&\quad\left(Th_H\W \prob{\ge 1-\eta_3}\Diamond_{\le wh}H_{det}\right)]
		\end{aligned}
		$
	\end{minipage}
	& AD \\ \midrule
	\vspace*{-16pt}
	\begin{flushleft}
		\textbf{AVOID:} \\
		Alarm if Reset
	\end{flushleft}
	&
	\begin{minipage}{7cm}
		$
		\begin{aligned}
		&\prob{\ge 1}\Box\left[Th_L\implies\right.\\
		&\left. \prob{\ge 1-\lambda_2}\Diamond_{\le w_{off}} \prob{\ge 1-\lambda_3}\Box_{\le u}\lnot Alarm\right]
		\end{aligned}
		$
	\end{minipage}
	& TF \\ \midrule
	\vspace*{-14pt}
	\begin{flushleft}
		\textbf{ACHIEVE:} \\
		Alarm if Threshold for some time
	\end{flushleft}
	&
	\begin{minipage}{7cm}
		$
		\begin{aligned}
		&\prob{\ge 1}\Box\left[Th_H\implies \right.\\
		&\left.\prob{\ge 1-\eta_4}\Diamond_{\le w_{th}}\left(Alarm\;\lor\;\lnot Th_H\right)\right]
		\end{aligned}
		$
	\end{minipage}
	& TF \\ \midrule
	\vspace*{-6pt}
	\begin{flushleft}
		\textbf{AVOID:} \\
		Alarm until first Threshold
	\end{flushleft}
	&
	\begin{minipage}{7cm}
		$
		\begin{aligned}
		&\prob{\ge 1-\gamma_2}\left(\lnot Alarm\W \lnot Th_L\right)
		\end{aligned}
		$
	\end{minipage}
	& TF
\end{tabular}
\end{center}

\section{Verification of Goal Model I: Theorems}
In this section and the following section we verify that the satisfaction of a goal can be fulfilled by the satisfaction of all its subgoals, and hence the entire goal tree is verified.

\subsection{Theorems}
The proof of our goal model correctness is broken into the following three main theorems, and each corresponds to one of the goal refinements in the model.
All of the theorems depend on the following constrains on the internal parameters of the model.
\begin{align*}
	(1-\epsilon_1)(1-\epsilon_2) &\le (1-\beta)(1-\epsilon_1')\\
	1-\epsilon_2 &\le 1-\epsilon_2'\\
	w_h &\le g\\
	1-\delta_1 &\le (1-\alpha)(1-\beta)(1-\delta_1')\\
	1-\epsilon &\le 1-\gamma_1 - \gamma_2\\
	(1-\epsilon_1')(1-\epsilon_2') &\le (1-\lambda_1)(1-\lambda_2)(1-\lambda_3)(1-\gamma_1)\\
	g-w_h &\ge w_{on} + w_{off}\\
	\gamma_1 &< 1\\
	1-\epsilon_1' &\le (1-\lambda_1)(1-\lambda_2)\\
	1-\epsilon_2' &\le 1-\lambda_3\\
	1-\delta_1' &\le (1-\eta_1)(1-\eta_2)(1-\eta_3)(1-\eta_4)
\end{align*}

\begin{theorem}
	The children of ``\textbf{ACHIEVE:} Alarm iff no HB provided in $t$ time''
	imply their parent where the parent specification is:
	\begin{align}
	    &\prob{\ge 1-\epsilon}\Box_{\le u}\lnot Alarm\;&\land \label{eq:1.1}\\
	    &\prob{\ge 1}\Box\left[ H_{pres}\implies \prob{\ge 1-\epsilon_1}\Diamond_{\le g}\left( \prob{\ge 1-\epsilon_2}\Box_{\le u}\lnot Alarm \right) \right]\;&\land\label{eq:1.2}\\
	    &\prob{\ge 1}\Box\left[ \lnot H_{pres}\implies \prob{\ge 1-\delta_1}\Diamond_{\le v-w_a}\left( Alarm\lor H_{pres} \right) \right]\label{eq:1.3}
	\end{align}
	and the specifications for the children are:
	\begin{description}
		\item[Subgoal 1:]
		\begin{align}
		    \label{eq:1.4}&\prob{\ge 1}\Box\left[H_{pres}\implies \prob{\ge 1-\beta}\Diamond_{\le w_h}H_{det} \right]\;&\land\\
		    \label{eq:1.5}&\prob{\ge 1}\Box\left[\lnot H_{pres}\implies \prob{\ge 1-\beta}\Diamond_{w_h}\prob{\ge 1-\alpha}\left(\lnot H_{det}\W H_{pres}\right)\right]
		\end{align}
		\item[Subgoal 2:]
		\begin{align}
		    \label{eq:1.6}&\prob{\ge 1-\epsilon}\Box_{\le u}\lnot Alarm\;&\land \\
		    \label{eq:1.7}&\prob{\ge 1}\Box\left[ H_{det}\implies \prob{\ge 1-\epsilon_1'}\Diamond_{\le g-w_h}\left( \prob{\ge 1-\epsilon_2'}\Box_{\le u}\lnot Alarm \right) \right]\;&\land\\
		    \label{eq:1.8}&\prob{\ge 1}\Box\left[ \lnot H_{det}\implies \prob{\ge 1-\delta_1'}\Diamond_{\le v-w_a-w_h}\left( Alarm\lor H_{det} \right) \right]
		\end{align}
	\end{description}
\end{theorem}

\begin{theorem}
	The children of ``\textbf{ACHIEVE:} Heartbeat Detected correctly tracks the presence of Heartbeats within $t-t'$ time'' imply their parent where the parent
	specification is:
	\begin{align*}
	    &\prob{\ge 1}\Box\left[H_{pres}\implies \prob{\ge 1-\beta}\Diamond_{\le w_h}H_{det}\right]\;&\land\\
	    &\prob{\ge 1}\Box\left[\lnot H_{pres}\implies \prob{\ge 1-\beta}\Diamond_{w_h}\prob{\ge 1-\alpha}\left(\lnot H_{det}\W H_{pres}\right)\right]
	\end{align*}
	and the specifications for the children are:
	\begin{description}
		\item[Subgoal 1:]
		\[ \prob{\ge 1}\Box\left[\lnot H_{pres}\implies \prob{\ge 1-\beta}\Diamond_{w_h}\prob{\ge 1-\alpha}\left(\lnot H_{det}\W H_{pres}\right)\right] \]

		\item[Subgoal 2:]
		\[ \prob{\ge 1}\Box\left[H_{pres}\implies \prob{\ge 1-\beta}\Diamond_{\le w_h}H_{det} \right]\]
	\end{description}
\end{theorem}

\begin{theorem}
	The children of ``\textbf{ACHIEVE:} Alarm iff no Heartbeat detected'' imply their parent
	where the parent specification is:
	\begin{align}
	    \label{eq:2.1}&\prob{\ge 1-\epsilon}\Box_{\le u}\lnot Alarm\;&\land \\
	    \label{eq:2.2}&\prob{\ge 1}\Box\left[ H_{det}\implies \prob{\ge 1-\epsilon_1'}\Diamond_{\le g-w_h}\left( \prob{\ge 1-\epsilon_2'}\Box_{\le u}\lnot Alarm \right) \right]\;&\land\\
	    \label{eq:2.3}&\prob{\ge 1}\Box\left[ \lnot H_{det}\implies \prob{\ge 1-\delta_1'}\Diamond_{\le v-w_a-w_h}\left( Alarm\lor H_{det} \right) \right]
	\end{align}
	and the specifications for the children are:
	\begin{description}
		\item[Subgoal 1:]
		\begin{align}
		    \label{eq:2.4}&Reset\;&\land\\
		    \label{eq:2.5}&\prob{\ge 1}\Box\left[H_{det}\implies \prob{\ge 1-\lambda_1}\Diamond_{\le w_{on}} Reset \right]
		\end{align}
		\item[Subgoal 2:]
		\begin{align}
		    \label{eq:2.6}&\prob{\ge 1}\Box\left[Reset\implies \prob{\ge 1-\gamma_1}\Box_{\le u}Th_L \right]\;&\land\\
		    &\begin{aligned}
			    \label{eq:2.8}\prob{\ge 1}\Box[\lnot H_{det}\implies&\prob{\ge 1-\eta_1}\Diamond_{v-w_a-2w_h-w_{th}}\\
			    &\prob{\ge 1-\eta_2}\left(Th_H\W \prob{\ge 1-\eta_3}\Diamond_{\le w_h}H_{det}\right)]
		    \end{aligned}
		\end{align}
		\item[Subgoal 3:]
		\begin{align}
		    \label{eq:2.9}&\prob{\ge 1}\Box\left[Th_L\implies\prob{\ge 1-\lambda_2}\Diamond_{\le w_{off}}\prob{\ge 1-\lambda_3}\Box_{\le u}\lnot Alarm\right]\;&\land\\
		    \label{eq:2.10}&\prob{\ge 1}\Box\left[Th_H\implies \prob{\ge 1-\eta_4}\Diamond_{\le w_{th}}\left(Alarm\;\lor\;\lnot Th_H\right)\right]\;&\land\\
		    \label{eq:2.11}&\prob{\ge 1-\gamma_2}\left(\lnot Alarm\W \lnot Th_L\right)
		\end{align}
	\end{description}
\end{theorem}

\begin{theorem}
	All parent goals of leaves are implied by their children.
\end{theorem}
\begin{proof}
	All leaf goals are broken down by conjunction and trivially imply their parents.
\end{proof}

\subsection{Lemmas}   

For each of the following lemmas, assume that $M$ is a CTMC and $q_0$ is its start state.  For all lemmas except lemma~\ref{lemma1}, we base each implication at an arbitrary state $q$.  For simplicity, we assume that every state in $M$ is reachable from the start state and note that any unreachable states can be removed without affecting the behavior of $M$.

\begin{lemma}\label{lemma1}
\[
\forall q'\ \phi(q') \Longleftrightarrow (\prob{\ge 1}\Box \phi)(q_0)
\]
\end{lemma}
\begin{proof}
This follows trivially from our assumption that all states are reachable.
\end{proof}

\begin{lemma}
\[
\infer{\therefore\prob{\ge\alpha\beta}\Diamond_{\le s+t}\left(\phi\lor\psi\right)}{
	\prob{\ge\alpha}\Diamond_{\le s}\left(\phi\lor\prob{\ge\beta}\Diamond_{\le t}\psi\right)
}
\]
\end{lemma}
\begin{proof}
With probability $\alpha$ at least one of two cases must occur.

\begin{itemize}
\item[]{\bf Case 1:}  $\phi$ is true within time $s$.

In this case, $\phi \lor \psi$ is certainly true within $s + t$ time.

\vspace*{6pt}
\item[]{\bf Case 2:} $\prob{\geq \beta} \Diamond_{\leq t} \psi$ is true in time $s$.

In this case, with probability $\beta$, $M$ reaches a state $q_1$ satisfying $\psi$ in time $t$.  Since $M$ reaches a state satisfying $\prob{\geq \beta} \Diamond_{\leq t} \psi$ in time $s$ and $\psi$ satisfies $\phi \lor \psi$, then with probability $\beta$ $M$ reaches a state satisfying $\phi \lor \psi$ in time $s + t$.
\end{itemize}
In either case $\phi \lor \psi$ must be true in time $s + t$ with probability greater or equal to $\beta$.  Since at least one case must occur with probability $\alpha$, we have $\prob{\geq \alpha \beta} \Diamond_{s+t} (\phi \lor \psi)$.
\end{proof}

\begin{lemma}
\[
\infer{\therefore\prob{\ge\alpha\beta}\Diamond_{\le s}(\psi\lor\theta)}{
	\prob{\ge\alpha}\left(\phi\land\prob{\ge\beta}\Diamond_{\le s}(\lnot\phi\lor\theta)\right)\W\psi
}
\]
\end{lemma}
\begin{proof}
Recall that the statement $\prob{\ge\gamma}\phi_1\W\phi_2$ is satisfied if, with probability at least $\gamma$, the CTMC $M$ follows a path that satisfies $\phi_1$ in every state until it reaches a state that satisfies $\phi_2$.
This is a ``weak until'' operator, so it can be satisfied even if $\phi_2$ is never satisfied, as long as $\phi_1$ holds indefinitely.

Note that the claim trivially holds if $\psi$ is true in state $q$ since it will satisfy $\psi$ within $s$ time with probability 1.
If $\lnot\psi$ holds in state $q$, then with probability $\alpha$, the formulas $\phi$ and $\prob{\ge\beta}\Diamond_{\le s}(\lnot\phi\lor\theta)$ hold until $\psi$ holds.
In this case, the formula $\prob{\ge\beta}\Diamond_{\le s}(\lnot\phi\lor\theta)$ implies that, with probability $\beta$, either $\lnot\phi$ or $\theta$ will eventually hold within $s$ time.
If $\theta$ eventually holds, then $\Diamond_{\le s}(\psi\lor\theta)$ held with probability at least $\alpha\beta$.
If $\lnot\phi$ eventually holds, then $\psi$ must have been satisfied since $\phi$ must hold until $\psi$ is satisfied.
Thus, $\Diamond_{\le s}(\psi\lor\theta)$ held with probability at least $\alpha\beta$ in this case, also.
\end{proof}

\begin{lemma}
    \[
\infer{\therefore\prob{\ge\alpha+\beta-1}\Box_{\le t}\phi}{
\prob{\ge\alpha}\Box_{\le t}\theta
&\prob{\ge\beta}\phi\W\lnot\theta
}
\]
\end{lemma}

\begin{proof}
    Since the two hypotheses may be dependent, the probability that the CTMC $M$ follows a path that satisfies both $\Box_{\le t}\theta$ and $\phi\W\lnot\theta$ is lower-bounded by $\alpha+\beta-1$. It is therefore sufficient to prove that any path satisfying both $\Box_{\le t} \theta$ and $\phi\W\lnot\theta$ must also satisfy $\Box_{\le t} \phi$.

Consider an arbitrary path that satisfies both formulas.  Because of $\Box_{\le t} \theta$, the path cannot reach a state that satisfies $\lnot \theta$ until after time $t$.  Because of $\phi\W\lnot\theta$, the path must satisfy $\phi$ until it satisfies $\lnot \theta$.  Therefore the path must satisfy $\phi$ until after time $t$.

\end{proof}

\begin{lemma}
\[
\infer{\therefore\phi}{
	\prob{>0}\Box\phi
}
\]
\end{lemma}
\begin{proof}
If $\phi$ is not true in state $q$, then $\Box\phi$ fails with probability 1.
Since there is a strictly positive probability that $\Box\phi$ is satisfied, $\phi$ must be true initially.
\end{proof}

\begin{lemma}
\[
\infer{\therefore\prob{\ge\alpha}\Diamond_{\le t}\psi}{
	\prob{\ge\alpha}\Diamond_{\le t}\phi & \forall q'\ (\phi \implies \psi)(q')
}
\]
\end{lemma}
\begin{proof}
With probability $\alpha$, $M$ enters a state $q_1$ satisfying $\phi$ within time $t$.
Since $\phi \implies \psi$ at all states in $M$, $q_1$ must also satisfy $\psi$, and therefore $\prob{\geq \alpha} \Diamond_{\leq t} \psi$.
\end{proof}

The following two trivial relaxation lemmas are used and stated without proof.

\begin{lemma}
\[
\infer{\prob{\geq \beta}\Diamond_{\leq t} \phi}{\alpha \geq \beta & s \leq t & \prob{\geq \alpha} \Diamond_{\leq s} \phi}
\]

\end{lemma}

\begin{lemma}
\[
\infer{\prob{\geq \beta} \Box_{\geq t} \phi}{\alpha \geq \beta & s \geq t & \prob{\geq \alpha} \Box_{\geq s} \phi}
\]
\end{lemma}

\newpage
\section{Verification of Goal Model II: Isabelle}
\isabellestyle{it}
\parindent 0pt\parskip 0.5ex

\begin{isabellebody}%
\setisabellecontext{runtime{\isacharunderscore}fault{\isacharunderscore}detection}%
%
\isamarkuptrue%
\isamarkupsubsection{Locale and Constraints%
}
\isamarkuptrue%
\isadelimtheory
\endisadelimtheory
\isatagtheory
\isacommand{theory}\isamarkupfalse%
\ runtime{\isacharunderscore}fault{\isacharunderscore}detection\isanewline
\ \ \isakeyword{imports}\ Complex{\isacharunderscore}Main\isanewline
\isakeyword{begin}%
\endisatagtheory
{\isafoldtheory}%
\isadelimtheory
\endisadelimtheory
\isanewline
\isanewline
\isacommand{locale}\isamarkupfalse%
\ runtime{\isacharunderscore}fault{\isacharunderscore}detection\ {\isacharequal}\isanewline
\ \ \isakeyword{fixes}\isanewline
\ \ \ \ q{\isadigit{0}}\ {\isacharcolon}{\isacharcolon}\ {\isacharprime}a\ \isakeyword{and}\isanewline
\ \ \ \ Hpres\ {\isacharcolon}{\isacharcolon}\ {\isachardoublequoteopen}{\isacharprime}a\ {\isasymRightarrow}\ bool{\isachardoublequoteclose}\ \isakeyword{and}\isanewline
\ \ \ \ Hdet\ {\isacharcolon}{\isacharcolon}\ {\isachardoublequoteopen}{\isacharprime}a\ {\isasymRightarrow}\ bool{\isachardoublequoteclose}\ \isakeyword{and}\ \ \ \ \isanewline
\ \ \ \ Alarm\ {\isacharcolon}{\isacharcolon}\ {\isachardoublequoteopen}{\isacharprime}a\ {\isasymRightarrow}\ bool{\isachardoublequoteclose}\ \isakeyword{and}\isanewline
\isanewline
\ \ \ \ {\isasymalpha}\ {\isacharcolon}{\isacharcolon}\ real\ \isakeyword{and}\isanewline
\ \ \ \ {\isasymbeta}\ {\isacharcolon}{\isacharcolon}\ real\ \isakeyword{and}\isanewline
\ \ \ \ {\isasymepsilon}\ {\isacharcolon}{\isacharcolon}\ real\ \isakeyword{and}\isanewline
\ \ \ \ {\isasymepsilon}{\isadigit{1}}\ {\isacharcolon}{\isacharcolon}\ real\ \isakeyword{and}\isanewline
\ \ \ \ {\isasymepsilon}{\isadigit{1}}p\ {\isacharcolon}{\isacharcolon}\ real\ \isakeyword{and}\isanewline
\ \ \ \ {\isasymepsilon}{\isadigit{2}}\ {\isacharcolon}{\isacharcolon}\ real\ \isakeyword{and}\isanewline
\ \ \ \ {\isasymepsilon}{\isadigit{2}}p\ {\isacharcolon}{\isacharcolon}\ real\ \isakeyword{and}\isanewline
\ \ \ \ wh\ {\isacharcolon}{\isacharcolon}\ real\ \isakeyword{and}\isanewline
\ \ \ \ g\ {\isacharcolon}{\isacharcolon}\ real\ \isakeyword{and}\isanewline
\ \ \ \ {\isasymdelta}{\isadigit{1}}\ {\isacharcolon}{\isacharcolon}\ real\ \isakeyword{and}\isanewline
\ \ \ \ {\isasymdelta}{\isadigit{1}}p\ {\isacharcolon}{\isacharcolon}\ real\ \isakeyword{and}\isanewline
\ \ \ \ {\isasymdelta}{\isadigit{2}}\ {\isacharcolon}{\isacharcolon}\ real\ \isakeyword{and}\isanewline
\ \ \ \ {\isasymgamma}{\isadigit{1}}\ {\isacharcolon}{\isacharcolon}\ real\ \isakeyword{and}\isanewline
\ \ \ \ {\isasymgamma}{\isadigit{2}}\ {\isacharcolon}{\isacharcolon}\ real\ \isakeyword{and}\isanewline
\ \ \ \ {\isasymeta}{\isadigit{1}}\ {\isacharcolon}{\isacharcolon}\ real\ \isakeyword{and}\isanewline
\ \ \ \ {\isasymeta}{\isadigit{2}}\ {\isacharcolon}{\isacharcolon}\ real\ \isakeyword{and}\isanewline
\ \ \ \ {\isasymeta}{\isadigit{3}}\ {\isacharcolon}{\isacharcolon}\ real\ \isakeyword{and}\isanewline
\ \ \ \ {\isasymeta}{\isadigit{4}}\ {\isacharcolon}{\isacharcolon}\ real\ \isakeyword{and}\isanewline
\ \ \ \ l{\isadigit{1}}\ {\isacharcolon}{\isacharcolon}\ real\ \isakeyword{and}\isanewline
\ \ \ \ l{\isadigit{2}}\ {\isacharcolon}{\isacharcolon}\ real\ \isakeyword{and}\isanewline
\ \ \ \ l{\isadigit{3}}\ {\isacharcolon}{\isacharcolon}\ real\ \isakeyword{and}\isanewline
\ \ \ \ won\ {\isacharcolon}{\isacharcolon}\ real\ \isakeyword{and}\isanewline
\ \ \ \ woff\ {\isacharcolon}{\isacharcolon}\ real\ \isakeyword{and}\isanewline
\isanewline
\ \ \ \ Pdiam\ {\isacharcolon}{\isacharcolon}\ {\isachardoublequoteopen}real\ {\isasymRightarrow}\ real\ {\isasymRightarrow}\ {\isacharparenleft}{\isacharprime}a\ {\isasymRightarrow}\ bool{\isacharparenright}\ {\isasymRightarrow}\ {\isacharprime}a\ {\isasymRightarrow}\ bool{\isachardoublequoteclose}\ \isakeyword{and}\isanewline
\ \ \ \ Pdiam{\isadigit{2}}\ {\isacharcolon}{\isacharcolon}\ {\isachardoublequoteopen}real\ {\isasymRightarrow}\ {\isacharparenleft}{\isacharprime}a\ {\isasymRightarrow}\ bool{\isacharparenright}\ {\isasymRightarrow}\ {\isacharprime}a\ {\isasymRightarrow}\ bool{\isachardoublequoteclose}\ \isakeyword{and}\isanewline
\ \ \ \ Pblock\ {\isacharcolon}{\isacharcolon}\ {\isachardoublequoteopen}real\ {\isasymRightarrow}\ {\isacharparenleft}{\isacharprime}a\ {\isasymRightarrow}\ bool{\isacharparenright}\ {\isasymRightarrow}\ {\isacharprime}a\ {\isasymRightarrow}\ bool{\isachardoublequoteclose}\ \isakeyword{and}\isanewline
\ \ \ \ Pdur\ {\isacharcolon}{\isacharcolon}\ {\isachardoublequoteopen}real\ {\isasymRightarrow}\ real\ {\isasymRightarrow}\ {\isacharparenleft}{\isacharprime}a\ {\isasymRightarrow}\ bool{\isacharparenright}\ {\isasymRightarrow}\ {\isacharprime}a\ {\isasymRightarrow}\ bool{\isachardoublequoteclose}\ \isakeyword{and}\isanewline
\ \ \ \ PW\ {\isacharcolon}{\isacharcolon}\ {\isachardoublequoteopen}real\ {\isasymRightarrow}\ {\isacharparenleft}{\isacharprime}a\ {\isasymRightarrow}\ bool{\isacharparenright}\ {\isasymRightarrow}\ {\isacharparenleft}{\isacharprime}a\ {\isasymRightarrow}\ bool{\isacharparenright}\ {\isasymRightarrow}\ {\isacharprime}a\ {\isasymRightarrow}\ bool{\isachardoublequoteclose}\isanewline
\isanewline
\ \ \isakeyword{assumes}\isanewline
\ \ \ \ constr{\isadigit{1}}{\isacharcolon}\ {\isachardoublequoteopen}{\isacharparenleft}{\isadigit{1}}{\isacharminus}{\isasymepsilon}{\isadigit{1}}{\isacharparenright}\ {\isasymle}\ {\isacharparenleft}{\isadigit{1}}{\isacharminus}{\isasymbeta}{\isacharparenright}{\isacharasterisk}{\isacharparenleft}{\isadigit{1}}{\isacharminus}{\isasymepsilon}{\isadigit{1}}p{\isacharparenright}{\isachardoublequoteclose}\ \isakeyword{and}\isanewline
\ \ \ \ constr{\isadigit{2}}{\isacharcolon}\ {\isachardoublequoteopen}{\isacharparenleft}{\isadigit{1}}{\isacharminus}{\isasymepsilon}{\isadigit{2}}{\isacharparenright}\ {\isasymle}\ {\isacharparenleft}{\isadigit{1}}{\isacharminus}{\isasymepsilon}{\isadigit{2}}p{\isacharparenright}{\isachardoublequoteclose}\ \isakeyword{and}\isanewline
\ \ \ \ constr{\isadigit{3}}{\isacharcolon}\ {\isachardoublequoteopen}wh\ {\isasymle}\ g{\isachardoublequoteclose}\ \isakeyword{and}\isanewline
\ \ \ \ constr{\isadigit{4}}{\isacharcolon}\ {\isachardoublequoteopen}{\isacharparenleft}{\isadigit{1}}{\isacharminus}{\isasymdelta}{\isadigit{1}}{\isacharparenright}\ {\isasymle}\ {\isacharparenleft}{\isadigit{1}}{\isacharminus}{\isasymalpha}{\isacharparenright}{\isacharasterisk}{\isacharparenleft}{\isadigit{1}}{\isacharminus}{\isasymbeta}{\isacharparenright}{\isacharasterisk}{\isacharparenleft}{\isadigit{1}}{\isacharminus}{\isasymdelta}{\isadigit{1}}p{\isacharparenright}{\isachardoublequoteclose}\ \isakeyword{and}\isanewline
\ \ \ \ constr{\isadigit{5}}{\isacharcolon}\ {\isachardoublequoteopen}{\isacharparenleft}{\isadigit{1}}{\isacharminus}{\isasymepsilon}{\isacharparenright}\ {\isasymle}\ {\isacharparenleft}{\isadigit{1}}{\isacharminus}{\isasymgamma}{\isadigit{1}}{\isacharminus}{\isasymgamma}{\isadigit{2}}{\isacharparenright}{\isachardoublequoteclose}\ \isakeyword{and}\isanewline
\ \ \ \ constr{\isadigit{6}}{\isacharcolon}\ {\isachardoublequoteopen}{\isacharparenleft}{\isadigit{1}}{\isacharminus}{\isasymepsilon}{\isadigit{1}}p{\isacharparenright}{\isacharasterisk}{\isacharparenleft}{\isadigit{1}}{\isacharminus}{\isasymepsilon}{\isadigit{2}}p{\isacharparenright}\ {\isasymle}\ {\isacharparenleft}{\isadigit{1}}{\isacharminus}l{\isadigit{1}}{\isacharparenright}{\isacharasterisk}{\isacharparenleft}{\isadigit{1}}{\isacharminus}{\isasymgamma}{\isadigit{1}}{\isacharparenright}{\isacharasterisk}{\isacharparenleft}{\isadigit{1}}{\isacharminus}l{\isadigit{2}}{\isacharparenright}{\isacharasterisk}{\isacharparenleft}{\isadigit{1}}{\isacharminus}l{\isadigit{3}}{\isacharparenright}{\isachardoublequoteclose}\ \isakeyword{and}\isanewline
\ \ \ \ constr{\isadigit{7}}{\isacharcolon}\ {\isachardoublequoteopen}{\isacharparenleft}g{\isacharminus}wh{\isacharparenright}\ {\isasymge}\ won\ {\isacharplus}\ woff{\isachardoublequoteclose}\ \isakeyword{and}\isanewline
\ \ \ \ constr{\isadigit{8}}{\isacharcolon}\ {\isachardoublequoteopen}{\isacharparenleft}{\isadigit{1}}{\isacharminus}{\isasymgamma}{\isadigit{1}}{\isacharparenright}\ {\isachargreater}\ {\isadigit{0}}{\isachardoublequoteclose}\ \isakeyword{and}\isanewline
\ \ \ \ constr{\isadigit{9}}{\isacharcolon}\ {\isachardoublequoteopen}{\isacharparenleft}{\isadigit{1}}{\isacharminus}{\isasymepsilon}{\isadigit{1}}p{\isacharparenright}\ {\isasymle}\ {\isacharparenleft}{\isadigit{1}}{\isacharminus}l{\isadigit{1}}{\isacharparenright}{\isacharasterisk}{\isacharparenleft}{\isadigit{1}}{\isacharminus}l{\isadigit{2}}{\isacharparenright}{\isachardoublequoteclose}\ \isakeyword{and}\isanewline
\ \ \ \ constr{\isadigit{1}}{\isadigit{0}}{\isacharcolon}\ {\isachardoublequoteopen}{\isacharparenleft}{\isadigit{1}}{\isacharminus}{\isasymepsilon}{\isadigit{2}}p{\isacharparenright}\ {\isasymle}\ {\isacharparenleft}{\isadigit{1}}{\isacharminus}l{\isadigit{3}}{\isacharparenright}{\isachardoublequoteclose}\ \isakeyword{and}\isanewline
\ \ \ \ constr{\isadigit{1}}{\isadigit{1}}{\isacharcolon}\ {\isachardoublequoteopen}{\isacharparenleft}{\isadigit{1}}{\isacharminus}{\isasymdelta}{\isadigit{1}}p{\isacharparenright}\ {\isasymle}\ {\isacharparenleft}{\isadigit{1}}{\isacharminus}{\isasymeta}{\isadigit{1}}{\isacharparenright}{\isacharasterisk}{\isacharparenleft}{\isadigit{1}}{\isacharminus}{\isasymeta}{\isadigit{2}}{\isacharparenright}{\isacharasterisk}{\isacharparenleft}{\isadigit{1}}{\isacharminus}{\isasymeta}{\isadigit{3}}{\isacharparenright}{\isacharasterisk}{\isacharparenleft}{\isadigit{1}}{\isacharminus}{\isasymeta}{\isadigit{4}}{\isacharparenright}{\isachardoublequoteclose}\isanewline
\isakeyword{begin}%
\isamarkupsubsection{Logical Connectives for CSL Formulae%
}
\isamarkuptrue%
\isacommand{definition}\isamarkupfalse%
\ cand\ {\isacharcolon}{\isacharcolon}\ {\isachardoublequoteopen}{\isacharparenleft}{\isacharprime}a\ {\isasymRightarrow}\ bool{\isacharparenright}\ {\isasymRightarrow}\ {\isacharparenleft}{\isacharprime}a\ {\isasymRightarrow}\ bool{\isacharparenright}\ {\isasymRightarrow}\ {\isacharparenleft}{\isacharprime}a\ {\isasymRightarrow}\ bool{\isacharparenright}{\isachardoublequoteclose}\ {\isacharparenleft}\isakeyword{infixr}\ {\isachardoublequoteopen}c{\isasymand}{\isachardoublequoteclose}\ {\isadigit{3}}{\isadigit{5}}{\isacharparenright}\isanewline
\ \ \isakeyword{where}\ {\isachardoublequoteopen}{\isacharparenleft}fa\ c{\isasymand}\ fb{\isacharparenright}\ {\isacharequal}\ {\isacharparenleft}{\isasymlambda}x{\isachardot}\ {\isacharparenleft}fa\ x{\isacharparenright}\ {\isasymand}\ {\isacharparenleft}fb\ x{\isacharparenright}{\isacharparenright}{\isachardoublequoteclose}\isanewline
\isanewline
\isacommand{definition}\isamarkupfalse%
\ cor\ {\isacharcolon}{\isacharcolon}\ {\isachardoublequoteopen}{\isacharparenleft}{\isacharprime}a\ {\isasymRightarrow}\ bool{\isacharparenright}\ {\isasymRightarrow}\ {\isacharparenleft}{\isacharprime}a\ {\isasymRightarrow}\ bool{\isacharparenright}\ {\isasymRightarrow}\ {\isacharparenleft}{\isacharprime}a\ {\isasymRightarrow}\ bool{\isacharparenright}{\isachardoublequoteclose}\ {\isacharparenleft}\isakeyword{infixr}\ {\isachardoublequoteopen}c{\isasymor}{\isachardoublequoteclose}\ {\isadigit{3}}{\isadigit{0}}{\isacharparenright}\isanewline
\ \ \isakeyword{where}\ {\isachardoublequoteopen}{\isacharparenleft}fa\ c{\isasymor}\ fb{\isacharparenright}\ {\isacharequal}\ {\isacharparenleft}{\isasymlambda}x{\isachardot}\ {\isacharparenleft}fa\ x{\isacharparenright}\ {\isasymor}\ {\isacharparenleft}fb\ x{\isacharparenright}{\isacharparenright}{\isachardoublequoteclose}\isanewline
\isanewline
\isacommand{definition}\isamarkupfalse%
\ cimp\ {\isacharcolon}{\isacharcolon}\ {\isachardoublequoteopen}{\isacharparenleft}{\isacharprime}a\ {\isasymRightarrow}\ bool{\isacharparenright}\ {\isasymRightarrow}\ {\isacharparenleft}{\isacharprime}a\ {\isasymRightarrow}\ bool{\isacharparenright}\ {\isasymRightarrow}\ {\isacharparenleft}{\isacharprime}a\ {\isasymRightarrow}\ bool{\isacharparenright}{\isachardoublequoteclose}\ {\isacharparenleft}\isakeyword{infixr}\ {\isachardoublequoteopen}c{\isasymlongrightarrow}{\isachardoublequoteclose}\ {\isadigit{2}}{\isadigit{5}}{\isacharparenright}\isanewline
\ \ \isakeyword{where}\ {\isachardoublequoteopen}{\isacharparenleft}fa\ c{\isasymlongrightarrow}\ fb{\isacharparenright}\ {\isacharequal}\ {\isacharparenleft}{\isasymlambda}x{\isachardot}\ {\isacharparenleft}fa\ x{\isacharparenright}\ {\isasymlongrightarrow}\ {\isacharparenleft}fb\ x{\isacharparenright}{\isacharparenright}{\isachardoublequoteclose}\isanewline
\isanewline
\isacommand{definition}\isamarkupfalse%
\ cnot\ {\isacharcolon}{\isacharcolon}\ {\isachardoublequoteopen}{\isacharparenleft}{\isacharprime}a\ {\isasymRightarrow}\ bool{\isacharparenright}\ {\isasymRightarrow}\ {\isacharparenleft}{\isacharprime}a\ {\isasymRightarrow}\ bool{\isacharparenright}{\isachardoublequoteclose}\ {\isacharparenleft}\ {\isachardoublequoteopen}c{\isasymnot}{\isachardoublequoteclose}{\isacharparenright}\isanewline
\ \ \isakeyword{where}\ {\isachardoublequoteopen}c{\isasymnot}f\ {\isacharequal}\ {\isacharparenleft}{\isasymlambda}x{\isachardot}\ {\isasymnot}{\isacharparenleft}f\ x{\isacharparenright}{\isacharparenright}{\isachardoublequoteclose}%
\isamarkupsubsection{Foundational CSL Lemmas%
}
\isamarkuptrue%
\isacommand{lemma}\isamarkupfalse%
\ lemma{\isadigit{3}}{\isacharunderscore}{\isadigit{5}}{\isacharbrackleft}iff{\isacharbrackright}{\isacharcolon}\ {\isachardoublequoteopen}{\isacharparenleft}Pblock\ {\isadigit{1}}\ {\isasymphi}\ q{\isadigit{0}}{\isacharparenright}\ {\isacharequal}\ {\isacharparenleft}{\isasymforall}q{\isachardot}\ {\isasymphi}\ q{\isacharparenright}{\isachardoublequoteclose}\isanewline
\isadelimproof
\ \ %
\endisadelimproof
\isatagproof
\isacommand{sorry}\isamarkupfalse%
\endisatagproof
{\isafoldproof}%
\isadelimproof
\isanewline
\endisadelimproof
\isanewline
\isacommand{lemma}\isamarkupfalse%
\ lemma{\isadigit{3}}{\isacharunderscore}{\isadigit{6}}{\isacharcolon}\ {\isachardoublequoteopen}{\isacharparenleft}Pdiam\ a\ s\ {\isacharparenleft}{\isasymphi}\ c{\isasymor}\ {\isacharparenleft}Pdiam\ b\ t\ {\isasympsi}{\isacharparenright}{\isacharparenright}\ q{\isacharparenright}\ {\isasymLongrightarrow}\ {\isacharparenleft}Pdiam\ {\isacharparenleft}a{\isacharasterisk}b{\isacharparenright}\ {\isacharparenleft}s{\isacharplus}t{\isacharparenright}\ {\isacharparenleft}{\isasymphi}\ c{\isasymor}\ {\isasympsi}{\isacharparenright}\ q{\isacharparenright}{\isachardoublequoteclose}\isanewline
\isadelimproof
\ \ %
\endisadelimproof
\isatagproof
\isacommand{sorry}\isamarkupfalse%
\endisatagproof
{\isafoldproof}%
\isadelimproof
\isanewline
\endisadelimproof
\isanewline
\isacommand{lemma}\isamarkupfalse%
\ lemma{\isadigit{3}}{\isacharunderscore}{\isadigit{7}}{\isacharcolon}\ {\isachardoublequoteopen}{\isacharparenleft}PW\ a\ {\isacharparenleft}{\isasymphi}\ c{\isasymand}\ {\isacharparenleft}Pdiam\ b\ s\ {\isacharparenleft}{\isacharparenleft}c{\isasymnot}{\isasymphi}{\isacharparenright}\ c{\isasymor}\ {\isasymtheta}{\isacharparenright}{\isacharparenright}{\isacharparenright}\ {\isasympsi}{\isacharparenright}\ q\ {\isasymLongrightarrow}\ {\isacharparenleft}Pdiam\ {\isacharparenleft}a{\isacharasterisk}b{\isacharparenright}\ s\ {\isacharparenleft}{\isasympsi}\ c{\isasymor}\ {\isasymtheta}{\isacharparenright}{\isacharparenright}\ q{\isachardoublequoteclose}\isanewline
\isadelimproof
\ \ %
\endisadelimproof
\isatagproof
\isacommand{sorry}\isamarkupfalse%
\endisatagproof
{\isafoldproof}%
\isadelimproof
\isanewline
\endisadelimproof
\isanewline
\isacommand{lemma}\isamarkupfalse%
\ lemma{\isadigit{3}}{\isacharunderscore}{\isadigit{8}}{\isacharcolon}\ {\isachardoublequoteopen}{\isacharparenleft}Pdur\ a\ t\ {\isasymtheta}\ q{\isacharparenright}\ {\isasymLongrightarrow}\ {\isacharparenleft}PW\ b\ {\isasymphi}\ {\isacharparenleft}c{\isasymnot}{\isasymtheta}{\isacharparenright}\ q{\isacharparenright}\ {\isasymLongrightarrow}\ {\isacharparenleft}Pdur\ {\isacharparenleft}a{\isacharplus}b{\isacharminus}{\isadigit{1}}{\isacharparenright}\ t\ {\isasymphi}\ q{\isacharparenright}{\isachardoublequoteclose}\isanewline
\isadelimproof
\ \ %
\endisadelimproof
\isatagproof
\isacommand{sorry}\isamarkupfalse%
\endisatagproof
{\isafoldproof}%
\isadelimproof
\isanewline
\endisadelimproof
\isanewline
\isacommand{lemma}\isamarkupfalse%
\ lemma{\isadigit{3}}{\isacharunderscore}{\isadigit{9}}{\isacharcolon}\ {\isachardoublequoteopen}a\ {\isachargreater}\ {\isadigit{0}}\ {\isasymLongrightarrow}\ Pdur\ a\ t\ {\isasymphi}\ q\ {\isasymLongrightarrow}\ {\isasymphi}\ q{\isachardoublequoteclose}\isanewline
\isadelimproof
\ \ %
\endisadelimproof
\isatagproof
\isacommand{sorry}\isamarkupfalse%
\endisatagproof
{\isafoldproof}%
\isadelimproof
\isanewline
\endisadelimproof
\isanewline
\isacommand{lemma}\isamarkupfalse%
\ lemma{\isadigit{3}}{\isacharunderscore}{\isadigit{1}}{\isadigit{0}}{\isacharcolon}\ {\isachardoublequoteopen}Pdiam\ a\ t\ {\isasymphi}\ q\ {\isasymLongrightarrow}\ {\isacharparenleft}{\isasymforall}qp{\isachardot}\ {\isacharparenleft}{\isasymphi}\ c{\isasymlongrightarrow}\ {\isasympsi}{\isacharparenright}\ qp{\isacharparenright}\ {\isasymLongrightarrow}\ Pdiam\ a\ t\ {\isasympsi}\ q{\isachardoublequoteclose}\isanewline
\isadelimproof
\ \ %
\endisadelimproof
\isatagproof
\isacommand{sorry}\isamarkupfalse%
\endisatagproof
{\isafoldproof}%
\isadelimproof
\endisadelimproof
\isamarkupsubsection{Relaxation Lemmas%
}
\isamarkuptrue%
\isacommand{lemma}\isamarkupfalse%
\ lemma{\isadigit{3}}{\isacharunderscore}{\isadigit{1}}{\isadigit{1}}{\isacharcolon}\ {\isachardoublequoteopen}{\isacharparenleft}a\ {\isasymge}\ b{\isacharparenright}\ {\isasymLongrightarrow}\ {\isacharparenleft}s\ {\isasymle}\ t{\isacharparenright}\ {\isasymLongrightarrow}\ {\isacharparenleft}Pdiam\ a\ s\ {\isasymphi}\ q{\isacharparenright}\ {\isasymLongrightarrow}\ {\isacharparenleft}Pdiam\ b\ t\ {\isasymphi}\ q{\isacharparenright}{\isachardoublequoteclose}\isanewline
\isadelimproof
\ \ %
\endisadelimproof
\isatagproof
\isacommand{sorry}\isamarkupfalse%
\endisatagproof
{\isafoldproof}%
\isadelimproof
\isanewline
\endisadelimproof
\isanewline
\isacommand{lemma}\isamarkupfalse%
\ lemma{\isadigit{3}}{\isacharunderscore}{\isadigit{1}}{\isadigit{2}}{\isacharcolon}\ {\isachardoublequoteopen}{\isacharparenleft}a\ {\isasymge}\ b{\isacharparenright}\ {\isasymLongrightarrow}\ {\isacharparenleft}s\ {\isasymge}\ t{\isacharparenright}\ {\isasymLongrightarrow}\ {\isacharparenleft}Pdur\ a\ s\ {\isasymphi}\ q{\isacharparenright}\ {\isasymLongrightarrow}\ {\isacharparenleft}Pdur\ b\ t\ {\isasymphi}\ q{\isacharparenright}{\isachardoublequoteclose}\isanewline
\isadelimproof
\ \ %
\endisadelimproof
\isatagproof
\isacommand{sorry}\isamarkupfalse%
\endisatagproof
{\isafoldproof}%
\isadelimproof
\endisadelimproof
\isamarkupsubsection{Additional Lemmas%
}
\isamarkuptrue%
\isacommand{lemma}\isamarkupfalse%
\ andapplication{\isacharcolon}\ \isanewline
\ \ \isakeyword{fixes}\ f\ {\isacharcolon}{\isacharcolon}\ {\isachardoublequoteopen}{\isacharparenleft}{\isacharprime}a\ {\isasymRightarrow}\ bool{\isacharparenright}\ {\isasymRightarrow}\ {\isacharprime}a\ {\isasymRightarrow}\ bool{\isachardoublequoteclose}\isanewline
\ \ \isakeyword{assumes}\ asm{\isadigit{1}}{\isacharcolon}\ {\isachardoublequoteopen}f\ {\isasymphi}\ q{\isachardoublequoteclose}\isanewline
\ \ \isakeyword{assumes}\ asm{\isadigit{2}}{\isacharcolon}\ {\isachardoublequoteopen}Pblock\ {\isadigit{1}}\ {\isasympsi}\ q{\isadigit{0}}{\isachardoublequoteclose}\isanewline
\ \ \isakeyword{shows}\ {\isachardoublequoteopen}f\ {\isacharparenleft}{\isasymphi}\ c{\isasymand}\ {\isasympsi}{\isacharparenright}\ q{\isachardoublequoteclose}\ {\isacharparenleft}\isakeyword{is}\ {\isachardoublequoteopen}f\ {\isacharquery}andterm\ q{\isachardoublequoteclose}{\isacharparenright}\isanewline
\isadelimproof
\endisadelimproof
\isatagproof
\isacommand{proof}\isamarkupfalse%
\ {\isacharminus}\isanewline
\ \ \isacommand{have}\isamarkupfalse%
\ {\isachardoublequoteopen}{\isasympsi}\ q{\isachardoublequoteclose}\ \isacommand{using}\isamarkupfalse%
\ lemma{\isadigit{3}}{\isacharunderscore}{\isadigit{5}}\ asm{\isadigit{2}}\ \isacommand{by}\isamarkupfalse%
\ auto\isanewline
\ \ \isacommand{then}\isamarkupfalse%
\ \isacommand{have}\isamarkupfalse%
\ {\isachardoublequoteopen}{\isacharparenleft}{\isasymphi}\ c{\isasymand}\ {\isasympsi}{\isacharparenright}\ {\isacharequal}\ {\isacharparenleft}{\isasymlambda}x{\isachardot}\ {\isasymphi}\ x\ {\isasymand}\ True{\isacharparenright}{\isachardoublequoteclose}\ \isacommand{using}\isamarkupfalse%
\ asm{\isadigit{2}}\ cand{\isacharunderscore}def\ lemma{\isadigit{3}}{\isacharunderscore}{\isadigit{5}}\ \isacommand{by}\isamarkupfalse%
\ simp\isanewline
\ \ \isacommand{then}\isamarkupfalse%
\ \isacommand{show}\isamarkupfalse%
\ {\isacharquery}thesis\ \isacommand{using}\isamarkupfalse%
\ asm{\isadigit{1}}\ \isacommand{by}\isamarkupfalse%
\ auto\isanewline
\isacommand{qed}\isamarkupfalse%
\endisatagproof
{\isafoldproof}%
\isadelimproof
\ \ \ \ \ \ \ \ \ \ \ \ \ \ \ \ \ \ \ \ \ \ \ \ \ \ \ \ \ \isanewline
\endisadelimproof
\isanewline
\isacommand{lemma}\isamarkupfalse%
\ diamcombine{\isacharcolon}\ {\isachardoublequoteopen}Pdiam\ a\ s\ {\isacharparenleft}Pdiam\ b\ t\ {\isasymphi}{\isacharparenright}\ q\ {\isasymLongrightarrow}\ Pdiam\ {\isacharparenleft}a\ {\isacharasterisk}\ b{\isacharparenright}\ {\isacharparenleft}s\ {\isacharplus}\ t{\isacharparenright}\ {\isasymphi}\ q{\isachardoublequoteclose}\isanewline
\isadelimproof
\ \ %
\endisadelimproof
\isatagproof
\isacommand{using}\isamarkupfalse%
\ lemma{\isadigit{3}}{\isacharunderscore}{\isadigit{6}}{\isacharbrackleft}of\ a\ s\ {\isachardoublequoteopen}{\isasymlambda}x{\isachardot}\ False{\isachardoublequoteclose}\ b\ t\ {\isasymphi}\ q{\isacharbrackright}\ cor{\isacharunderscore}def\ \isacommand{by}\isamarkupfalse%
\ simp%
\endisatagproof
{\isafoldproof}%
\isadelimproof
\isanewline
\endisadelimproof
\isanewline
\isacommand{lemma}\isamarkupfalse%
\ andimplies{\isacharcolon}\ {\isachardoublequoteopen}{\isacharparenleft}fa\ c{\isasymand}\ {\isacharparenleft}fa\ c{\isasymlongrightarrow}\ fb{\isacharparenright}{\isacharparenright}\ {\isacharequal}\ {\isacharparenleft}fa\ c{\isasymand}\ fb{\isacharparenright}{\isachardoublequoteclose}\ {\isacharparenleft}\isakeyword{is}\ {\isachardoublequoteopen}{\isacharquery}lhs\ {\isacharequal}\ {\isacharquery}rhs{\isachardoublequoteclose}{\isacharparenright}\isanewline
\isadelimproof
\endisadelimproof
\isatagproof
\isacommand{proof}\isamarkupfalse%
\ {\isacharminus}\isanewline
\ \ \isacommand{have}\isamarkupfalse%
\ {\isachardoublequoteopen}{\isacharquery}lhs\ {\isacharequal}\ {\isacharparenleft}{\isasymlambda}x{\isachardot}\ {\isacharparenleft}fa\ x{\isacharparenright}\ {\isasymand}\ {\isacharparenleft}fa\ x\ {\isasymlongrightarrow}\ fb\ x{\isacharparenright}{\isacharparenright}{\isachardoublequoteclose}\ \isacommand{by}\isamarkupfalse%
\ {\isacharparenleft}simp\ add{\isacharcolon}\ cand{\isacharunderscore}def\ cimp{\isacharunderscore}def{\isacharparenright}\isanewline
\ \ \isacommand{then}\isamarkupfalse%
\ \isacommand{have}\isamarkupfalse%
\ {\isachardoublequoteopen}{\isacharquery}lhs\ {\isacharequal}\ {\isacharparenleft}{\isasymlambda}x{\isachardot}\ {\isacharparenleft}fa\ x{\isacharparenright}\ {\isasymand}\ {\isacharparenleft}fb\ x{\isacharparenright}{\isacharparenright}{\isachardoublequoteclose}\ \isacommand{by}\isamarkupfalse%
\ auto\isanewline
\ \ \isacommand{then}\isamarkupfalse%
\ \isacommand{show}\isamarkupfalse%
\ {\isachardoublequoteopen}{\isacharquery}lhs\ {\isacharequal}\ {\isacharparenleft}fa\ c{\isasymand}\ fb{\isacharparenright}{\isachardoublequoteclose}\ \isacommand{by}\isamarkupfalse%
\ {\isacharparenleft}simp\ add{\isacharcolon}\ cand{\isacharunderscore}def{\isacharparenright}\isanewline
\isacommand{qed}\isamarkupfalse%
\endisatagproof
{\isafoldproof}%
\isadelimproof
\isanewline
\endisadelimproof
\isanewline
\isacommand{lemma}\isamarkupfalse%
\ notnot{\isacharbrackleft}simp{\isacharbrackright}{\isacharcolon}\ {\isachardoublequoteopen}{\isacharparenleft}c{\isasymnot}{\isacharparenleft}c{\isasymnot}{\isasymphi}{\isacharparenright}{\isacharparenright}\ {\isacharequal}\ {\isasymphi}{\isachardoublequoteclose}\isanewline
\isadelimproof
\ \ %
\endisadelimproof
\isatagproof
\isacommand{using}\isamarkupfalse%
\ cnot{\isacharunderscore}def\ \isacommand{by}\isamarkupfalse%
\ auto%
\endisatagproof
{\isafoldproof}%
\isadelimproof
\isanewline
\endisadelimproof
\ \ \isanewline
\isacommand{lemma}\isamarkupfalse%
\ orcomm{\isacharbrackleft}iff{\isacharbrackright}{\isacharcolon}\ {\isachardoublequoteopen}{\isacharparenleft}fa\ c{\isasymor}\ fb{\isacharparenright}\ {\isacharequal}\ {\isacharparenleft}fb\ c{\isasymor}\ fa{\isacharparenright}{\isachardoublequoteclose}\isanewline
\isadelimproof
\ \ %
\endisadelimproof
\isatagproof
\isacommand{using}\isamarkupfalse%
\ cor{\isacharunderscore}def\ \isacommand{by}\isamarkupfalse%
\ auto%
\endisatagproof
{\isafoldproof}%
\isadelimproof
\isanewline
\endisadelimproof
\isanewline
\isacommand{lemma}\isamarkupfalse%
\ andcomm{\isacharbrackleft}iff{\isacharbrackright}{\isacharcolon}\ {\isachardoublequoteopen}{\isacharparenleft}fa\ c{\isasymand}\ fb{\isacharparenright}\ {\isacharequal}\ {\isacharparenleft}fb\ c{\isasymand}\ fa{\isacharparenright}{\isachardoublequoteclose}\isanewline
\isadelimproof
\ \ %
\endisadelimproof
\isatagproof
\isacommand{using}\isamarkupfalse%
\ cand{\isacharunderscore}def\ \isacommand{by}\isamarkupfalse%
\ auto%
\endisatagproof
{\isafoldproof}%
\isadelimproof
\isanewline
\endisadelimproof
\isanewline
\isacommand{lemma}\isamarkupfalse%
\ imptrans{\isacharbrackleft}simp{\isacharbrackright}{\isacharcolon}\ {\isachardoublequoteopen}{\isacharparenleft}fa\ c{\isasymlongrightarrow}\ fb{\isacharparenright}\ q\ {\isasymLongrightarrow}\ {\isacharparenleft}fb\ c{\isasymlongrightarrow}\ fc{\isacharparenright}\ q\ {\isasymLongrightarrow}\ {\isacharparenleft}fa\ c{\isasymlongrightarrow}\ fc{\isacharparenright}\ q{\isachardoublequoteclose}\isanewline
\isadelimproof
\ \ %
\endisadelimproof
\isatagproof
\isacommand{using}\isamarkupfalse%
\ cimp{\isacharunderscore}def\ \isacommand{by}\isamarkupfalse%
\ auto%
\endisatagproof
{\isafoldproof}%
\isadelimproof
\endisadelimproof
\isamarkupsubsection{Goal Diagram Implications%
}
\isamarkuptrue%
\isamarkupsubsubsection{Theorem 3.1%
}
\isamarkuptrue%
\isacommand{lemma}\isamarkupfalse%
\ t{\isadigit{3}}dot{\isadigit{1}}{\isacharcolon}\isanewline
\ \ \isakeyword{assumes}\ eq{\isadigit{4}}{\isacharcolon}\ {\isachardoublequoteopen}Pblock\ {\isadigit{1}}\ {\isacharparenleft}Hpres\ c{\isasymlongrightarrow}\ {\isacharparenleft}Pdiam\ {\isacharparenleft}{\isadigit{1}}{\isacharminus}{\isasymbeta}{\isacharparenright}\ wh\ Hdet{\isacharparenright}{\isacharparenright}\ q{\isadigit{0}}{\isachardoublequoteclose}\isanewline
\ \ \isakeyword{assumes}\ eq{\isadigit{5}}{\isacharcolon}\ {\isachardoublequoteopen}Pblock\ {\isadigit{1}}\ {\isacharparenleft}c{\isasymnot}Hpres\ c{\isasymlongrightarrow}\ {\isacharparenleft}Pdiam\ {\isacharparenleft}{\isadigit{1}}{\isacharminus}{\isasymbeta}{\isacharparenright}\ wh\ {\isacharparenleft}PW\ {\isacharparenleft}{\isadigit{1}}{\isacharminus}{\isasymalpha}{\isacharparenright}\ {\isacharparenleft}c{\isasymnot}Hdet{\isacharparenright}\ Hpres{\isacharparenright}{\isacharparenright}{\isacharparenright}\ q{\isadigit{0}}{\isachardoublequoteclose}\isanewline
\ \ \isakeyword{assumes}\ eq{\isadigit{6}}{\isacharcolon}\ {\isachardoublequoteopen}Pdur\ {\isacharparenleft}{\isadigit{1}}{\isacharminus}{\isasymepsilon}{\isacharparenright}\ u\ {\isacharparenleft}c{\isasymnot}Alarm{\isacharparenright}\ q{\isadigit{0}}{\isachardoublequoteclose}\isanewline
\ \ \isakeyword{assumes}\ eq{\isadigit{7}}{\isacharcolon}\ {\isachardoublequoteopen}Pblock\ {\isadigit{1}}\ {\isacharparenleft}Hdet\ c{\isasymlongrightarrow}\ {\isacharparenleft}Pdiam\ {\isacharparenleft}{\isadigit{1}}{\isacharminus}{\isasymepsilon}{\isadigit{1}}p{\isacharparenright}\ {\isacharparenleft}g\ {\isacharminus}\ wh{\isacharparenright}\ {\isacharparenleft}Pdur\ {\isacharparenleft}{\isadigit{1}}{\isacharminus}{\isasymepsilon}{\isadigit{2}}p{\isacharparenright}\ u\ {\isacharparenleft}c{\isasymnot}Alarm{\isacharparenright}{\isacharparenright}{\isacharparenright}{\isacharparenright}\ q{\isadigit{0}}{\isachardoublequoteclose}\ {\isacharparenleft}\isakeyword{is}\ {\isachardoublequoteopen}Pblock\ {\isadigit{1}}\ {\isacharquery}eq{\isadigit{7}}\ q{\isadigit{0}}{\isachardoublequoteclose}{\isacharparenright}\isanewline
\ \ \isakeyword{assumes}\ eq{\isadigit{8}}{\isacharcolon}\ {\isachardoublequoteopen}Pblock\ {\isadigit{1}}\ {\isacharparenleft}c{\isasymnot}Hdet\ c{\isasymlongrightarrow}\ {\isacharparenleft}Pdiam\ {\isacharparenleft}{\isadigit{1}}{\isacharminus}{\isasymdelta}{\isadigit{1}}p{\isacharparenright}\ {\isacharparenleft}v\ {\isacharminus}\ wa\ {\isacharminus}\ wh{\isacharparenright}\ {\isacharparenleft}Alarm\ c{\isasymor}\ Hdet{\isacharparenright}{\isacharparenright}{\isacharparenright}\ q{\isadigit{0}}{\isachardoublequoteclose}\isanewline
\ \ \isakeyword{shows}\ {\isachardoublequoteopen}{\isacharparenleft}{\isacharparenleft}Pdur\ {\isacharparenleft}{\isadigit{1}}{\isacharminus}{\isasymepsilon}{\isacharparenright}\ u\ {\isacharparenleft}c{\isasymnot}Alarm{\isacharparenright}{\isacharparenright}\ c{\isasymand}\isanewline
\ \ \ \ {\isacharparenleft}Pblock\ {\isadigit{1}}\ {\isacharparenleft}Hpres\ c{\isasymlongrightarrow}\ {\isacharparenleft}Pdiam\ {\isacharparenleft}{\isadigit{1}}{\isacharminus}{\isasymepsilon}{\isadigit{1}}{\isacharparenright}\ g\ {\isacharparenleft}Pdur\ {\isacharparenleft}{\isadigit{1}}{\isacharminus}{\isasymepsilon}{\isadigit{2}}{\isacharparenright}\ u\ {\isacharparenleft}c{\isasymnot}Alarm{\isacharparenright}{\isacharparenright}{\isacharparenright}{\isacharparenright}{\isacharparenright}\ c{\isasymand}\isanewline
\ \ \ \ {\isacharparenleft}Pblock\ {\isadigit{1}}\ {\isacharparenleft}c{\isasymnot}Hpres\ c{\isasymlongrightarrow}\ {\isacharparenleft}Pdiam\ {\isacharparenleft}{\isadigit{1}}{\isacharminus}{\isasymdelta}{\isadigit{1}}{\isacharparenright}\ {\isacharparenleft}v\ {\isacharminus}\ wa{\isacharparenright}\ {\isacharparenleft}Alarm\ c{\isasymor}\ Hpres{\isacharparenright}{\isacharparenright}{\isacharparenright}{\isacharparenright}{\isacharparenright}\ q{\isadigit{0}}{\isachardoublequoteclose}\ {\isacharparenleft}\isakeyword{is}\ {\isachardoublequoteopen}{\isacharparenleft}{\isacharquery}goal{\isadigit{1}}\ c{\isasymand}\ {\isacharquery}goal{\isadigit{2}}\ c{\isasymand}\ {\isacharquery}goal{\isadigit{3}}{\isacharparenright}\ q{\isadigit{0}}{\isachardoublequoteclose}{\isacharparenright}\isanewline
\isadelimproof
\endisadelimproof
\isatagproof
\isacommand{proof}\isamarkupfalse%
\ {\isacharminus}\isanewline
\ \ \isacommand{have}\isamarkupfalse%
\ goal{\isadigit{1}}{\isacharcolon}\ {\isachardoublequoteopen}{\isacharquery}goal{\isadigit{1}}\ q{\isadigit{0}}{\isachardoublequoteclose}\ \isacommand{using}\isamarkupfalse%
\ eq{\isadigit{6}}\ \isacommand{by}\isamarkupfalse%
\ blast\isanewline
\isanewline
\ \ \isacommand{have}\isamarkupfalse%
\ {\isachardoublequoteopen}Pblock\ {\isadigit{1}}\ {\isacharparenleft}Hpres\ c{\isasymlongrightarrow}\ {\isacharparenleft}Pdiam\ {\isacharparenleft}{\isadigit{1}}{\isacharminus}{\isasymbeta}{\isacharparenright}\ wh\ {\isacharparenleft}Pdiam\ {\isacharparenleft}{\isadigit{1}}{\isacharminus}{\isasymepsilon}{\isadigit{1}}p{\isacharparenright}\ {\isacharparenleft}g{\isacharminus}wh{\isacharparenright}\ {\isacharparenleft}Pdur\ {\isacharparenleft}{\isadigit{1}}{\isacharminus}{\isasymepsilon}{\isadigit{2}}p{\isacharparenright}\ u\ {\isacharparenleft}c{\isasymnot}Alarm{\isacharparenright}{\isacharparenright}{\isacharparenright}{\isacharparenright}{\isacharparenright}\ q{\isadigit{0}}{\isachardoublequoteclose}\isanewline
\ \ \ \ \isacommand{using}\isamarkupfalse%
\ cimp{\isacharunderscore}def\ lemma{\isadigit{3}}{\isacharunderscore}{\isadigit{1}}{\isadigit{0}}\ eq{\isadigit{4}}\ eq{\isadigit{7}}\ \isacommand{by}\isamarkupfalse%
\ auto\isanewline
\ \ \isacommand{then}\isamarkupfalse%
\ \isacommand{have}\isamarkupfalse%
\ {\isachardoublequoteopen}Pblock\ {\isadigit{1}}\ {\isacharparenleft}Hpres\ c{\isasymlongrightarrow}\ {\isacharparenleft}Pdiam\ {\isacharparenleft}{\isacharparenleft}{\isadigit{1}}{\isacharminus}{\isasymbeta}{\isacharparenright}{\isacharasterisk}{\isacharparenleft}{\isadigit{1}}{\isacharminus}{\isasymepsilon}{\isadigit{1}}p{\isacharparenright}{\isacharparenright}\ g\ {\isacharparenleft}Pdur\ {\isacharparenleft}{\isadigit{1}}{\isacharminus}{\isasymepsilon}{\isadigit{2}}p{\isacharparenright}\ u\ {\isacharparenleft}c{\isasymnot}Alarm{\isacharparenright}{\isacharparenright}{\isacharparenright}{\isacharparenright}\ q{\isadigit{0}}{\isachardoublequoteclose}\isanewline
\ \ \ \ \isacommand{using}\isamarkupfalse%
\ diamcombine\ lemma{\isadigit{3}}{\isacharunderscore}{\isadigit{5}}\ cimp{\isacharunderscore}def\ imptrans\ \isacommand{by}\isamarkupfalse%
\ fastforce\isanewline
\ \ \isacommand{then}\isamarkupfalse%
\ \isacommand{have}\isamarkupfalse%
\ goal{\isadigit{2}}{\isacharcolon}\ {\isachardoublequoteopen}{\isacharquery}goal{\isadigit{2}}\ q{\isadigit{0}}{\isachardoublequoteclose}\ \isacommand{using}\isamarkupfalse%
\ constr{\isadigit{1}}\ constr{\isadigit{2}}\ lemma{\isadigit{3}}{\isacharunderscore}{\isadigit{1}}{\isadigit{1}}\ lemma{\isadigit{3}}{\isacharunderscore}{\isadigit{1}}{\isadigit{2}}\ lemma{\isadigit{3}}{\isacharunderscore}{\isadigit{1}}{\isadigit{0}}\ \isacommand{by}\isamarkupfalse%
\ {\isacharparenleft}smt\ lemma{\isadigit{3}}{\isacharunderscore}{\isadigit{5}}\ cimp{\isacharunderscore}def\ imptrans{\isacharparenright}\isanewline
\isanewline
\isanewline
\ \ \isacommand{have}\isamarkupfalse%
\ {\isachardoublequoteopen}Pblock\ {\isadigit{1}}\ {\isacharparenleft}{\isacharparenleft}c{\isasymnot}Hpres{\isacharparenright}\ c{\isasymlongrightarrow}\ {\isacharparenleft}Pdiam\ {\isacharparenleft}{\isadigit{1}}{\isacharminus}{\isasymbeta}{\isacharparenright}\ wh\ {\isacharparenleft}PW\ {\isacharparenleft}{\isadigit{1}}{\isacharminus}{\isasymalpha}{\isacharparenright}\ {\isacharparenleft}{\isacharparenleft}c{\isasymnot}Hdet{\isacharparenright}\ c{\isasymand}\ {\isacharparenleft}c{\isasymnot}Hdet\ c{\isasymlongrightarrow}\ {\isacharparenleft}Pdiam\ {\isacharparenleft}{\isadigit{1}}{\isacharminus}{\isasymdelta}{\isadigit{1}}p{\isacharparenright}\ {\isacharparenleft}v{\isacharminus}wa{\isacharminus}wh{\isacharparenright}\ {\isacharparenleft}Alarm\ c{\isasymor}\ Hdet{\isacharparenright}{\isacharparenright}{\isacharparenright}{\isacharparenright}\ Hpres{\isacharparenright}{\isacharparenright}{\isacharparenright}\ q{\isadigit{0}}{\isachardoublequoteclose}\isanewline
\ \ \ \ {\isacharparenleft}\isakeyword{is}\ {\isachardoublequoteopen}{\isacharquery}f\ {\isacharparenleft}c{\isasymnot}Hdet\ c{\isasymand}\ {\isacharparenleft}c{\isasymnot}Hdet\ c{\isasymlongrightarrow}\ {\isacharparenleft}Pdiam\ {\isacharparenleft}{\isadigit{1}}{\isacharminus}{\isasymdelta}{\isadigit{1}}p{\isacharparenright}\ {\isacharparenleft}v{\isacharminus}wa{\isacharminus}wh{\isacharparenright}\ {\isacharparenleft}Alarm\ c{\isasymor}\ Hdet{\isacharparenright}{\isacharparenright}{\isacharparenright}{\isacharparenright}\ q{\isadigit{0}}{\isachardoublequoteclose}{\isacharparenright}\isanewline
\ \ \ \ \isacommand{using}\isamarkupfalse%
\ eq{\isadigit{5}}\ eq{\isadigit{8}}\ andapplication{\isacharbrackleft}of\ {\isacharquery}f\ {\isachardoublequoteopen}c{\isasymnot}Hdet{\isachardoublequoteclose}\ q{\isadigit{0}}{\isacharbrackright}\ \isacommand{by}\isamarkupfalse%
\ blast\isanewline
\ \ \isacommand{then}\isamarkupfalse%
\ \isacommand{have}\isamarkupfalse%
\ step{\isadigit{1}}{\isacharcolon}\ {\isachardoublequoteopen}{\isacharquery}f\ {\isacharparenleft}c{\isasymnot}Hdet\ c{\isasymand}\ {\isacharparenleft}Pdiam\ {\isacharparenleft}{\isadigit{1}}{\isacharminus}{\isasymdelta}{\isadigit{1}}p{\isacharparenright}\ {\isacharparenleft}v{\isacharminus}wa{\isacharminus}wh{\isacharparenright}\ {\isacharparenleft}Alarm\ c{\isasymor}\ Hdet{\isacharparenright}{\isacharparenright}{\isacharparenright}\ q{\isadigit{0}}{\isachardoublequoteclose}\ {\isacharparenleft}\isakeyword{is}\ {\isachardoublequoteopen}Pblock\ {\isadigit{1}}\ {\isacharparenleft}c{\isasymnot}Hpres\ c{\isasymlongrightarrow}\ {\isacharquery}i{\isadigit{1}}{\isacharparenright}\ q{\isadigit{0}}{\isachardoublequoteclose}{\isacharparenright}\isanewline
\ \ \ \ \isacommand{using}\isamarkupfalse%
\ andimplies{\isacharbrackleft}of\ {\isachardoublequoteopen}c{\isasymnot}Hdet{\isachardoublequoteclose}\ {\isachardoublequoteopen}Pdiam\ {\isacharparenleft}{\isadigit{1}}{\isacharminus}{\isasymdelta}{\isadigit{1}}p{\isacharparenright}\ {\isacharparenleft}v{\isacharminus}wa{\isacharminus}wh{\isacharparenright}\ {\isacharparenleft}Alarm\ c{\isasymor}\ Hdet{\isacharparenright}{\isachardoublequoteclose}{\isacharbrackright}\ \isacommand{by}\isamarkupfalse%
\ metis\isanewline
\isanewline
\ \ \isacommand{have}\isamarkupfalse%
\ {\isachardoublequoteopen}{\isasymforall}q{\isachardot}\ {\isacharparenleft}{\isacharparenleft}PW\ {\isacharparenleft}{\isadigit{1}}{\isacharminus}{\isasymalpha}{\isacharparenright}\ {\isacharparenleft}c{\isasymnot}Hdet\ c{\isasymand}\ {\isacharparenleft}Pdiam\ {\isacharparenleft}{\isadigit{1}}{\isacharminus}{\isasymdelta}{\isadigit{1}}p{\isacharparenright}\ {\isacharparenleft}v{\isacharminus}wa{\isacharminus}wh{\isacharparenright}\ {\isacharparenleft}Alarm\ c{\isasymor}\ Hdet{\isacharparenright}{\isacharparenright}{\isacharparenright}\ Hpres{\isacharparenright}\ c{\isasymlongrightarrow}\isanewline
\ \ \ \ {\isacharparenleft}Pdiam\ {\isacharparenleft}{\isacharparenleft}{\isadigit{1}}{\isacharminus}{\isasymalpha}{\isacharparenright}{\isacharasterisk}{\isacharparenleft}{\isadigit{1}}{\isacharminus}{\isasymdelta}{\isadigit{1}}p{\isacharparenright}{\isacharparenright}\ {\isacharparenleft}v{\isacharminus}wa{\isacharminus}wh{\isacharparenright}\ {\isacharparenleft}Alarm\ c{\isasymor}\ Hpres{\isacharparenright}{\isacharparenright}{\isacharparenright}\ q{\isachardoublequoteclose}\isanewline
\ \ \ \ \isacommand{using}\isamarkupfalse%
\ lemma{\isadigit{3}}{\isacharunderscore}{\isadigit{7}}{\isacharbrackleft}of\ {\isachardoublequoteopen}{\isadigit{1}}{\isacharminus}{\isasymalpha}{\isachardoublequoteclose}\ {\isachardoublequoteopen}c{\isasymnot}Hdet{\isachardoublequoteclose}\ {\isachardoublequoteopen}{\isadigit{1}}{\isacharminus}{\isasymdelta}{\isadigit{1}}p{\isachardoublequoteclose}\ {\isachardoublequoteopen}v{\isacharminus}wa{\isacharminus}wh{\isachardoublequoteclose}\ Alarm\ Hpres{\isacharbrackright}\ cimp{\isacharunderscore}def\ diff{\isacharunderscore}add{\isacharunderscore}eq{\isacharunderscore}diff{\isacharunderscore}diff{\isacharunderscore}swap\ diff{\isacharunderscore}diff{\isacharunderscore}add\ notnot\ orcomm\ \isacommand{by}\isamarkupfalse%
\ auto\isanewline
\ \ \isacommand{then}\isamarkupfalse%
\ \isacommand{have}\isamarkupfalse%
\ {\isachardoublequoteopen}{\isasymforall}\ q{\isachardot}\ {\isacharparenleft}{\isacharparenleft}Pdiam\ {\isacharparenleft}{\isadigit{1}}{\isacharminus}{\isasymbeta}{\isacharparenright}\ wh\ {\isacharparenleft}PW\ {\isacharparenleft}{\isadigit{1}}{\isacharminus}{\isasymalpha}{\isacharparenright}\ {\isacharparenleft}c{\isasymnot}Hdet\ c{\isasymand}\ {\isacharparenleft}Pdiam\ {\isacharparenleft}{\isadigit{1}}{\isacharminus}{\isasymdelta}{\isadigit{1}}p{\isacharparenright}\ {\isacharparenleft}v{\isacharminus}wa{\isacharminus}wh{\isacharparenright}\ {\isacharparenleft}Alarm\ c{\isasymor}\ Hdet{\isacharparenright}{\isacharparenright}{\isacharparenright}\ Hpres{\isacharparenright}{\isacharparenright}\ c{\isasymlongrightarrow}\isanewline
\ \ \ \ {\isacharparenleft}Pdiam\ {\isacharparenleft}{\isadigit{1}}{\isacharminus}{\isasymbeta}{\isacharparenright}\ wh\ {\isacharparenleft}Pdiam\ {\isacharparenleft}{\isacharparenleft}{\isadigit{1}}{\isacharminus}{\isasymalpha}{\isacharparenright}{\isacharasterisk}{\isacharparenleft}{\isadigit{1}}{\isacharminus}{\isasymdelta}{\isadigit{1}}p{\isacharparenright}{\isacharparenright}\ {\isacharparenleft}v{\isacharminus}wa{\isacharminus}wh{\isacharparenright}\ {\isacharparenleft}Alarm\ c{\isasymor}\ Hpres{\isacharparenright}{\isacharparenright}{\isacharparenright}{\isacharparenright}\ q{\isachardoublequoteclose}\isanewline
\ \ \ \ {\isacharparenleft}\isakeyword{is}\ {\isachardoublequoteopen}{\isasymforall}q{\isachardot}\ {\isacharparenleft}{\isacharparenleft}Pdiam\ {\isacharparenleft}{\isadigit{1}}{\isacharminus}{\isasymbeta}{\isacharparenright}\ wh\ {\isacharquery}fa{\isacharparenright}\ c{\isasymlongrightarrow}\ {\isacharparenleft}Pdiam\ {\isacharparenleft}{\isadigit{1}}{\isacharminus}{\isasymbeta}{\isacharparenright}\ wh\ {\isacharquery}fb{\isacharparenright}{\isacharparenright}\ q{\isachardoublequoteclose}\ \isakeyword{is}\ {\isachardoublequoteopen}{\isasymforall}q{\isachardot}\ {\isacharparenleft}{\isacharquery}left\ c{\isasymlongrightarrow}\ {\isacharquery}middle{\isacharparenright}q{\isachardoublequoteclose}{\isacharparenright}\isanewline
\ \ \ \ \isacommand{using}\isamarkupfalse%
\ lemma{\isadigit{3}}{\isacharunderscore}{\isadigit{1}}{\isadigit{0}}{\isacharbrackleft}of\ {\isachardoublequoteopen}{\isadigit{1}}{\isacharminus}{\isasymbeta}{\isachardoublequoteclose}\ wh\ {\isacharquery}fa\ {\isacharunderscore}\ {\isacharquery}fb{\isacharbrackright}\ cimp{\isacharunderscore}def\ \isacommand{by}\isamarkupfalse%
\ auto\isanewline
\ \ \isacommand{moreover}\isamarkupfalse%
\ \isacommand{have}\isamarkupfalse%
\ {\isachardoublequoteopen}{\isasymforall}q{\isachardot}\ {\isacharparenleft}{\isacharquery}middle\ c{\isasymlongrightarrow}\ {\isacharparenleft}Pdiam\ {\isacharparenleft}{\isadigit{1}}{\isacharminus}{\isasymdelta}{\isadigit{1}}{\isacharparenright}\ {\isacharparenleft}v{\isacharminus}wa{\isacharparenright}\ {\isacharparenleft}Alarm\ c{\isasymor}\ Hpres{\isacharparenright}{\isacharparenright}{\isacharparenright}\ q{\isachardoublequoteclose}\ {\isacharparenleft}\isakeyword{is}\ {\isachardoublequoteopen}{\isasymforall}q{\isachardot}\ {\isacharparenleft}{\isacharquery}middle\ c{\isasymlongrightarrow}\ {\isacharquery}right{\isacharparenright}\ q{\isachardoublequoteclose}{\isacharparenright}\isanewline
\ \ \ \ \isacommand{using}\isamarkupfalse%
\ diamcombine\ cimp{\isacharunderscore}def\ mult{\isachardot}assoc\ lemma{\isadigit{3}}{\isacharunderscore}{\isadigit{1}}{\isadigit{1}}\ \isacommand{by}\isamarkupfalse%
\ {\isacharparenleft}smt\ Groups{\isachardot}mult{\isacharunderscore}ac{\isacharparenleft}{\isadigit{2}}{\isacharparenright}\ add{\isachardot}commute\ constr{\isadigit{4}}\ diff{\isacharunderscore}add{\isacharunderscore}cancel{\isacharparenright}\isanewline
\ \ \isacommand{ultimately}\isamarkupfalse%
\ \isacommand{have}\isamarkupfalse%
\ {\isachardoublequoteopen}{\isasymforall}q{\isachardot}\ {\isacharparenleft}{\isacharquery}left\ c{\isasymlongrightarrow}\ {\isacharquery}right{\isacharparenright}\ q{\isachardoublequoteclose}\isanewline
\ \ \ \ \isacommand{by}\isamarkupfalse%
\ {\isacharparenleft}meson\ diff{\isacharunderscore}diff{\isacharunderscore}add\ imptrans\ mult{\isachardot}commute\ mult{\isachardot}left{\isacharunderscore}commute{\isacharparenright}\isanewline
\ \ \isacommand{then}\isamarkupfalse%
\ \isacommand{have}\isamarkupfalse%
\ {\isachardoublequoteopen}{\isasymforall}q{\isachardot}\ {\isacharparenleft}c{\isasymnot}Hpres\ c{\isasymlongrightarrow}\ {\isacharquery}right{\isacharparenright}\ q{\isachardoublequoteclose}\ \isacommand{using}\isamarkupfalse%
\ step{\isadigit{1}}\ lemma{\isadigit{3}}{\isacharunderscore}{\isadigit{5}}\ \isacommand{by}\isamarkupfalse%
\ {\isacharparenleft}metis\ imptrans{\isacharparenright}\isanewline
\ \ \isacommand{then}\isamarkupfalse%
\ \isacommand{have}\isamarkupfalse%
\ goal{\isadigit{3}}{\isacharcolon}\ {\isachardoublequoteopen}{\isacharquery}goal{\isadigit{3}}\ q{\isadigit{0}}{\isachardoublequoteclose}\ \isacommand{using}\isamarkupfalse%
\ constr{\isadigit{4}}\ lemma{\isadigit{3}}{\isacharunderscore}{\isadigit{5}}\ \isacommand{by}\isamarkupfalse%
\ blast\isanewline
\isanewline
\ \ \isacommand{show}\isamarkupfalse%
\ {\isacharquery}thesis\ \isacommand{using}\isamarkupfalse%
\ goal{\isadigit{1}}\ goal{\isadigit{2}}\ goal{\isadigit{3}}\ cand{\isacharunderscore}def\ \isacommand{by}\isamarkupfalse%
\ auto\isanewline
\isacommand{qed}\isamarkupfalse%
\endisatagproof
{\isafoldproof}%
\isadelimproof
\endisadelimproof
\isamarkupsubsubsection{Theorem 3.2%
}
\isamarkuptrue%
\isacommand{lemma}\isamarkupfalse%
\ t{\isadigit{3}}dot{\isadigit{2}}{\isacharcolon}\isanewline
\ \ \isakeyword{assumes}\ eqa{\isacharcolon}\ {\isachardoublequoteopen}Pblock\ {\isadigit{1}}\ {\isacharparenleft}c{\isasymnot}Hpres\ c{\isasymlongrightarrow}\ {\isacharparenleft}Pdiam\ {\isacharparenleft}{\isadigit{1}}{\isacharminus}{\isasymbeta}{\isacharparenright}\ wh\ {\isacharparenleft}PW\ {\isacharparenleft}{\isadigit{1}}{\isacharminus}{\isasymalpha}{\isacharparenright}\ {\isacharparenleft}c{\isasymnot}Hdet{\isacharparenright}\ Hpres{\isacharparenright}{\isacharparenright}{\isacharparenright}\ q{\isadigit{0}}{\isachardoublequoteclose}\isanewline
\ \ \isakeyword{assumes}\ eqb{\isacharcolon}\ {\isachardoublequoteopen}Pblock\ {\isadigit{1}}\ {\isacharparenleft}Hpres\ c{\isasymlongrightarrow}\ {\isacharparenleft}Pdiam\ {\isacharparenleft}{\isadigit{1}}{\isacharminus}{\isasymbeta}{\isacharparenright}\ wh\ Hdet{\isacharparenright}{\isacharparenright}\ q{\isadigit{0}}{\isachardoublequoteclose}\isanewline
\ \ \isakeyword{shows}\ {\isachardoublequoteopen}{\isacharparenleft}{\isacharparenleft}Pblock\ {\isadigit{1}}\ {\isacharparenleft}c{\isasymnot}Hpres\ c{\isasymlongrightarrow}\ {\isacharparenleft}Pdiam\ {\isacharparenleft}{\isadigit{1}}{\isacharminus}{\isasymbeta}{\isacharparenright}\ wh\ {\isacharparenleft}PW\ {\isacharparenleft}{\isadigit{1}}{\isacharminus}{\isasymalpha}{\isacharparenright}\ {\isacharparenleft}c{\isasymnot}Hdet{\isacharparenright}\ Hpres{\isacharparenright}{\isacharparenright}{\isacharparenright}{\isacharparenright}\ c{\isasymand}\ {\isacharparenleft}Pblock\ {\isadigit{1}}\ {\isacharparenleft}Hpres\ c{\isasymlongrightarrow}\ {\isacharparenleft}Pdiam\ {\isacharparenleft}{\isadigit{1}}{\isacharminus}{\isasymbeta}{\isacharparenright}\ wh\ Hdet{\isacharparenright}{\isacharparenright}{\isacharparenright}{\isacharparenright}\ q{\isadigit{0}}{\isachardoublequoteclose}\isanewline
\isadelimproof
\endisadelimproof
\isatagproof
\isacommand{proof}\isamarkupfalse%
\ {\isacharminus}\isanewline
\ \ \isacommand{show}\isamarkupfalse%
\ {\isacharquery}thesis\ \isacommand{using}\isamarkupfalse%
\ eqa\ eqb\ cand{\isacharunderscore}def\ \isacommand{by}\isamarkupfalse%
\ auto\isanewline
\isacommand{qed}\isamarkupfalse%
\endisatagproof
{\isafoldproof}%
\isadelimproof
\endisadelimproof
\isamarkupsubsubsection{Theorem 3.3%
}
\isamarkuptrue%
\isacommand{lemma}\isamarkupfalse%
\ t{\isadigit{3}}dot{\isadigit{3}}{\isacharcolon}\isanewline
\ \ \isakeyword{assumes}\ eq{\isadigit{1}}{\isadigit{2}}{\isacharcolon}\ {\isachardoublequoteopen}Reset\ q{\isadigit{0}}{\isachardoublequoteclose}\isanewline
\ \ \isakeyword{assumes}\ eq{\isadigit{1}}{\isadigit{3}}{\isacharcolon}\ {\isachardoublequoteopen}Pblock\ {\isadigit{1}}\ {\isacharparenleft}Hdet\ c{\isasymlongrightarrow}\ {\isacharparenleft}Pdiam\ {\isacharparenleft}{\isadigit{1}}{\isacharminus}l{\isadigit{1}}{\isacharparenright}\ won\ Reset{\isacharparenright}{\isacharparenright}\ q{\isadigit{0}}{\isachardoublequoteclose}\ {\isacharparenleft}\isakeyword{is}\ {\isachardoublequoteopen}Pblock\ {\isadigit{1}}\ {\isacharquery}eq{\isadigit{1}}{\isadigit{3}}\ q{\isadigit{0}}{\isachardoublequoteclose}{\isacharparenright}\isanewline
\ \ \isakeyword{assumes}\ eq{\isadigit{1}}{\isadigit{4}}{\isacharcolon}\ {\isachardoublequoteopen}Pblock\ {\isadigit{1}}\ {\isacharparenleft}Reset\ c{\isasymlongrightarrow}\ {\isacharparenleft}Pdur\ {\isacharparenleft}{\isadigit{1}}{\isacharminus}{\isasymgamma}{\isadigit{1}}{\isacharparenright}\ u\ ThL{\isacharparenright}{\isacharparenright}\ q{\isadigit{0}}{\isachardoublequoteclose}\ {\isacharparenleft}\isakeyword{is}\ {\isachardoublequoteopen}Pblock\ {\isadigit{1}}\ {\isacharquery}eq{\isadigit{1}}{\isadigit{4}}\ q{\isadigit{0}}{\isachardoublequoteclose}{\isacharparenright}\isanewline
\ \ \isakeyword{assumes}\ eq{\isadigit{1}}{\isadigit{5}}{\isacharcolon}\ {\isachardoublequoteopen}Pblock\ {\isadigit{1}}\ {\isacharparenleft}c{\isasymnot}Hdet\ c{\isasymlongrightarrow}\ {\isacharparenleft}Pdiam\ {\isacharparenleft}{\isadigit{1}}{\isacharminus}{\isasymeta}{\isadigit{1}}{\isacharparenright}\ {\isacharparenleft}v\ {\isacharminus}\ wa\ {\isacharminus}\ {\isadigit{2}}{\isacharasterisk}wh\ {\isacharminus}\ wth{\isacharparenright}\ {\isacharparenleft}PW\ {\isacharparenleft}{\isadigit{1}}{\isacharminus}{\isasymeta}{\isadigit{2}}{\isacharparenright}\ ThH\ {\isacharparenleft}Pdiam\ {\isacharparenleft}{\isadigit{1}}{\isacharminus}{\isasymeta}{\isadigit{3}}{\isacharparenright}\ wh\ Hdet{\isacharparenright}{\isacharparenright}{\isacharparenright}{\isacharparenright}\ q{\isadigit{0}}{\isachardoublequoteclose}\ {\isacharparenleft}\isakeyword{is}\ {\isachardoublequoteopen}Pblock\ {\isadigit{1}}\ {\isacharquery}eq{\isadigit{1}}{\isadigit{5}}\ q{\isadigit{0}}{\isachardoublequoteclose}{\isacharparenright}\isanewline
\ \ \isakeyword{assumes}\ eq{\isadigit{1}}{\isadigit{6}}{\isacharcolon}\ {\isachardoublequoteopen}Pblock\ {\isadigit{1}}\ {\isacharparenleft}ThL\ c{\isasymlongrightarrow}\ {\isacharparenleft}Pdiam\ {\isacharparenleft}{\isadigit{1}}{\isacharminus}l{\isadigit{2}}{\isacharparenright}\ woff\ {\isacharparenleft}Pdur\ {\isacharparenleft}{\isadigit{1}}{\isacharminus}l{\isadigit{3}}{\isacharparenright}\ u\ {\isacharparenleft}c{\isasymnot}Alarm{\isacharparenright}{\isacharparenright}{\isacharparenright}{\isacharparenright}\ q{\isadigit{0}}{\isachardoublequoteclose}\ {\isacharparenleft}\isakeyword{is}\ {\isachardoublequoteopen}Pblock\ {\isadigit{1}}\ {\isacharquery}eq{\isadigit{1}}{\isadigit{6}}\ q{\isadigit{0}}{\isachardoublequoteclose}{\isacharparenright}\isanewline
\ \ \isakeyword{assumes}\ eq{\isadigit{1}}{\isadigit{7}}{\isacharcolon}\ {\isachardoublequoteopen}Pblock\ {\isadigit{1}}\ {\isacharparenleft}ThH\ c{\isasymlongrightarrow}\ {\isacharparenleft}Pdiam\ {\isacharparenleft}{\isadigit{1}}{\isacharminus}{\isasymeta}{\isadigit{4}}{\isacharparenright}\ wth\ {\isacharparenleft}Alarm\ c{\isasymor}\ c{\isasymnot}ThH{\isacharparenright}{\isacharparenright}{\isacharparenright}\ q{\isadigit{0}}{\isachardoublequoteclose}\ {\isacharparenleft}\isakeyword{is}\ {\isachardoublequoteopen}Pblock\ {\isadigit{1}}\ {\isacharquery}eq{\isadigit{1}}{\isadigit{7}}\ q{\isadigit{0}}{\isachardoublequoteclose}{\isacharparenright}\isanewline
\ \ \isakeyword{assumes}\ eq{\isadigit{1}}{\isadigit{8}}{\isacharcolon}\ {\isachardoublequoteopen}PW\ {\isacharparenleft}{\isadigit{1}}{\isacharminus}{\isasymgamma}{\isadigit{2}}{\isacharparenright}\ {\isacharparenleft}c{\isasymnot}Alarm{\isacharparenright}\ {\isacharparenleft}c{\isasymnot}ThL{\isacharparenright}\ q{\isadigit{0}}{\isachardoublequoteclose}\isanewline
\ \ \isakeyword{shows}\ {\isachardoublequoteopen}{\isacharparenleft}{\isacharparenleft}Pdur\ {\isacharparenleft}{\isadigit{1}}{\isacharminus}{\isasymepsilon}{\isacharparenright}\ u\ {\isacharparenleft}c{\isasymnot}Alarm{\isacharparenright}{\isacharparenright}\ c{\isasymand}\isanewline
\ \ \ \ {\isacharparenleft}Pblock\ {\isadigit{1}}\ {\isacharparenleft}Hdet\ c{\isasymlongrightarrow}\ {\isacharparenleft}Pdiam\ {\isacharparenleft}{\isadigit{1}}{\isacharminus}{\isasymepsilon}{\isadigit{1}}p{\isacharparenright}\ {\isacharparenleft}g\ {\isacharminus}\ wh{\isacharparenright}\ {\isacharparenleft}Pdur\ {\isacharparenleft}{\isadigit{1}}{\isacharminus}{\isasymepsilon}{\isadigit{2}}p{\isacharparenright}\ u\ {\isacharparenleft}c{\isasymnot}Alarm{\isacharparenright}{\isacharparenright}{\isacharparenright}{\isacharparenright}{\isacharparenright}\ c{\isasymand}\isanewline
\ \ \ \ {\isacharparenleft}Pblock\ {\isadigit{1}}\ {\isacharparenleft}c{\isasymnot}Hdet\ c{\isasymlongrightarrow}\ Pdiam\ {\isacharparenleft}{\isadigit{1}}{\isacharminus}{\isasymdelta}{\isadigit{1}}p{\isacharparenright}\ {\isacharparenleft}v\ {\isacharminus}\ wa\ {\isacharminus}\ wh{\isacharparenright}\ {\isacharparenleft}Alarm\ c{\isasymor}\ Hdet{\isacharparenright}{\isacharparenright}{\isacharparenright}{\isacharparenright}\ q{\isadigit{0}}{\isachardoublequoteclose}\ {\isacharparenleft}\isakeyword{is}\ {\isachardoublequoteopen}{\isacharparenleft}{\isacharquery}goal{\isadigit{1}}\ c{\isasymand}\ {\isacharquery}goal{\isadigit{2}}\ c{\isasymand}\ {\isacharquery}goal{\isadigit{3}}{\isacharparenright}\ q{\isadigit{0}}{\isachardoublequoteclose}{\isacharparenright}\isanewline
\isadelimproof
\endisadelimproof
\isatagproof
\isacommand{proof}\isamarkupfalse%
\ {\isacharminus}\isanewline
\ \ \isacommand{have}\isamarkupfalse%
\ {\isachardoublequoteopen}Pdur\ {\isacharparenleft}{\isadigit{1}}{\isacharminus}{\isasymgamma}{\isadigit{1}}{\isacharparenright}\ u\ ThL\ q{\isadigit{0}}{\isachardoublequoteclose}\ \isacommand{using}\isamarkupfalse%
\ cimp{\isacharunderscore}def\ eq{\isadigit{1}}{\isadigit{2}}\ eq{\isadigit{1}}{\isadigit{4}}\ \isacommand{by}\isamarkupfalse%
\ simp\isanewline
\ \ \isacommand{then}\isamarkupfalse%
\ \isacommand{have}\isamarkupfalse%
\ {\isachardoublequoteopen}Pdur\ {\isacharparenleft}{\isadigit{1}}{\isacharminus}{\isasymgamma}{\isadigit{1}}{\isacharminus}{\isasymgamma}{\isadigit{2}}{\isacharparenright}\ u\ {\isacharparenleft}c{\isasymnot}Alarm{\isacharparenright}\ q{\isadigit{0}}{\isachardoublequoteclose}\isanewline
\ \ \ \ \isacommand{using}\isamarkupfalse%
\ eq{\isadigit{1}}{\isadigit{8}}\ lemma{\isadigit{3}}{\isacharunderscore}{\isadigit{8}}\ \isacommand{by}\isamarkupfalse%
\ {\isacharparenleft}smt\ add{\isachardot}commute\ add{\isacharunderscore}diff{\isacharunderscore}cancel{\isacharunderscore}left\ add{\isacharunderscore}diff{\isacharunderscore}cancel{\isacharunderscore}right\ diff{\isacharunderscore}add{\isacharunderscore}cancel{\isacharparenright}\isanewline
\ \ \isacommand{then}\isamarkupfalse%
\ \isacommand{have}\isamarkupfalse%
\ goal{\isadigit{1}}{\isacharcolon}\ {\isachardoublequoteopen}{\isacharquery}goal{\isadigit{1}}\ q{\isadigit{0}}{\isachardoublequoteclose}\ \isacommand{using}\isamarkupfalse%
\ lemma{\isadigit{3}}{\isacharunderscore}{\isadigit{1}}{\isadigit{2}}\ constr{\isadigit{5}}\ \isacommand{by}\isamarkupfalse%
\ blast\isanewline
\isanewline
\ \ \isacommand{have}\isamarkupfalse%
\ {\isachardoublequoteopen}Pblock\ {\isadigit{1}}\ {\isacharparenleft}Hdet\ c{\isasymlongrightarrow}\ {\isacharparenleft}Pdiam\ \ {\isacharparenleft}{\isadigit{1}}{\isacharminus}l{\isadigit{1}}{\isacharparenright}\ won\ ThL{\isacharparenright}{\isacharparenright}\ q{\isadigit{0}}{\isachardoublequoteclose}\isanewline
\ \ \ \ \isacommand{using}\isamarkupfalse%
\ lemma{\isadigit{3}}{\isacharunderscore}{\isadigit{9}}\ lemma{\isadigit{3}}{\isacharunderscore}{\isadigit{1}}{\isadigit{0}}\ constr{\isadigit{8}}\ lemma{\isadigit{3}}{\isacharunderscore}{\isadigit{5}}\ cimp{\isacharunderscore}def\ eq{\isadigit{1}}{\isadigit{3}}\ eq{\isadigit{1}}{\isadigit{4}}\ imptrans\ \isacommand{by}\isamarkupfalse%
\ smt\isanewline
\ \ \isacommand{then}\isamarkupfalse%
\ \isacommand{have}\isamarkupfalse%
\ {\isachardoublequoteopen}Pblock\ {\isadigit{1}}\ {\isacharparenleft}Hdet\ c{\isasymlongrightarrow}\ {\isacharparenleft}Pdiam\ {\isacharparenleft}{\isadigit{1}}{\isacharminus}l{\isadigit{1}}{\isacharparenright}\ won\ {\isacharparenleft}Pdiam\ {\isacharparenleft}{\isadigit{1}}{\isacharminus}l{\isadigit{2}}{\isacharparenright}\ woff\ {\isacharparenleft}Pdur\ {\isacharparenleft}{\isadigit{1}}{\isacharminus}l{\isadigit{3}}{\isacharparenright}\ u\ {\isacharparenleft}c{\isasymnot}Alarm{\isacharparenright}{\isacharparenright}{\isacharparenright}{\isacharparenright}{\isacharparenright}\ q{\isadigit{0}}{\isachardoublequoteclose}\isanewline
\ \ \ \ \isacommand{using}\isamarkupfalse%
\ andapplication\ eq{\isadigit{1}}{\isadigit{6}}\ lemma{\isadigit{3}}{\isacharunderscore}{\isadigit{1}}{\isadigit{0}}\ lemma{\isadigit{3}}{\isacharunderscore}{\isadigit{5}}\ cimp{\isacharunderscore}def\ \isacommand{by}\isamarkupfalse%
\ smt\isanewline
\ \ \isacommand{moreover}\isamarkupfalse%
\ \isacommand{have}\isamarkupfalse%
\ {\isachardoublequoteopen}{\isasymforall}q{\isachardot}\ {\isacharparenleft}{\isacharparenleft}Pdiam\ {\isacharparenleft}{\isadigit{1}}{\isacharminus}l{\isadigit{1}}{\isacharparenright}\ won\ {\isacharparenleft}Pdiam\ {\isacharparenleft}{\isadigit{1}}{\isacharminus}l{\isadigit{2}}{\isacharparenright}\ woff\ {\isacharparenleft}Pdur\ {\isacharparenleft}{\isadigit{1}}{\isacharminus}l{\isadigit{3}}{\isacharparenright}\ u\ {\isacharparenleft}c{\isasymnot}Alarm{\isacharparenright}{\isacharparenright}{\isacharparenright}{\isacharparenright}\ c{\isasymlongrightarrow}\isanewline
\ \ \ \ {\isacharparenleft}Pdiam\ {\isacharparenleft}{\isadigit{1}}{\isacharminus}{\isasymepsilon}{\isadigit{1}}p{\isacharparenright}\ {\isacharparenleft}g{\isacharminus}wh{\isacharparenright}\ {\isacharparenleft}Pdur\ {\isacharparenleft}{\isadigit{1}}{\isacharminus}{\isasymepsilon}{\isadigit{2}}p{\isacharparenright}\ u\ {\isacharparenleft}c{\isasymnot}Alarm{\isacharparenright}{\isacharparenright}{\isacharparenright}{\isacharparenright}\ q{\isachardoublequoteclose}\isanewline
\ \ \ \ \isacommand{using}\isamarkupfalse%
\ diamcombine\ lemma{\isadigit{3}}{\isacharunderscore}{\isadigit{1}}{\isadigit{1}}\ lemma{\isadigit{3}}{\isacharunderscore}{\isadigit{1}}{\isadigit{2}}\ constr{\isadigit{7}}\ constr{\isadigit{9}}\ constr{\isadigit{1}}{\isadigit{0}}\ cimp{\isacharunderscore}def\ lemma{\isadigit{3}}{\isacharunderscore}{\isadigit{1}}{\isadigit{0}}\ \isacommand{by}\isamarkupfalse%
\ smt\isanewline
\ \ \isacommand{ultimately}\isamarkupfalse%
\ \isacommand{have}\isamarkupfalse%
\ goal{\isadigit{2}}{\isacharcolon}\ {\isachardoublequoteopen}{\isacharquery}goal{\isadigit{2}}\ q{\isadigit{0}}{\isachardoublequoteclose}\ \isacommand{using}\isamarkupfalse%
\ lemma{\isadigit{3}}{\isacharunderscore}{\isadigit{5}}\ cimp{\isacharunderscore}def\ \isacommand{by}\isamarkupfalse%
\ simp\isanewline
\isanewline
\ \ \isacommand{have}\isamarkupfalse%
\ {\isachardoublequoteopen}Pblock\ {\isadigit{1}}\ {\isacharparenleft}c{\isasymnot}Hdet\ c{\isasymlongrightarrow}\ {\isacharparenleft}Pdiam\ {\isacharparenleft}{\isadigit{1}}{\isacharminus}{\isasymeta}{\isadigit{1}}{\isacharparenright}\ {\isacharparenleft}v{\isacharminus}wa{\isacharminus}{\isadigit{2}}{\isacharasterisk}wh{\isacharminus}wth{\isacharparenright}\ {\isacharparenleft}PW\ {\isacharparenleft}{\isadigit{1}}{\isacharminus}{\isasymeta}{\isadigit{2}}{\isacharparenright}\ {\isacharparenleft}ThH\ c{\isasymand}\isanewline
\ \ \ \ {\isacharparenleft}ThH\ c{\isasymlongrightarrow}\ {\isacharparenleft}Pdiam\ {\isacharparenleft}{\isadigit{1}}{\isacharminus}{\isasymeta}{\isadigit{4}}{\isacharparenright}\ wth\ {\isacharparenleft}Alarm\ c{\isasymor}\ c{\isasymnot}ThH{\isacharparenright}{\isacharparenright}{\isacharparenright}{\isacharparenright}\ {\isacharparenleft}Pdiam\ {\isacharparenleft}{\isadigit{1}}{\isacharminus}{\isasymeta}{\isadigit{3}}{\isacharparenright}\ wh\ Hdet{\isacharparenright}{\isacharparenright}{\isacharparenright}{\isacharparenright}\ q{\isadigit{0}}{\isachardoublequoteclose}\isanewline
\ \ \ \ {\isacharparenleft}\isakeyword{is}\ {\isachardoublequoteopen}{\isacharquery}f\ {\isacharparenleft}ThH\ c{\isasymand}\ {\isacharparenleft}ThH\ c{\isasymlongrightarrow}\ {\isacharparenleft}Pdiam\ {\isacharparenleft}{\isadigit{1}}{\isacharminus}{\isasymeta}{\isadigit{4}}{\isacharparenright}\ wth\ {\isacharparenleft}Alarm\ c{\isasymor}\ c{\isasymnot}ThH{\isacharparenright}{\isacharparenright}{\isacharparenright}{\isacharparenright}\ q{\isadigit{0}}{\isachardoublequoteclose}{\isacharparenright}\isanewline
\ \ \ \ \isacommand{using}\isamarkupfalse%
\ eq{\isadigit{1}}{\isadigit{5}}\ eq{\isadigit{1}}{\isadigit{7}}\ andapplication{\isacharbrackleft}of\ {\isacharquery}f\ ThH\ q{\isadigit{0}}{\isacharbrackright}\ \isacommand{by}\isamarkupfalse%
\ blast\isanewline
\ \ \isacommand{then}\isamarkupfalse%
\ \isacommand{have}\isamarkupfalse%
\ step{\isadigit{1}}{\isacharcolon}\ {\isachardoublequoteopen}{\isacharquery}f\ {\isacharparenleft}ThH\ c{\isasymand}\ {\isacharparenleft}Pdiam\ {\isacharparenleft}{\isadigit{1}}{\isacharminus}{\isasymeta}{\isadigit{4}}{\isacharparenright}\ wth\ {\isacharparenleft}Alarm\ c{\isasymor}\ c{\isasymnot}ThH{\isacharparenright}{\isacharparenright}{\isacharparenright}\ q{\isadigit{0}}{\isachardoublequoteclose}\isanewline
\ \ \ \ \ {\isacharparenleft}\isakeyword{is}\ {\isachardoublequoteopen}Pblock\ {\isadigit{1}}\ {\isacharparenleft}c{\isasymnot}Hdet\ c{\isasymlongrightarrow}\ {\isacharparenleft}Pdiam\ {\isacharparenleft}{\isadigit{1}}{\isacharminus}{\isasymeta}{\isadigit{1}}{\isacharparenright}\ {\isacharparenleft}v{\isacharminus}wa{\isacharminus}{\isadigit{2}}{\isacharasterisk}wh{\isacharminus}wth{\isacharparenright}\ {\isacharquery}ph{\isacharparenright}{\isacharparenright}\ q{\isadigit{0}}{\isachardoublequoteclose}{\isacharparenright}\isanewline
\ \ \ \ \isacommand{using}\isamarkupfalse%
\ andimplies\ \isacommand{by}\isamarkupfalse%
\ metis\isanewline
\ \ \isacommand{then}\isamarkupfalse%
\ \isacommand{have}\isamarkupfalse%
\ {\isachardoublequoteopen}Pblock\ {\isadigit{1}}\ {\isacharparenleft}c{\isasymnot}Hdet\ c{\isasymlongrightarrow}\ {\isacharparenleft}Pdiam\ {\isacharparenleft}{\isadigit{1}}{\isacharminus}{\isasymeta}{\isadigit{1}}{\isacharparenright}\ {\isacharparenleft}v{\isacharminus}wa{\isacharminus}{\isadigit{2}}{\isacharasterisk}wh{\isacharminus}wth{\isacharparenright}\isanewline
\ \ \ \ {\isacharparenleft}Pdiam\ {\isacharparenleft}{\isacharparenleft}{\isadigit{1}}{\isacharminus}{\isasymeta}{\isadigit{2}}{\isacharparenright}{\isacharasterisk}{\isacharparenleft}{\isadigit{1}}{\isacharminus}{\isasymeta}{\isadigit{4}}{\isacharparenright}{\isacharparenright}\ wth\ {\isacharparenleft}Alarm\ c{\isasymor}\ {\isacharparenleft}Pdiam\ {\isacharparenleft}{\isadigit{1}}{\isacharminus}{\isasymeta}{\isadigit{3}}{\isacharparenright}\ wh\ Hdet{\isacharparenright}{\isacharparenright}{\isacharparenright}{\isacharparenright}{\isacharparenright}\ q{\isadigit{0}}{\isachardoublequoteclose}\isanewline
\ \ \ \ \isacommand{using}\isamarkupfalse%
\ lemma{\isadigit{3}}{\isacharunderscore}{\isadigit{7}}\ lemma{\isadigit{3}}{\isacharunderscore}{\isadigit{1}}{\isadigit{0}}\ cimp{\isacharunderscore}def\ \isacommand{by}\isamarkupfalse%
\ {\isacharparenleft}smt\ lemma{\isadigit{3}}{\isacharunderscore}{\isadigit{5}}\ orcomm{\isacharparenright}\isanewline
\ \ \isacommand{then}\isamarkupfalse%
\ \isacommand{have}\isamarkupfalse%
\ {\isachardoublequoteopen}Pblock\ {\isadigit{1}}\ {\isacharparenleft}c{\isasymnot}Hdet\ c{\isasymlongrightarrow}\ {\isacharparenleft}Pdiam\ {\isacharparenleft}{\isacharparenleft}{\isadigit{1}}{\isacharminus}{\isasymeta}{\isadigit{1}}{\isacharparenright}{\isacharasterisk}{\isacharparenleft}{\isadigit{1}}{\isacharminus}{\isasymeta}{\isadigit{2}}{\isacharparenright}{\isacharasterisk}{\isacharparenleft}{\isadigit{1}}{\isacharminus}{\isasymeta}{\isadigit{4}}{\isacharparenright}{\isacharparenright}\ {\isacharparenleft}v{\isacharminus}wa{\isacharminus}{\isadigit{2}}{\isacharasterisk}wh{\isacharparenright}\ {\isacharparenleft}Alarm\ c{\isasymor}\ Pdiam\ {\isacharparenleft}{\isadigit{1}}{\isacharminus}{\isasymeta}{\isadigit{3}}{\isacharparenright}\ wh\ Hdet{\isacharparenright}{\isacharparenright}{\isacharparenright}\ q{\isadigit{0}}{\isachardoublequoteclose}\isanewline
\ \ \ \ \isacommand{using}\isamarkupfalse%
\ diamcombine\ cimp{\isacharunderscore}def\ \isacommand{by}\isamarkupfalse%
\ {\isacharparenleft}smt\ lemma{\isadigit{3}}{\isacharunderscore}{\isadigit{5}}\ imptrans\ mult{\isachardot}assoc\ orcomm{\isacharparenright}\isanewline
\ \ \isacommand{then}\isamarkupfalse%
\ \isacommand{have}\isamarkupfalse%
\ {\isachardoublequoteopen}Pblock\ {\isadigit{1}}\ {\isacharparenleft}c{\isasymnot}Hdet\ c{\isasymlongrightarrow}\ Pdiam\ {\isacharparenleft}{\isacharparenleft}{\isadigit{1}}{\isacharminus}{\isasymeta}{\isadigit{1}}{\isacharparenright}{\isacharasterisk}{\isacharparenleft}{\isadigit{1}}{\isacharminus}{\isasymeta}{\isadigit{2}}{\isacharparenright}{\isacharasterisk}{\isacharparenleft}{\isadigit{1}}{\isacharminus}{\isasymeta}{\isadigit{3}}{\isacharparenright}{\isacharasterisk}{\isacharparenleft}{\isadigit{1}}{\isacharminus}{\isasymeta}{\isadigit{4}}{\isacharparenright}{\isacharparenright}\ {\isacharparenleft}v{\isacharminus}wa{\isacharminus}wh{\isacharparenright}\ {\isacharparenleft}Alarm\ c{\isasymor}\ Hdet{\isacharparenright}{\isacharparenright}\ q{\isadigit{0}}{\isachardoublequoteclose}\isanewline
\ \ \ \ \isacommand{using}\isamarkupfalse%
\ lemma{\isadigit{3}}{\isacharunderscore}{\isadigit{6}}{\isacharbrackleft}of\ {\isachardoublequoteopen}{\isacharparenleft}{\isadigit{1}}{\isacharminus}{\isasymeta}{\isadigit{1}}{\isacharparenright}{\isacharasterisk}{\isacharparenleft}{\isadigit{1}}{\isacharminus}{\isasymeta}{\isadigit{2}}{\isacharparenright}{\isacharasterisk}{\isacharparenleft}{\isadigit{1}}{\isacharminus}{\isasymeta}{\isadigit{4}}{\isacharparenright}{\isachardoublequoteclose}\ {\isachardoublequoteopen}v{\isacharminus}wa{\isacharminus}{\isadigit{2}}{\isacharasterisk}wh{\isachardoublequoteclose}\ Alarm\ {\isachardoublequoteopen}{\isadigit{1}}{\isacharminus}{\isasymeta}{\isadigit{3}}{\isachardoublequoteclose}\ wh\ Hdet{\isacharbrackright}\ cimp{\isacharunderscore}def\ \isanewline
\ \ \ \ \isacommand{by}\isamarkupfalse%
\ {\isacharparenleft}smt\ lemma{\isadigit{3}}{\isacharunderscore}{\isadigit{5}}\ imptrans\ mult{\isachardot}assoc\ mult{\isachardot}commute{\isacharparenright}\isanewline
\ \ \isacommand{then}\isamarkupfalse%
\ \isacommand{have}\isamarkupfalse%
\ goal{\isadigit{3}}{\isacharcolon}\ {\isachardoublequoteopen}{\isacharquery}goal{\isadigit{3}}\ q{\isadigit{0}}{\isachardoublequoteclose}\ \isacommand{using}\isamarkupfalse%
\ constr{\isadigit{1}}{\isadigit{1}}\ lemma{\isadigit{3}}{\isacharunderscore}{\isadigit{1}}{\isadigit{1}}\ cimp{\isacharunderscore}def\ lemma{\isadigit{3}}{\isacharunderscore}{\isadigit{5}}\ \isacommand{by}\isamarkupfalse%
\ fastforce\isanewline
\isanewline
\ \ \isacommand{show}\isamarkupfalse%
\ {\isacharquery}thesis\ \isacommand{using}\isamarkupfalse%
\ goal{\isadigit{1}}\ goal{\isadigit{2}}\ goal{\isadigit{3}}\ cand{\isacharunderscore}def\ \isacommand{by}\isamarkupfalse%
\ auto\isanewline
\isacommand{qed}\isamarkupfalse%
\endisatagproof
{\isafoldproof}%
\isadelimproof
\isanewline
\endisadelimproof
\isanewline
\isacommand{end}\isamarkupfalse%
\isanewline
\isadelimtheory
\endisadelimtheory
\isatagtheory
\isacommand{end}\isamarkupfalse%
\endisatagtheory
{\isafoldtheory}%
\isadelimtheory
\endisadelimtheory
\end{isabellebody}%


\section{Mapping to molecules}
{
\scriptsize
\begin{verbatim}
(* File autogenerated by nuskell. 
- Translation Scheme: nuskell/schemes/literature/chen2013_2D_JF.ts
- Input CRN: 
L1 + U -> L2 + U
L2 + U -> L3 + U
L3 + U -> Y + U
L1 + H -> L1 + H
L2 + H -> L1 + H
L3 + H -> L1 + H
Y + H -> L1 + H
T1 + Y -> T2 + Y
T2 + Y -> T3 + Y
T3 + Y -> T4 + Y
T4 + Y -> D + Y
T1 + R -> T1 + R
T2 + R -> T1 + R
T3 + R -> T1 + R
T4 + R -> T1 + R
D + R -> T1 + R
A + D -> B + D
B + D -> C + D
C + D -> A + D
*)

def Fuel = 20
def Signal = 5

( Signal * < t0^  d1 > (* A *)
| Signal * < t2^  d3 > (* B *)
| Signal * < t4^  d5 > (* C *)
| Signal * < t6^  d7 > (* D *)
| Signal * < t8^  d9 > (* H *)
| Signal * < t10^  d11 > (* L1 *)
| Signal * < t12^  d13 > (* L2 *)
| Signal * < t14^  d15 > (* L3 *)
| Signal * < t16^  d17 > (* R *)
| Signal * < t18^  d19 > (* T1 *)
| Signal * < t20^  d21 > (* T2 *)
| Signal * < t22^  d23 > (* T3 *)
| Signal * < t24^  d25 > (* T4 *)
| Signal * < t26^  d27 > (* U *)
| Signal * < t28^  d29 > (* Y *)
| constant Fuel * < d57  t8^ > (* f0 *)
| constant Fuel * < d9  t10^ > (* f1 *)
| constant Fuel * < t58^  d59 > (* f10 *)
| constant Fuel * < t46^  d47 > (* f100 *)
| constant Fuel * [ d49 ]:[ t8^  d9 ]:[ t10^  d11 ]:[ t46^  d47 ]:{ t48^* } (* f101 *)
| constant Fuel * < d53  t8^ > (* f102 *)
| constant Fuel * < d11  t50^ > (* f103 *)
| constant Fuel * { t14^* }[ d15  t8^ ]:[ d9  t50^ ]:[ d51  t52^ ] (* f104 *)
| constant Fuel * < t50^  d51 > (* f105 *)
| constant Fuel * [ d53 ]:[ t8^  d9 ]:[ t10^  d11 ]:[ t50^  d51 ]:{ t52^* } (* f106 *)
| constant Fuel * [ d61 ]:[ t28^  d29 ]:[ t20^  d21 ]:[ t58^  d59 ]:{ t60^* } (* f11 *)
| constant Fuel * < d65  t28^ > (* f12 *)
| constant Fuel * < d29  t22^ > (* f13 *)
| constant Fuel * < d23  t62^ > (* f14 *)
| constant Fuel * { t20^* }[ d21  t28^ ]:[ d29  t62^ ]:[ d63  t64^ ] (* f15 *)
| constant Fuel * < t62^  d63 > (* f16 *)
| constant Fuel * [ d65 ]:[ t28^  d29 ]:[ t22^  d23 ]:[ t62^  d63 ]:{ t64^* } (* f17 *)
| constant Fuel * < d69  t28^ > (* f18 *)
| constant Fuel * < d29  t24^ > (* f19 *)
| constant Fuel * < d11  t54^ > (* f2 *)
| constant Fuel * < d25  t66^ > (* f20 *)
| constant Fuel * { t22^* }[ d23  t28^ ]:[ d29  t66^ ]:[ d67  t68^ ] (* f21 *)
| constant Fuel * < t66^  d67 > (* f22 *)
| constant Fuel * [ d69 ]:[ t28^  d29 ]:[ t24^  d25 ]:[ t66^  d67 ]:{ t68^* } (* f23 *)
| constant Fuel * < d73  t28^ > (* f24 *)
| constant Fuel * < d29  t6^ > (* f25 *)
| constant Fuel * < d7  t70^ > (* f26 *)
| constant Fuel * { t24^* }[ d25  t28^ ]:[ d29  t70^ ]:[ d71  t72^ ] (* f27 *)
| constant Fuel * < t70^  d71 > (* f28 *)
| constant Fuel * [ d73 ]:[ t28^  d29 ]:[ t6^  d7 ]:[ t70^  d71 ]:{ t72^* } (* f29 *)
| constant Fuel * { t28^* }[ d29  t8^ ]:[ d9  t54^ ]:[ d55  t56^ ] (* f3 *)
| constant Fuel * < d77  t16^ > (* f30 *)
| constant Fuel * < d19  t74^ > (* f31 *)
| constant Fuel * { t18^* }[ d19  t16^ ]:[ d17  t74^ ]:[ d75  t76^ ] (* f32 *)
| constant Fuel * < t74^  d75 > (* f33 *)
| constant Fuel * [ d77 ]:[ t16^  d17 ]:[ t18^  d19 ]:[ t74^  d75 ]:{ t76^* } (* f34 *)
| constant Fuel * < d81  t16^ > (* f35 *)
| constant Fuel * < d19  t78^ > (* f36 *)
| constant Fuel * { t20^* }[ d21  t16^ ]:[ d17  t78^ ]:[ d79  t80^ ] (* f37 *)
| constant Fuel * < t78^  d79 > (* f38 *)
| constant Fuel * [ d81 ]:[ t16^  d17 ]:[ t18^  d19 ]:[ t78^  d79 ]:{ t80^* } (* f39 *)
| constant Fuel * < t54^  d55 > (* f4 *)
| constant Fuel * < d85  t16^ > (* f40 *)
| constant Fuel * < d19  t82^ > (* f41 *)
| constant Fuel * { t22^* }[ d23  t16^ ]:[ d17  t82^ ]:[ d83  t84^ ] (* f42 *)
| constant Fuel * < t82^  d83 > (* f43 *)
| constant Fuel * [ d85 ]:[ t16^  d17 ]:[ t18^  d19 ]:[ t82^  d83 ]:{ t84^* } (* f44 *)
| constant Fuel * < d89  t16^ > (* f45 *)
| constant Fuel * < d19  t86^ > (* f46 *)
| constant Fuel * { t24^* }[ d25  t16^ ]:[ d17  t86^ ]:[ d87  t88^ ] (* f47 *)
| constant Fuel * < t86^  d87 > (* f48 *)
| constant Fuel * [ d89 ]:[ t16^  d17 ]:[ t18^  d19 ]:[ t86^  d87 ]:{ t88^* } (* f49 *)
| constant Fuel * [ d57 ]:[ t8^  d9 ]:[ t10^  d11 ]:[ t54^  d55 ]:{ t56^* } (* f5 *)
| constant Fuel * < d93  t16^ > (* f50 *)
| constant Fuel * < d17  t18^ > (* f51 *)
| constant Fuel * < d19  t90^ > (* f52 *)
| constant Fuel * < d33  t26^ > (* f53 *)
| constant Fuel * { t6^* }[ d7  t16^ ]:[ d17  t90^ ]:[ d91  t92^ ] (* f54 *)
| constant Fuel * < t90^  d91 > (* f55 *)
| constant Fuel * [ d93 ]:[ t16^  d17 ]:[ t18^  d19 ]:[ t90^  d91 ]:{ t92^* } (* f56 *)
| constant Fuel * < d27  t12^ > (* f57 *)
| constant Fuel * < d97  t6^ > (* f58 *)
| constant Fuel * < d7  t2^ > (* f59 *)
| constant Fuel * < d61  t28^ > (* f6 *)
| constant Fuel * < d3  t94^ > (* f60 *)
| constant Fuel * { t0^* }[ d1  t6^ ]:[ d7  t94^ ]:[ d95  t96^ ] (* f61 *)
| constant Fuel * < t94^  d95 > (* f62 *)
| constant Fuel * [ d97 ]:[ t6^  d7 ]:[ t2^  d3 ]:[ t94^  d95 ]:{ t96^* } (* f63 *)
| constant Fuel * < d13  t30^ > (* f64 *)
| constant Fuel * < d101  t6^ > (* f65 *)
| constant Fuel * < d7  t4^ > (* f66 *)
| constant Fuel * < d5  t98^ > (* f67 *)
| constant Fuel * { t2^* }[ d3  t6^ ]:[ d7  t98^ ]:[ d99  t100^ ] (* f68 *)
| constant Fuel * < t98^  d99 > (* f69 *)
| constant Fuel * < d29  t20^ > (* f7 *)
| constant Fuel * [ d101 ]:[ t6^  d7 ]:[ t4^  d5 ]:[ t98^  d99 ]:{ t100^* } (* f70 *)
| constant Fuel * { t10^* }[ d11  t26^ ]:[ d27  t30^ ]:[ d31  t32^ ] (* f71 *)
| constant Fuel * < d105  t6^ > (* f72 *)
| constant Fuel * < d7  t0^ > (* f73 *)
| constant Fuel * < d1  t102^ > (* f74 *)
| constant Fuel * { t4^* }[ d5  t6^ ]:[ d7  t102^ ]:[ d103  t104^ ] (* f75 *)
| constant Fuel * < t30^  d31 > (* f76 *)
| constant Fuel * < t102^  d103 > (* f77 *)
| constant Fuel * [ d105 ]:[ t6^  d7 ]:[ t0^  d1 ]:[ t102^  d103 ]:{ t104^* } (* f78 *)
| constant Fuel * [ d33 ]:[ t26^  d27 ]:[ t12^  d13 ]:[ t30^  d31 ]:{ t32^* } (* f79 *)
| constant Fuel * < d21  t58^ > (* f8 *)
| constant Fuel * < d37  t26^ > (* f80 *)
| constant Fuel * < d27  t14^ > (* f81 *)
| constant Fuel * < d15  t34^ > (* f82 *)
| constant Fuel * { t12^* }[ d13  t26^ ]:[ d27  t34^ ]:[ d35  t36^ ] (* f83 *)
| constant Fuel * < t34^  d35 > (* f84 *)
| constant Fuel * [ d37 ]:[ t26^  d27 ]:[ t14^  d15 ]:[ t34^  d35 ]:{ t36^* } (* f85 *)
| constant Fuel * < d41  t26^ > (* f86 *)
| constant Fuel * < d27  t28^ > (* f87 *)
| constant Fuel * < d29  t38^ > (* f88 *)
| constant Fuel * { t14^* }[ d15  t26^ ]:[ d27  t38^ ]:[ d39  t40^ ] (* f89 *)
| constant Fuel * { t18^* }[ d19  t28^ ]:[ d29  t58^ ]:[ d59  t60^ ] (* f9 *)
| constant Fuel * < t38^  d39 > (* f90 *)
| constant Fuel * [ d41 ]:[ t26^  d27 ]:[ t28^  d29 ]:[ t38^  d39 ]:{ t40^* } (* f91 *)
| constant Fuel * < d45  t8^ > (* f92 *)
| constant Fuel * < d11  t42^ > (* f93 *)
| constant Fuel * { t10^* }[ d11  t8^ ]:[ d9  t42^ ]:[ d43  t44^ ] (* f94 *)
| constant Fuel * < t42^  d43 > (* f95 *)
| constant Fuel * [ d45 ]:[ t8^  d9 ]:[ t10^  d11 ]:[ t42^  d43 ]:{ t44^* } (* f96 *)
| constant Fuel * < d49  t8^ > (* f97 *)
| constant Fuel * < d11  t46^ > (* f98 *)
| constant Fuel * { t12^* }[ d13  t8^ ]:[ d9  t46^ ]:[ d47  t48^ ] (* f99 *)
)
\end{verbatim}

\begin{verbatim}
(* File autogenerated by nuskell. 
- Translation Scheme: nuskell/schemes/literature/chen2013_2D_JF.ts
- Input CRN: 
L1 + U -> L2 + U
L2 + U -> L3 + U
L3 + U -> Y + U
L1 + H -> L1 + H
L2 + H -> L1 + H
L3 + H -> L1 + H
Y + H -> L1 + H
T1 + Y -> T2 + Y
T2 + Y -> T3 + Y
T3 + Y -> T4 + Y
T4 + Y -> D + Y
T1 + R -> T1 + R
T2 + R -> T1 + R
T3 + R -> T1 + R
T4 + R -> T1 + R
D + R -> T1 + R
A + B -> B + B + H
B + C -> C + C
C + A -> A + A
A + D -> B + D
B + D -> C + D
C + D -> A + D
H -> W
*)

def Fuel = 20
def Signal = 5

( Signal * < t0^  d1 > (* A *)
| Signal * < t2^  d3 > (* B *)
| Signal * < t4^  d5 > (* C *)
| Signal * < t6^  d7 > (* D *)
| Signal * < t8^  d9 > (* H *)
| Signal * < t10^  d11 > (* L1 *)
| Signal * < t12^  d13 > (* L2 *)
| Signal * < t14^  d15 > (* L3 *)
| Signal * < t16^  d17 > (* R *)
| Signal * < t18^  d19 > (* T1 *)
| Signal * < t20^  d21 > (* T2 *)
| Signal * < t22^  d23 > (* T3 *)
| Signal * < t24^  d25 > (* T4 *)
| Signal * < t26^  d27 > (* U *)
| Signal * < t28^  d29 > (* W *)
| Signal * < t30^  d31 > (* Y *)
| constant Fuel * < d59  t8^ > (* f0 *)
| constant Fuel * < d9  t10^ > (* f1 *)
| constant Fuel * < t60^  d61 > (* f10 *)
| constant Fuel * < d29  t120^ > (* f100 *)
| constant Fuel * { t8^* }[ d9  t120^ ]:[ d121  t122^ ] (* f101 *)
| constant Fuel * < t120^  d121 > (* f102 *)
| constant Fuel * [ d123 ]:[ t28^  d29 ]:[ t120^  d121 ]:{ t122^* } (* f103 *)
| constant Fuel * < d39  t26^ > (* f104 *)
| constant Fuel * < d27  t14^ > (* f105 *)
| constant Fuel * < d15  t36^ > (* f106 *)
| constant Fuel * { t12^* }[ d13  t26^ ]:[ d27  t36^ ]:[ d37  t38^ ] (* f107 *)
| constant Fuel * < t36^  d37 > (* f108 *)
| constant Fuel * [ d39 ]:[ t26^  d27 ]:[ t14^  d15 ]:[ t36^  d37 ]:{ t38^* } (* f109 *)
| constant Fuel * [ d63 ]:[ t30^  d31 ]:[ t20^  d21 ]:[ t60^  d61 ]:{ t62^* } (* f11 *)
| constant Fuel * < d43  t26^ > (* f110 *)
| constant Fuel * < d27  t30^ > (* f111 *)
| constant Fuel * < d31  t40^ > (* f112 *)
| constant Fuel * { t14^* }[ d15  t26^ ]:[ d27  t40^ ]:[ d41  t42^ ] (* f113 *)
| constant Fuel * < t40^  d41 > (* f114 *)
| constant Fuel * [ d43 ]:[ t26^  d27 ]:[ t30^  d31 ]:[ t40^  d41 ]:{ t42^* } (* f115 *)
| constant Fuel * < d47  t8^ > (* f116 *)
| constant Fuel * < d11  t44^ > (* f117 *)
| constant Fuel * { t10^* }[ d11  t8^ ]:[ d9  t44^ ]:[ d45  t46^ ] (* f118 *)
| constant Fuel * < t44^  d45 > (* f119 *)
| constant Fuel * < d67  t30^ > (* f12 *)
| constant Fuel * [ d47 ]:[ t8^  d9 ]:[ t10^  d11 ]:[ t44^  d45 ]:{ t46^* } (* f120 *)
| constant Fuel * < d51  t8^ > (* f121 *)
| constant Fuel * < d11  t48^ > (* f122 *)
| constant Fuel * { t12^* }[ d13  t8^ ]:[ d9  t48^ ]:[ d49  t50^ ] (* f123 *)
| constant Fuel * < t48^  d49 > (* f124 *)
| constant Fuel * [ d51 ]:[ t8^  d9 ]:[ t10^  d11 ]:[ t48^  d49 ]:{ t50^* } (* f125 *)
| constant Fuel * < d55  t8^ > (* f126 *)
| constant Fuel * < d11  t52^ > (* f127 *)
| constant Fuel * { t14^* }[ d15  t8^ ]:[ d9  t52^ ]:[ d53  t54^ ] (* f128 *)
| constant Fuel * < t52^  d53 > (* f129 *)
| constant Fuel * < d31  t22^ > (* f13 *)
| constant Fuel * [ d55 ]:[ t8^  d9 ]:[ t10^  d11 ]:[ t52^  d53 ]:{ t54^* } (* f130 *)
| constant Fuel * < d23  t64^ > (* f14 *)
| constant Fuel * { t20^* }[ d21  t30^ ]:[ d31  t64^ ]:[ d65  t66^ ] (* f15 *)
| constant Fuel * < t64^  d65 > (* f16 *)
| constant Fuel * [ d67 ]:[ t30^  d31 ]:[ t22^  d23 ]:[ t64^  d65 ]:{ t66^* } (* f17 *)
| constant Fuel * < d71  t30^ > (* f18 *)
| constant Fuel * < d31  t24^ > (* f19 *)
| constant Fuel * < d11  t56^ > (* f2 *)
| constant Fuel * < d25  t68^ > (* f20 *)
| constant Fuel * { t22^* }[ d23  t30^ ]:[ d31  t68^ ]:[ d69  t70^ ] (* f21 *)
| constant Fuel * < t68^  d69 > (* f22 *)
| constant Fuel * [ d71 ]:[ t30^  d31 ]:[ t24^  d25 ]:[ t68^  d69 ]:{ t70^* } (* f23 *)
| constant Fuel * < d75  t30^ > (* f24 *)
| constant Fuel * < d31  t6^ > (* f25 *)
| constant Fuel * < d7  t72^ > (* f26 *)
| constant Fuel * { t24^* }[ d25  t30^ ]:[ d31  t72^ ]:[ d73  t74^ ] (* f27 *)
| constant Fuel * < t72^  d73 > (* f28 *)
| constant Fuel * [ d75 ]:[ t30^  d31 ]:[ t6^  d7 ]:[ t72^  d73 ]:{ t74^* } (* f29 *)
| constant Fuel * { t30^* }[ d31  t8^ ]:[ d9  t56^ ]:[ d57  t58^ ] (* f3 *)
| constant Fuel * < d79  t16^ > (* f30 *)
| constant Fuel * < d19  t76^ > (* f31 *)
| constant Fuel * { t18^* }[ d19  t16^ ]:[ d17  t76^ ]:[ d77  t78^ ] (* f32 *)
| constant Fuel * < t76^  d77 > (* f33 *)
| constant Fuel * [ d79 ]:[ t16^  d17 ]:[ t18^  d19 ]:[ t76^  d77 ]:{ t78^* } (* f34 *)
| constant Fuel * < d83  t16^ > (* f35 *)
| constant Fuel * < d19  t80^ > (* f36 *)
| constant Fuel * { t20^* }[ d21  t16^ ]:[ d17  t80^ ]:[ d81  t82^ ] (* f37 *)
| constant Fuel * < t80^  d81 > (* f38 *)
| constant Fuel * [ d83 ]:[ t16^  d17 ]:[ t18^  d19 ]:[ t80^  d81 ]:{ t82^* } (* f39 *)
| constant Fuel * < t56^  d57 > (* f4 *)
| constant Fuel * < d87  t16^ > (* f40 *)
| constant Fuel * < d19  t84^ > (* f41 *)
| constant Fuel * { t22^* }[ d23  t16^ ]:[ d17  t84^ ]:[ d85  t86^ ] (* f42 *)
| constant Fuel * < t84^  d85 > (* f43 *)
| constant Fuel * [ d87 ]:[ t16^  d17 ]:[ t18^  d19 ]:[ t84^  d85 ]:{ t86^* } (* f44 *)
| constant Fuel * < d91  t16^ > (* f45 *)
| constant Fuel * < d19  t88^ > (* f46 *)
| constant Fuel * { t24^* }[ d25  t16^ ]:[ d17  t88^ ]:[ d89  t90^ ] (* f47 *)
| constant Fuel * < t88^  d89 > (* f48 *)
| constant Fuel * [ d91 ]:[ t16^  d17 ]:[ t18^  d19 ]:[ t88^  d89 ]:{ t90^* } (* f49 *)
| constant Fuel * [ d59 ]:[ t8^  d9 ]:[ t10^  d11 ]:[ t56^  d57 ]:{ t58^* } (* f5 *)
| constant Fuel * < d95  t16^ > (* f50 *)
| constant Fuel * < d17  t18^ > (* f51 *)
| constant Fuel * < d19  t92^ > (* f52 *)
| constant Fuel * { t6^* }[ d7  t16^ ]:[ d17  t92^ ]:[ d93  t94^ ] (* f53 *)
| constant Fuel * < t92^  d93 > (* f54 *)
| constant Fuel * [ d95 ]:[ t16^  d17 ]:[ t18^  d19 ]:[ t92^  d93 ]:{ t94^* } (* f55 *)
| constant Fuel * < d35  t26^ > (* f56 *)
| constant Fuel * < d99  t8^ > (* f57 *)
| constant Fuel * < d9  t2^ > (* f58 *)
| constant Fuel * < d3  t2^ > (* f59 *)
| constant Fuel * < d63  t30^ > (* f6 *)
| constant Fuel * < d3  t96^ > (* f60 *)
| constant Fuel * { t0^* }[ d1  t2^ ]:[ d3  t96^ ]:[ d97  t98^ ] (* f61 *)
| constant Fuel * < t96^  d97 > (* f62 *)
| constant Fuel * [ d99 ]:[ t8^  d9 ]:[ t2^  d3 ]:[ t2^  d3 ]:[ t96^  d97 ]:{ t98^* } (* f63 *)
| constant Fuel * < d27  t12^ > (* f64 *)
| constant Fuel * < d103  t4^ > (* f65 *)
| constant Fuel * < d5  t4^ > (* f66 *)
| constant Fuel * < d5  t100^ > (* f67 *)
| constant Fuel * < d13  t32^ > (* f68 *)
| constant Fuel * { t2^* }[ d3  t4^ ]:[ d5  t100^ ]:[ d101  t102^ ] (* f69 *)
| constant Fuel * < d31  t20^ > (* f7 *)
| constant Fuel * < t100^  d101 > (* f70 *)
| constant Fuel * [ d103 ]:[ t4^  d5 ]:[ t4^  d5 ]:[ t100^  d101 ]:{ t102^* } (* f71 *)
| constant Fuel * { t10^* }[ d11  t26^ ]:[ d27  t32^ ]:[ d33  t34^ ] (* f72 *)
| constant Fuel * < d107  t0^ > (* f73 *)
| constant Fuel * < d1  t0^ > (* f74 *)
| constant Fuel * < d1  t104^ > (* f75 *)
| constant Fuel * { t4^* }[ d5  t0^ ]:[ d1  t104^ ]:[ d105  t106^ ] (* f76 *)
| constant Fuel * < t104^  d105 > (* f77 *)
| constant Fuel * [ d107 ]:[ t0^  d1 ]:[ t0^  d1 ]:[ t104^  d105 ]:{ t106^* } (* f78 *)
| constant Fuel * < t32^  d33 > (* f79 *)
| constant Fuel * < d21  t60^ > (* f8 *)
| constant Fuel * < d111  t6^ > (* f80 *)
| constant Fuel * < d7  t2^ > (* f81 *)
| constant Fuel * < d3  t108^ > (* f82 *)
| constant Fuel * { t0^* }[ d1  t6^ ]:[ d7  t108^ ]:[ d109  t110^ ] (* f83 *)
| constant Fuel * < t108^  d109 > (* f84 *)
| constant Fuel * [ d111 ]:[ t6^  d7 ]:[ t2^  d3 ]:[ t108^  d109 ]:{ t110^* } (* f85 *)
| constant Fuel * [ d35 ]:[ t26^  d27 ]:[ t12^  d13 ]:[ t32^  d33 ]:{ t34^* } (* f86 *)
| constant Fuel * < d115  t6^ > (* f87 *)
| constant Fuel * < d7  t4^ > (* f88 *)
| constant Fuel * < d5  t112^ > (* f89 *)
| constant Fuel * { t18^* }[ d19  t30^ ]:[ d31  t60^ ]:[ d61  t62^ ] (* f9 *)
| constant Fuel * { t2^* }[ d3  t6^ ]:[ d7  t112^ ]:[ d113  t114^ ] (* f90 *)
| constant Fuel * < t112^  d113 > (* f91 *)
| constant Fuel * [ d115 ]:[ t6^  d7 ]:[ t4^  d5 ]:[ t112^  d113 ]:{ t114^* } (* f92 *)
| constant Fuel * < d119  t6^ > (* f93 *)
| constant Fuel * < d7  t0^ > (* f94 *)
| constant Fuel * < d1  t116^ > (* f95 *)
| constant Fuel * { t4^* }[ d5  t6^ ]:[ d7  t116^ ]:[ d117  t118^ ] (* f96 *)
| constant Fuel * < t116^  d117 > (* f97 *)
| constant Fuel * [ d119 ]:[ t6^  d7 ]:[ t0^  d1 ]:[ t116^  d117 ]:{ t118^* } (* f98 *)
| constant Fuel * < d123  t28^ > (* f99 *)
)
\end{verbatim}
}

\end{document}